\documentclass[a4paper,fleqn]{cas-sc}

\usepackage[authoryear]{natbib}
\usepackage{subfigure}
\usepackage{amsthm}
\usepackage[section]{placeins} 
\usepackage{physics}
\usepackage[dvipsnames]{xcolor}
\usepackage{makecell}
\usepackage{tabularray}
\UseTblrLibrary{diagbox}
\usepackage{multicol}
\usepackage{multirow}
\usepackage{utfsym}
\newtheorem{theorem}{Theorem}
\newtheorem{lemma}[]{Lemma}

\newtheorem{definition}[]{Definition}

\newtheorem{remark}[]{Remark}

\hypersetup{
  colorlinks,
  citecolor=blue,
  linkcolor=blue,
  urlcolor=blue}
\begin{document}
\let\WriteBookmarks\relax
\def\floatpagepagefraction{1}
\def\textpagefraction{.001}

% Short title
\shorttitle{Mitigating Stop-and-Go Traffic Congestion with Operator Learning}

% Short author
\shortauthors{Yihuai Zhang, Ruiguo Zhong, Huan Yu}

% Main title of the paper
\title {\Large{\textbf{Mitigating Stop-and-Go Traffic Congestion with Operator Learning}}}                      
% Title footnote mark
% eg: \tnotemark[1]
% \tnotemark[1,2]

% Title footnote 1.
% eg: \tnotetext[1]{Title footnote text}
% \tnotetext[<tnote number>]{<tnote text>} 
% \tnotetext[1]{This document is the results of the research
%   project funded by the National Science Foundation.}

% \tnotetext[2]{The second title footnote which is a longer text matter
%   to fill through the whole text width and overflow into
%   another line in the footnotes area of the first page.}

% First author
%
% Options: Use if required
% eg: \author[1,3]{Author Name}[type=editor,
%       style=chinese,
%       auid=000,
%       bioid=1,
%       prefix=Sir,
%       orcid=0000-0000-0000-0000,
%       facebook=<facebook id>,
%       twitter=<twitter id>,
%       linkedin=<linkedin id>,
%       gplus=<gplus id>]
\author[1]{Yihuai Zhang}[]
% \ead{yzhang169@connect.hkust-gz.edu.cn}

% Email id of the first author
% \ead{yzhang169@connect.hkust-gz.edu.cn}

% Second author
\author[1]{Ruiguo Zhong}[]

\author[1,2]{Huan Yu}[]

% Email id of the second author
% \ead{huanyu@ust.hk}

\cormark[1]

% Address/affiliation
\affiliation[1]{organization={Thrust of Intelligent Transportation, The Hong Kong University of Science and Technology (Guangzhou)},
    addressline={Nansha}, 
    city={Guangzhou},
    % citysep={}, % Uncomment if no comma needed between city and postcode
    postcode={511458}, 
    %state={Trivandrum},
    country={China}}
    
\affiliation[2]{organization={Department of Civil and Environmental Engineering, The Hong Kong University of Science and Technology},
    % addressline={Nansha}, 
    city={Hong Kong SAR},
    % citysep={}, % Uncomment if no comma needed between city and postcode
    % postcode={511458}, 
    % state={Guangdong},
    country={China}}

% Address/affiliation
% \affiliation[3]{organization={Your origination 3},
%     % addressline={}, 
%     city={Hong Kong SAR},
%     % citysep={}, % Uncomment if no comma needed between city and postcode
%     % postcode={695014}, 
%     % state={Trivandrum},
%     country={China}}

% Corresponding author text
\cortext[cor1]{Corresponding author: Huan Yu(huanyu@ust.hk)}

% % Footnote text
% \fntext[fn1]{This is the first author footnote. but is common to third
%   author as well.}
% \fntext[fn2]{Another author footnote, this is a very long footnote and
%   it should be a really long footnote. But this footnote is not yet
%   sufficiently long enough to make two lines of footnote text.}

% % For a title note without a number/mark
% \nonumnote{This note has no numbers. In this work we demonstrate $a_b$
%   the formation Y\_1 of a new type of polariton on the interface
%   between a cuprous oxide slab and a polystyrene micro-sphere placed
%   on the slab.
%   }

% Here goes the abstract
\begin{abstract}
This paper presents a novel neural operator learning framework for designing boundary control to mitigate stop-and-go congestion on freeways. The freeway traffic dynamics are described by second-order coupled hyperbolic partial differential equations (PDEs), i.e. the Aw-Rascle-Zhang (ARZ) macroscopic traffic flow model. The proposed framework learns feedback boundary control strategies from the closed-loop PDE solution using backstepping controllers, which are widely employed for boundary stabilization of PDE systems. The PDE backstepping control design is time-consuming and requires intensive depth of expertise, since it involves constructing and solving backstepping control kernels. Existing machine learning methods for solving PDE control problems, such as physics-informed neural networks (PINNs) and reinforcement learning (RL), face the challenge of retraining when PDE system parameters and initial conditions change. To address these challenges, we present neural operator (NO) learning schemes for the ARZ traffic system that not only ensure closed-loop stability robust to parameter and initial condition variations but also accelerate boundary controller computation. The first scheme embeds NO-approximated control gain kernels within a analytical state feedback backstepping controller, while the second one directly learns a boundary control law from functional mapping between model parameters to closed-loop PDE solution. The stability guarantee of the NO-approximated control laws is obtained using Lyapunov analysis. We further propose the physics-informed neural operator (PINO) to reduce the reliance on extensive training data. 
The performance of the NO schemes is evaluated by simulated and real traffic data, compared with the benchmark backstepping controller,  a Proportional Integral (PI) controller, and a PINN-based controller. The NO-approximated methods achieve a computational speedup of approximately 300 times with only a 1\% error trade-off compared to the backstepping controller, while outperforming the other two controllers in both accuracy and computational efficiency. The robustness of the NO schemes is validated using real traffic data, and tested across various initial traffic conditions and demand scenarios. The results show that neural operators can significantly expedite and simplify the process of obtaining controllers for traffic PDE systems with great potential application for traffic management. 
\end{abstract}

% Use if graphical abstract is present
% \begin{graphicalabstract}
% \includegraphics{figs/grabs.pdf}
% \end{graphicalabstract}

% Research highlights
%\begin{highlights}
%\item A operator learning framework for stabilizing stop-and-go traffic is proposed.
%\item The theoretical guarantee is proved for the operator learning framework through Lyapunov analysis.
%\item The physics-informed operator learning is %introduced to improve the robustness of the model.
%\item The proposed operator learning framework is examined in the numerical simulations. 
%\end{highlights}

% Keywords
% Each keyword is separated by \sep
\begin{keywords}
Freeway traffic control\sep Partial differential equations (PDEs) \sep Neural operators \sep Backstepping control
\end{keywords}

\maketitle

\section{Introduction}
Stop-and-go traffic oscillations are a common phenomenon on freeways, causing increased travel time, fuel consumption, and traffic accidents~\citep{belletti_prediction_2015,de_palma_traffic_2011,schonhof_empirical_2007,siri2021freeway,flynn2009self}. Freeway traffic control is focused on designing control strategies to mitigate stop-and-go traffic congestion, mainly implemented by road-based traffic management systems such as ramp metering or varying speed limits. The ramp metering controls the ramp inflow to the mainline, while varying speed limits regulate the speed of the mainline vehicle~\citep{hoogendoorn2016lessons,horowitz2005design,papamichail2010heuristic}. In recent decades, many studies have focused on vehicle-based control using connected automated vehicles~\citep{lee2024traffic,zhao2023safety,zheng2020smoothing,avedisov2020impacts}. Vehicle cruising speed control algorithms are developed to optimize car-following behaviors of a platoon of vehicles. Compared with emerging vehicle-based control methods that are developed based on individual vehicles, road-based traffic management and control utilize aggregated values such as traffic speed and flow to regulate the whole traffic on the road.
This paper will mainly focus on designing control strategies to suppress spatial-temporal traffic oscillations via ramp metering. We focus on using second-order macroscopic Partial Differential Equations (PDEs) models, due to its simplicity and analytic capability in modeling freeway traffic. An operator learning framework based on neural operators will be developed to achieve boundary stabilization of stop-and-go traffic on the freeway.

\subsection{Freeway traffic control}
Vehicular traffic dynamics on the highway are often described using macroscopic traffic models at an aggregated level. PDEs that are continuous in time and space are used to describe the spatial and temporal evolution of the macroscopic traffic variables, that is, density, speed, and flow rate. Macroscopic traffic models are categorized by the first-order density PDE model and second-order density and speed PDE model. The Lighthill and Whitham and Richard (LWR) model~\citep{lighthill1955kinematic,richards1956shock,whitham2011linear} is widely used for modeling the density evolution on the road section, but cannot capture stop-and-go oscillations on the freeway. The second-order Aw-Rascle-Zhang (ARZ) model~\citep{aw_resurrection_2000,zhang_non-equilibrium_2002} allows the traffc speed has its own dynamical acceleration equation describing speed evolution as a function of density, local speed and their gradients, and therefore is adopted to describe the stop-and-go oscillations. 

Control designs for freeway traffic stabilization using macroscopic PDE models can be classified into two categories: ``discretize then design'' and ``design then discretize'', based on whether numerical discretization of the PDE model is applied before or after the control design. The choice between the two approaches depends on various factors, including computational efficiency, accuracy requirements, and ease of implementation. 

{\bf Discretize then design} refers to the application of numerical discretization of traffic PDE model in time or space first and then designing control algorithms for the discretized models. Discretized traffic PDE models in time are Ordinary Differential Equations (ODEs) and the ones both in space and time are difference equations, which can reduce the difficulty of next-step control designs,  facilitate modular designs, and increase scalability for network problems.
The first-order Cell Transmission Model (CTM) is the discretized LWR model~\citep{daganzo1994cell}, in which the road section is divided into many ``cells'' and the propagation of traffic density in cells depends on the inflow and outflow of each cell, i.e., supply and demand. The second-order discrete METANET model is derived by discretizing and extending the Payne-Whitham model~\citep{kotsialos2002traffic,wang2022macroscopic}. Scaling up the CTM to the network level, the link-node cell transmission model was proposed by~\citep{muralidharan2009freeway} and the supply and demand relations still hold in the complex network road geometry. 
% Other extensions of CTM, such as the stochastic cell transmission model describing the stochastic demand and supply of the system~\citep{sumalee2011stochastic}, and the variable-length cell transmission model describing the position of the congestion wave front~\citep{canudas2018variable}, are also proposed to describe the traffic dynamics.

Based on the discretized models, various control strategies have been proposed for freeway traffic control. 
The classical feedback control ALINEA and PI-ALINEA have been proposed to resolve downstream bottlenecks using local ramp metering~\citep{papageorgiou1991alinea,wang2014local}. Additionally,~\citep{muller2015microsimulation} designed an integral controller for variable speed limits to prevent congestion formation at active bottlenecks.~\citep{carlson2011local} designed feedback mainstream traffic flow control on motorways using METANET that achieved the same performance of the optimal control approach. To solve the on-ramp metering control problem, the asymmetric cell transmission model (ACTM) was proposed to reduce the total time spent on a given road section~\citep{gomes2006optimal} and then extended to network traffic~\citep{muralidharan2012optimal}. Besides, model predictive control~\citep{liu2016model,bellemans2006model,muralidharan2015computationally}, distributed control~\citep{vcivcic2021coordinating,reilly2015distributed}, event-triggered control~\citep{ferrara2015event,ferrara2016design}, reinforcement learning~\citep{pan2021integrated,han2022physics} can also be applied for traffic control design.

Althrough adoption of discretized models in the modeling of macroscopic traffic makes the model simple and reduces the computational burden, the discretized cells inevitably generate errors in actual applications~\citep{mohan2013state}. In addition, CTM assumes that the density and speed in each cell is uniform and vehicles are assumed to have instantaneous acceleration and deceleration, which is also another unrealistic phenomenon~\citep{daganzo1994cell}. Furthermore, assuming that traffic is uniformly distributed within each cell violates the ``causality'' property. If the inflow to a cell has stopped or is declining, then the cell needs to redistribute the traffic uniformly towards backward. It is not consistent with the property that vehicles are influenced only by traffic ahead and not by traffic behind~\citep{carey2021cell}. Nonlinear traffic dynamics are not well captured by most ``discretize then design'' methods. 

{\bf Design then discretize} approaches avoid introducing the numerical approximation errors before control designs, therefore leading to more accurate control solutions. Control designs are directly proposed for the LWR PDE model or the ARZ PDE model that are continuous in time and space. In particular, the ``design then discretize" control approaches mainly includes
Lyapunov-based design~\citep{bastin_stability_2016}, backstepping design~\citep{krstic_boundary_2008}, optimal control design,~\citep{delle2017traffic,colombo2004minimising}, distributed control~\citep{bekiaris-liberis_pde-based_2021,qi_delay-compensated_2023}, and learning-based control approaches~\citep{belletti_expert_2018,yu_reinforcement_2022}.

The ``design then discretize'' approaches also offer greater adaptability for different control objectives of traffic  (i.e., traffic stabilization~\citep{yu_traffic_2019}, and traffic regulation~\citep{delle2017traffic}) or different traffic scenarios (i.e.,pure traffic~\citep{karafyllis_feedback_2019-1}, multi-class traffic~\citep{burkhardt_stop-and-go_2021}), as modifications can be directly applied to continuous modeling without affecting the discretization schemes. The preservation of continuity in time also offers flexibility for control design.
% Karafyllis designed explicit formulas for feedback laws to keep the system represented by the LWR model at the equilibrium point and achieved exponential stability~\citep{karafyllis_feedback_2019-1}.~\citep{zhang_stochastic_2017,zhang_pi_2019} developed a boundary feedback controller to stabilize both the stochastic traffic system and the deterministic traffic system. The system was described by the ARZ model, and the control inputs were implemented by ramp metering and variable speed limits. Considering the optimal control of traffic systems at the network level,~\citep{goatin2016speed} proposed an optimal control problem to minimize travel time in a traffic flow network represented by the LWR model with conservation laws. The optimal controller also works for second-order models~\citep{kolb2017capacity}.
% Backstepping control originally starts from the field of PDE boundary control~\citep{krstic_boundary_2008}. The first application of backstepping in boundary control of freeway traffic was proposed by~\citep{yu_traffic_2019} and then extended to multi-lane, and multi-class traffic PDE model~\citep{yu_traffic_2022}. The backstepping method is a good tool to suppress the traffic oscillations on road sections. 
% To reduce the system computational burden, Espitia~\citep{espitia_traffic_2022} developed the event-triggered control scheme for stabilizing the traffic system. However, the event-triggered control scheme is also highly computational load due to the backstepping design of the system. 
Among these PDE-based control methods, they are all dealing with PDEs whose computational burden and problem formulation is higher compared with ODEs. Even for the basic feedback controller in~\citep{zhang_pi_2019}, one needs to solve linear matrix inequalities (LMIs) to get the control gains for the traffic system. Solving LMIs would be time-consuming and need specific domain knowledge of it that also raises the threshold for using this method. Reinforcement learning can partially tackle this problem but it needs retraining for different parameters and traffic scenarios. However, for large-scale of traffic such as link-level traffic, the  ``design then discretize'' method may present implementation challenges. The dynamics of traffic would be more complex in cascaded traffic scenarios~\citep{yu_traffic_2022}, therefore, the control design and and solving analytical model problems would be more computationally intensive and technically demanding. In this paper, we design the boundary controller to mitigate traffic oscillations using the backstepping method and then we design the operator learning framework to reduce the computational burden for the traffic system. The comparison between different control methods and models is shown in Table~\ref{compare_table}.

\begin{table}[htbp]
    \centering
    \caption{Comparison between different methods and models}
    \begin{tabular}{c c c c c c}
    \hline
        \textbf{Reference} & \textbf{Model} & \textbf{Control methods} & \makecell{\textbf{High}\\ \textbf{efficiency}}  & \makecell{\textbf{Theoretical} \\ \textbf{stability} \\ \textbf{guarantee}} & \makecell{\textbf{Easy}\\ \textbf{implemen-}\\ \textbf{tation}} \\
    \hline
    \cite{papageorgiou1991alinea} & CTM & Integral control & \usym{2714} &  & \usym{2714}\\
    % \hline
    % \cite{wang2014local} & CTM & Feedback control (PI) &\usym{2714} & \usym{2718} & \usym{2714} \\
    % \hline
    \cite{gomes2006optimal} & ACTM & Optimal control &  &  & \usym{2714}\\
    % \hline
    \cite{ferrara2015event} & CTM & Event-triggered MPC & &  & \usym{2714}\\
    % \hline
    \cite{carlson2011local} & METANET & Integral control &\usym{2714} &  & \usym{2714} \\
    % \hline
    \cite{han2022physics} & METANET & Reinforcement learning & \usym{2714} &  &\usym{2714}\\
    % \hline
    \cite{belletti_expert_2018} & LWR & Reinforcement learning &\usym{2714} &  & \usym{2714}\\
    % \hline
    \cite{karafyllis_feedback_2019-1} & LWR & Feedback control  &  & \usym{2714} & \\
    % \hline
    \cite{delle2017traffic} & LWR & Optimal control &  &  & \usym{2714}\\
    % \hline
    \cite{zhang_pi_2019} & ARZ & \makecell{PI control}  & & \usym{2714} &  \\
    % \hline
    \cite{yu_traffic_2019} & ARZ & Backstepping & & \usym{2714} & \\
    % \hline
    This paper & ARZ & Backstepping $+$ NO & \usym{2714}  &\usym{2714} & \usym{2714}\\
    \hline
    \end{tabular}
    \label{compare_table}
\end{table}

\textbf {PDE backstepping control} has been widely studied for boundary stabilization of the hyperbolic PDE models~\citep{krstic_boundary_2008,krstic_backstepping_2008, vazquez_backstepping_2011,anfinsen_adaptive_2019,zhang2024robust}. ~\citep{krstic_backstepping_2008} first proposed the backstepping controller for stabilization of hyperbolic PDEs by simply actuating the boundary conditions, for example, flow input or speed at the boundaries of the ARZ model. Over the past decades, backstepping approaches have been extended for robust control design by~\citep{auriol_robust_2020, karafyllis_input--state_2019}, output disturbance rejection by~\citep{lamare_adding_2016}, and adaptive design  by~\citep{anfinsen_adaptive_2019} with respect to parameter uncertainty and disturbances. 

Motivated by the stop-and-go traffic modeled by the ARZ PDEs,~\citep{yu_traffic_2022} first applied backstepping method for congested freeway traffic control problems and then extended to multi-lane, and multi-class traffic PDE models. The control objective of the freeway traffic using backstepping method is to stabilize the traffic  states at their equilibrium points. Unlike the results in~\citep{delle2017traffic} whose objective is to regulate the outflow to a desired outflow, backstepping method aims to make the full states of density and speed stay at a spatially uniform value. The ARZ PDE system is transformed into an exponentially stable target system using the backstepping transformation along with the controller design. The control gains are obtained by solving the kernel equations of the transformation. The boundary control law then is constructed with the backstepping gain kernels and system states.
% Solving control gain kernels and calculating the control law are time-consuming and the system computational load is high. When dealing with more complex traffic systems, such as multi-class traffic~\citep{burkhardt_stop-and-go_2021}, cascaded freeway traffic~\citep{yu_simultaneous_2022}, multi-lane traffic~\citep{yu_output_2021}, deriving the explicit boundary controller not only requires more computation time for the traffic system, but also requires years of training in control theory to handle the intricacy involved in constructing the backstepping transformation for control engineers. This motivates the design of our neural operators in this paper.
\begin{figure}[!htbp]
    \centering
    \includegraphics[width= \linewidth]{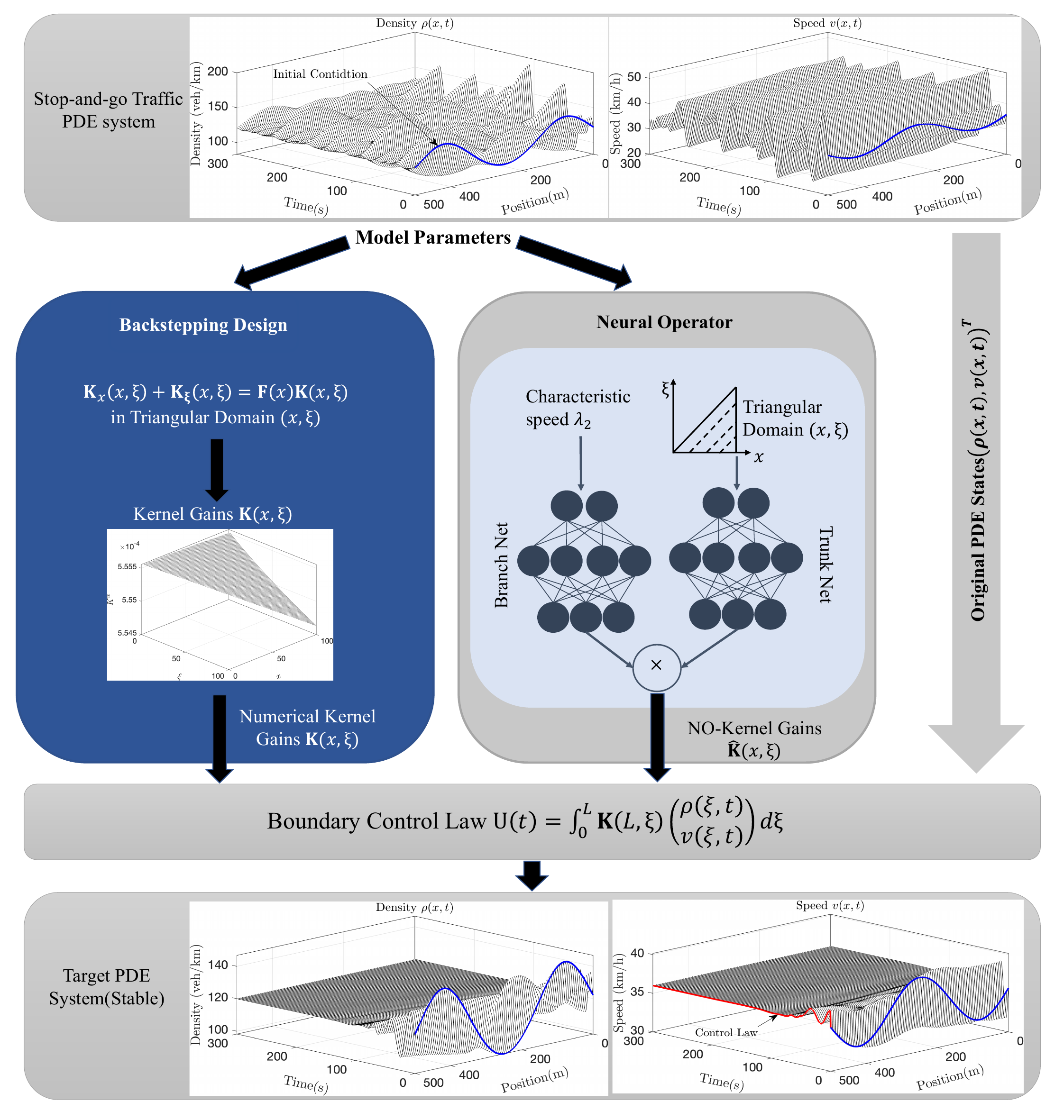}
    \caption{The diagram of the backstepping method and the proposed neural operator framework}
    \label{bsdia}
\end{figure}

\subsection{Neural operators for control of PDEs}
Over the recent decades, machine learning (ML) methods have emerged as powerful tools for solving complex algebraic equations and PDEs. Physics-informed neural network (PINN) has ability to solve forward problem of PDEs~\citep{lu_deepxde_2021}. 
PINN can be used for traffic state and fundamental diagram estimation, which can learn a functional form of the fundamental diagram and solve the first-order and second-order traffic flow models~\citep{shi_physics-informed_2022,zhang_physics-informed_2024,zhao2023observer}. However, PINN can not solve some multi-scale dynamical PDE systems because it is sensitive to hyper-parameters, and it only learns the solution of a single PDE system~\citep{sun2020surrogate}. It can not generalize to PDEs with different parameters, initial conditions, and boundary conditions.

Compared with PINN and other traditional ML methods, neural operators (NO) present exciting advances due to their ability to learn the operator mapping of functionals~\citep{kovachki_neural_2023}. The standard frameworks for NO are DeepONet~\citep{lu_learning_2021} and Fourier Neural Operator (FNO)~\citep{li_fourier_2021}. The performance of DeepONet and FNO is comparable for relatively simple settings, but the performance of FNO deteriorates greatly for complex geometries~\citep{lu2022comprehensive}. A significant advantage of DeepONet and related structures is the ability to freely discretize output functions~\citep{lin2021operator}. This flexibility allows the network to predict the values of the output functions at any given domain~\citep{zhang2023belnet}. To make neural operators more accurate and robust, the extended DeepONet and FNO were proposed~\citep{lu2022comprehensive}. With further consideration of the absence of training data and the generality of neural operators, Wang proposed physics-informed DeepONets to learn the solution of arbitrary PDEs~\citep{wang_learning_2021}. These neural operators offer distinctive advantages compared to other traditional ML methods due to their simplicity in solving complex problems such as climate forecasting~\citep{pathak2022fourcastnet}, multiphase flow~\citep{wen2022u}, heterogeneous material modeling~\citep{you2022learning}.
Therefore, it is of great value for solving PDEs and backstepping kernel equations using neural operators.
~\citep{bhan_operator_2023} adopted neural operators for nonlinear adaptive control. The operator learning framework was proposed to accelerate nonlinear adaptive control. In addition, they apply the operator learning method for bypassing gain and control computations in PDE control~\citep{bhan2023neural,bhan_operator_2023}. With the operator learning framework, there is no need to compute backstepping kernel gains by solving kernel equations numerically.

All of the aforementioned NO-based results demonstrated the capability of neural operators in accelerating computational speed and analytical derivation for PDE control design. However, to the author's knowledge,  PDE boundary control using neural operators has not been studied in traffic control problems before. The neural operator framework has great potential to derive methodological learning-based control for the freeway traffic stabilization problem, which will be extensively discussed in this paper. 

\subsection{Contributions}
{
To address the limitations of existing traffic control methods based on macroscopic PDE models in mitigating stop-and-go traffic oscillations, we identify key constraints that motivate our proposed operator learning framework and the main contributions.
}
%all the aforementioned methods based on design then discretize principle could remove the stop-and-go wave at certain level. However, all these methods especially backstepping method are suffering the computational burden due to the PDE representation of traffic dynamics and kernel equations. There are several reasons why we need to develop new methods to efficiently mitigate the stop-and-go traffic wave
{
\begin{itemize}
    \item "Discretize then Design" control methods are inherently limited by their discretization schemes, often resulting in discretization traffic states errors, loss of continuity between cells, and difficulties in handling nonlinearities of traffic dynamics.  
    Our proposed Neural Operator (NO) scheme is notably invariant to discretization. 
    \item Backstepping, as a representative method of ``design then discretize'', offers unique advantages in free traffic control—providing a closed-form solution with lower errors and enhanced robustness, and maintaining the continuity of traffic states, compared to "discretize then design" methods However, the backstepping method incurs a significant computational burden due to the complexity of solving the control gain kernels for PDEs. The PDE control design also demands a deep level of expertise, complicating its practical implementation. The proposed NO scheme is designed to accelerate the computational process and bypass the need for analytical design through data-driven training. 
    \item Although PINN can partially expedite the computation process for solving specific PDEs, the method is restricted to learning PDE solution of one instance, characterized by a single set of boundary and initial conditions. This limitation renders it ineffective in adapting to changes of traffic patterns and demands. Conversely, the proposed NO methods are designed to learn functional mappings from PDE model parameters to the control gains. These mappings remain invariant regardless of changes in traffic's initial and boundary conditions or system parameters, thereby providing a more flexible and robust solution for freeway traffic control.
\end{itemize}
}
{
To overcome the above problems, we proposed NO-based methods to approximate the operator mapping from congested wave speeds to control gain kernels. This approach significantly reduces computational burden and allows closed-loop results to generalize across varying traffic demands and patterns. Author's previous result on RL traffic control~\citep{yu_reinforcement_2022}   trained a reinforcement learning boundary controller using the ARZ PDE model as a simulator. The results are obtained for one-instance and there needs retraining given different traffic conditions. Different from the focus of this paper on traffic stabilization, another line of research has explored learning-based approaches for traffic state estimation ~\citep{shi_machine_2022} using FNO and using PINN~\citep{zhao2023observer}.  Authors investigated the PDE observer design using partial measurement information to infer the full spatial-temporal traffic states. 
}

{
Two mappings, the NO-approximated gain kernels and the NO-approximated control law, are developed to improve the computation speed of gain kernels and control law based our previous results in~\citep{zhang2024neural}. Different with our previous results, we further developed the physics-informed neural operator (PINO) framework to learn the mapping from system parameters to the backstepping kernels to deal with limited data scenarios. The operator learning framework based on the bakcstepping method is shown in Fig.~\ref{bsdia}. To the best of author's knowledge, this is the first result for the application of operator learning in traffic control, in particular for boundary control of the macroscopic ARZ PDE model. The theoretical contribution lies in proposing the novel Lyapunov analysis for neural operators and proving the stability of the NO-based closed-loop traffic system. Extensive simulation results demonstrate  that the NO-approximated methods significantly accelerate the computation of control laws while successfully stabilizing the traffic system under varying conditions. This approach presents an efficient, accurate, and robust solution for traffic control. 
}

The structure of the paper is as follows: Section 2 presents the design of a boundary controller for the ARZ traffic model utilizing the backstepping method. In Section 3, the neural operator is introduced to approximate the backstepping kernels, followed by a Lyapunov analysis of the NO-based kernels. Additionally, the NO-approximated backstepping control law is developed, and it is demonstrated that practical exponential stability of the system is achieved. The extension to PINO is also discussed to illustrate its effectiveness in the absence of input-output data. Section 4 details the experiments conducted on the neural operators for kernels and control law, along with a comparative analysis of the NO-based methods, the PI controller, the PINN-based controller, and the backstepping controller. Finally, Section 5 provides the conclusion of the paper.

\section{Boundary Control of the ARZ PDE Model}
The macroscopic traffic dynamics on a given road are described by the nonlinear coupled hyperbolic PDEs, i.e., the ARZ model. The model is defined by:
\begin{align}
    \partial_t \rho + \partial_x (\rho v) &= 0, \label{origin1}\\
    \partial_t (v - V(\rho)) + v \partial_x(v - V(\rho)) &= \frac{V(\rho) - v}{\tau}, \label{origin2}
\end{align}
where $\rho(x,t)$ is the traffic density, $v(x,t)$ denotes the traffic speed, defined in the spatial and time domain $(x,t) \in [0, L]\times [0, +\infty)$. The reaction time  $\tau$ denotes how long it takes for drivers' behavior adapting to equilibrium density-speed relation $V(\rho)$ . The fundamental diagram $V(\rho)$ describes the relation between the traffic density and speed. The fundamental diagram should guarantee the flow function $q(\rho) = \rho V(\rho)$ to be strictly concave(i.e., $q''(\rho)<0$). It can be selected as the Greenshield's model:
\begin{align}
    V(\rho) = v_f\left(1 - \left(\frac{\rho}{\rho_m}\right)^\gamma\right),
\end{align}
where $v_f$ is the maximum speed for the traffic flow, $\rho_m$ denotes the maximum density. It should be noted that the proposed control design does not restrict the choice of fundamental diagram as long as the flow-density relation $Q(\rho)$ is twice differentiable and concave. We define the equilibrium state of the system as $(\rho^\star, v^\star)$. We have established the relation between the equilibrium speed and density using the fundamental diagram 
\begin{align}
    v^\star = V(\rho^\star),
\end{align}
Considering the traffic conditions on the freeway, we set the inlet boundary $x=0$ as a constant traffic flow $q^\star = \rho^\star v^\star$, thus we get inlet boundary condition as:
\begin{align}
    \rho(0,t) = \frac{q^\star}{v(0,t)}. \label{bc_q}
\end{align}
At the outlet of the road section, we set the traffic density as $\rho^\star$ to obtain the following boundary condition for traffic speed
\begin{align}
    v(L,t) = \frac{q(L,t)}{\rho^\star} + U(t), \label{bc_v}
\end{align}
where $U(t)$ is the control input, actuating the flow of vehicles that leaving of the road section from the outlet. It can be implemented by ramp metering to regulate the traffic flow at the outlet of the mainline road. In this paper, we deal with the congested traffic condition, meaning that the traffic density $\rho(x,t)$ is larger than the critical density for the traffic system. The control objective is to regulate the traffic density and speed at its equilibrium points $(\rho^\star,v^\star)$ within finite time. We consider the congested traffic where stop-and-go traffic oscillations arise. The congested equilibrium density $\rho^\star$ is chosen such that $\rho^\star  > \rho_c$ where $\rho_c$ is the critical density that satisfies $ Q^{'}(\rho)|_{\rho = \rho_c} = 0$ and defines the congested and free traffic. 

The boundary control input is designed to mitigate the traffic oscillations,
\begin{align}
    U(t) = &- (\rho(L,t)v(L,t) - q^\star) + \left(\rho^\star + \frac{v^\star}{V'(\rho^\star)}\right)(v(L,t) - v^\star) + \int_0^L \mathrm{e}^{\frac{\xi}{\tau v^\star}}K^w(L,\xi) (q(\xi,t) - q^\star) d\xi \nonumber\\
    &- \int_0^L 
    \left(\frac{v^\star}{V'(\rho^\star)}K^v(L,\xi)+ \left(\rho^\star + \frac{v^\star}{V'(\rho^\star)}\right)\mathrm{e}^{\frac{\xi}{\tau v^\star}}K^w(L,\xi)\right)({v}(\xi,t) - v^\star) d\xi \label{control_original}
\end{align}
where $K^w(L,\xi)$, $K^v(L,\xi)$ are kernel gains of the controller. They are obtained from the backstepping control design. The computation of the control gains will be given in the next section. The following lemma is stated regarding the closed-loop system. 
{
\begin{lemma}[\cite{yu_traffic_2019}]
    The system \eqref{origin1}-\eqref{origin2} with boundary conditions \eqref{bc_q}-\eqref{bc_v} and initial conditions $\rho(x,0), v(x,0) \in L^2[0, L]$ under the control law \eqref{control_original} whose kernels are solved by~\eqref{ker1}-\eqref{ker4} is locally exponentially stable in $L_2$-sense at finite time $t_f = \frac{L}{v^\star} + \frac{L}{-\rho^\star V'(\rho^\star) - v^\star}$, and the traffic density $\rho(x,t)$ and speed $v(x,t)$ converge to their equilibrium points at finite time.
    \begin{align*}
        \norm{\rho(x,t) - \rho^\star} &\to 0,\\
        \norm{v(x,t) - v^\star} &\to 0.
    \end{align*}
\end{lemma}
}
\subsection{Backstepping controller design}
%There is no explicit analytical solution to the nonlinear hyperbolic PDE system . 
We firstly linearize the PDE system \eqref{origin1}-\eqref{origin2} around its equilibrium point $(\rho^\star, v^\star)$.  The small deviation of the equilibrium point are defined as
\begin{align}
    \Bar{\rho}(x,t) &= \rho(x,t) - \rho^\star, \\
    \Bar{v}(x,t) &= v(x,t) - v^\star.
\end{align}
We also have the small deviation of traffic flow $\Bar{q}(x,t) = q(x,t)- q^\star$. The original system is a nonlinear hyperbolic PDE where the two PDE states are coupled in domain. We transform the original model into the boundary control model for the follow-up backstepping design. Then we can define the following spatial transformation for the linearzied PDE states $(\tilde \rho, \tilde v)$,
\begin{align}
     \Tilde{w}(x,t) &= \mathrm{e}^{\frac{x}{\tau v^\star}}\left(\Bar{q}(x,t) -  \left(\rho^\star + \frac{v^\star}{V'(\rho^\star)} \right)\Bar{v}(x,t)\right), \Tilde{v}(x,t) = -\frac{v^\star}{V'(\rho^\star)} \Bar{v}(x,t)
    \label{riemann}
\end{align}
The system~\eqref{origin1}-\eqref{origin2} with boundary conditions~\eqref{bc_q}-\eqref{bc_v} are then transformed into the following linearized boundary control model. 
\begin{align}
    \partial_t \Tilde{w}(x,t) + \lambda_1 \partial_x \Tilde{w}(x,t) &= 0, \label{bs-q}\\
    \partial_t \Tilde{v}(x,t) - \lambda_2 \partial_x \Tilde{v}(x,t) &= c(x)\Tilde{w}(x,t),\\
    \Tilde{w}(0,t) &= - r \Tilde{v}(0,t),\\
    \Tilde{v}(L,t) &= \kappa\Tilde{w}(L,t)+ U(t) \label{bs-bc_l},
\end{align}
where the coefficients are $c(x) = - \frac{1}{\tau}\mathrm{e}^{-\frac{x}{\tau v^\star}}$, $r = \frac{-\rho^\star V'(\rho^\star) - v^\star}{v^\star}$, $\kappa = \mathrm{e}^{-\frac{L}{\tau v^\star}}$. 
The characteristic speed $\lambda_1,\lambda_2$ of the traffic PDE represent the the propagation direction of traffic waves. 
\begin{align}
    \lambda_1 =& v^\star,\\
    \lambda_2 =& - \rho^\star V'(\rho^\star) - v^\star.
\end{align}
The characteristic speeds of the free-flow and congested traffic are:  
a) \textbf{free-flow regime}:  $\lambda_1>0,  -\lambda_2>0$, the traffic waves transports from upstream to downstream; b)\textbf{ congested regime}: When the characteristic speed of traffic speed is negative, such as $\lambda_1>0$ and $-\lambda_2<0$, the traffic is in the congested regime. The speed oscillation will transport from downstream to upstream while the traffic density still transport from upstream to the downstream, making the traffic become congested. 

%The shock wave can be observed by the fundamental diagram. 

By applying the backstepping transformation, the plant system \eqref{bs-q}-\eqref{bs-bc_l} is converted into a target system such that in-domain unstable couplings are transformed to the boundary and then the oscillations are damped out by boundary actuation. The following backstepping transformation is introduced,
\begin{align}
    \alpha(x,t) &= \Tilde{w}(x,t),\label{back-1}\\
    \beta(x,t) &= \Tilde{v}(x,t) - \int_0^x K^w(x,\xi)\Tilde{w}(\xi,t)d\xi -  \int_0^x K^v(x,\xi)\Tilde{v}(\xi,t)d\xi, \label{back-2}
\end{align}
where $\alpha(x,t)$, $\beta(x,t)$ are the transformed PDE states of the target system and $K^w(x,\xi)$, $K^v(x,\xi)$ are the kernels  of the transformation. This transformation converts the system \eqref{bs-q}-\eqref{bs-bc_l} into an exponential stable target system combined with the kernel equations \eqref{ker1}-\eqref{ker4} and the backstepping control law \eqref{control_bs}. The kernel equations evolve in the triangular domain $\mathcal{T} = \{ (x,\xi): 0 \leq \xi \leq x <L \}$.
\begin{align}
    \lambda_2 K^w_x(x,\xi) - \lambda_1 K^w_\xi(x,\xi) &= c(x)K^v(x,\xi), \label{ker1}\\
    \lambda_2 K^v_x(x,\xi)  + \lambda_2 K^v_\xi(x,\xi) &= 0,\\
    K^w(x,x) &=-\frac{c(x)}{\lambda_1 + \lambda_2}, \\
    K^v(x,0) &= -K^w(x,0).\label{ker4}
\end{align}
Using the transformation and the kernel equations, we can get the target system as follows:
\begin{align}
    \partial_t \alpha(x,t) + \lambda_1 \partial_x\alpha(x,t) &= 0, \label{tar1}\\
    \partial_t \beta(x,t) - \lambda_2 \partial_x \beta(x,t) &= 0,\\
    \alpha(0,t) &= -r\beta(0,t),\\
    \beta(L,t) &= 0\label{tar2}.
\end{align}
with the backstepping control law chosen as follows, \begin{align}\label{control_bs}
    U(t) = - \kappa \Tilde{w}(L,t) + \int_0^L K^w(L,\xi)\Tilde{w}(\xi,t) d\xi + \int_0^L K^v(L,\xi)\Tilde{v}(\xi,t) d\xi,
\end{align}
 such that $\beta(L,t) = 0 $. The target system \eqref{tar1}-\eqref{tar2} is proved to beexponentially stable in $L_2$-sense following the results in \citep{yu_traffic_2019}. The target system is equivalent to the linearized system \eqref{bs-q}-\eqref{bs-bc_l} since the aforementioned transformation is invertiable. {Writing \eqref{control_bs} in the original coordinates $(\rho,v)$, we obtain the control input~\eqref{control_original} for the original system \eqref{origin1}-\eqref{origin2}. Thus, applying the control law to the system, and then the unstable traffic system could be stabilized.}

\section{Neural Operator Learning to Accelerate Computation of Backstepping Gain Kernels}
In this section, we propose the operator learning framework to accelerate the computation process of the backstepping method. The boundary control law is constructed using the backstepping gain kernels. Solving the kernel equations and calculating the boundary control law can be time-consuming. The neural operators developed in this section are used to reduce the computational burden. Three operator learning schemes. We will use neural operators to learn the mapping from the characteristic speed to the backstepping gain kernels, and the boundary control law. Finally we develop the physics-informed operator learning framework.
\subsection{Neural operator for approximating backstepping gain kernels}
%\subsection{Neural operator}
The neural operator is employed for approximating the operator mapping of functionals. In the section, we introduce the neural operator using DeepONet to approximate the mapping from the characteristic speed $\lambda_2$ to kernels $K^w(x,\xi)$ and $K^v(x,\xi)$. A neural operator (NO) for approximating a nonlinear mapping $\mathcal{G}: \mathcal{U} \mapsto \mathcal{V}$

\begin{align}
    \mathcal{G}_{\mathbb{N}}\left(\mathbf{u}_m\right)(y)=\sum_{k=1}^p 
    \underbrace{g^{\mathcal{N}}\left(\mathbf{u}_m ; \vartheta^{(k)}\right)}_{\mathrm{branch}} \underbrace{f^{\mathcal{N}}\left(y ; \theta^{(k)}\right)}_{\mathrm{trunk}},\label{neuralop}
\end{align}
where $\mathcal{U}$ and $\mathcal{V}$ are function spaces of continuous functions $u\in \mathcal{U}$, $v\in\mathcal{V}$. And $\mathbf{u}_m$ is the evaluation of function $u$ at points $x_i = x_1,\dots,x_m$. $p$ is the number of basis components in the target space, $y\in Y$ is the location of the output function $v(y)$ evaluations, and $g^{\mathcal{N}}$, $f^{\mathcal{N}}$ are NNs termed branch and trunk networks. $\vartheta^{(k)},\theta^{(k)}$ denote all trainable weights and bias parameters in the branch and trunk networks.
{
\begin{lemma}[DeepONet universal approximation theorem \citep{bhan2023neural,chen1995universal,deng2022approximation}]\label{Deeptheo}
    Let $X \subset \mathbb{R}^{d_x}$, $Y$ $\subset$ $\mathbb{R}^{d_y}$ be compact sets of vectors $x\in X$ and $y\in Y$. Let $\mathcal{U}$: $X \rightarrow \mathbb{U} \subset \mathbb{R}^{d_u}$ and $\mathcal{V}$: $Y \rightarrow \mathbb{V} \subset \mathbb{R}^{d_v}$ be sets of continuous functions $u(x)$ and $v(y)$, respectively. Assume the operator $\mathcal{G}$: $\mathcal{U} \rightarrow \mathcal{V}$ is continuous. Then, for all $\epsilon > 0$, there exists a $m^\star,p^\star \in \mathbb{N}$ such that for each $m \geq m^\star$, $p \geq p^\star$, there exist $\theta^{(k)}$, $\vartheta^{(k)}$, neural networks $f^{\mathcal{N}}\left(\cdot ; \theta^{(k)}\right), g^{\mathcal{N}}\left(\cdot ; \vartheta^{(k)}\right), k= 1,\dots,p$ and $x_j \in X, j=1, \ldots, m$, with corresponding $\mathbf{u}_m=\left(u\left(x_1\right), u\left(x_2\right), \cdots, u\left(x_m\right)\right)^{\top}$, such that 
    \begin{align}
        \sup_{\mathbf{u}\in \mathcal{U}}\sup_{y\in Y } \left|\mathcal{G}(\mathbf{u})(y)-\mathcal{G}_{\mathbb{N}}\left(\mathbf{u}_m\right)(y)\right|<\epsilon, 
    \end{align}
    for all functions $u \in \mathcal{U}$ and all values $y\in Y$ of $\mathcal{G}(\mathbf{u})(y) \in \mathcal{V}$.
\end{lemma}
}

\begin{definition}
    The kernel operator $\mathcal{K}$: $\mathbb{R}^{+} \rightarrow C^1(\mathcal{T}) \times C^1(\mathcal{T})$ is defined by:
%     {\setlength\abovedisplayskip{0.1cm}
% \setlength\belowdisplayskip{0.1cm}
    \begin{align}
        K^{w}(x,\xi) &:= \mathcal{K}^w(\lambda_2)(x,\xi),\\
        K^{v}(x,\xi) &:= \mathcal{K}^v(\lambda_2)(x,\xi).
    \end{align}
\end{definition}
The kernel operator $\mathcal{K}$ denotes the mapping from the characteristic speed to the backstepping transformation kernels. Based on Theorem \ref{Deeptheo}, we have the following lemma on the approximation of the neural operator for the kernel equations:
\begin{lemma}\label{NO-K}
    For all $\epsilon > 0$, there exists a neural operator $\mathcal{\Hat{K}}$ that for all $(x,\xi)\in \mathcal{T}$,
    \begin{align}\label{neuraloperator}
    \sup_{\lambda_2 \in \mathcal{U}} \norm{ \mathcal{K}(\lambda_2)(x,\xi)  - \mathcal{\Hat{K}}(\lambda_2)(x,\xi)} < \epsilon.
        % \left| \mathcal{K}(\lambda_2)  - \mathcal{\Hat{K}}(\lambda_2) \right| + \left|\partial_x\left( \mathcal{K}(\lambda_2) - \mathcal{\Hat{K}}(\lambda_2) \right) \right| + \left|\partial_{\xi}\left( \mathcal{K}(\lambda_2)  - \mathcal{\Hat{K}}(\lambda_2)\right) \right| < \epsilon.
    \end{align}
\end{lemma}
\begin{proof}
    The existence, uniqueness of the kernel equations have been proved in \citep{vazquez_backstepping_2011}. So the mapping $\mathcal{{K}} : \mathbb{R}^+ \rightarrow C^1(\mathcal{T}) \times C^1(\mathcal{T})$ from $\lambda_2$ to $K^w(x,\xi),K^v(x,\xi)$ indicated by \eqref{ker1}-\eqref{ker4} and the solution of the kernel equations exists. The neural operator $\mathcal{\hat{K}}$ approximates the backstepping kernels for a given $\lambda_2$ and their derivatives in the triangular domain $\mathcal{T}$. Using Theorem \ref{Deeptheo}, the maximum approximation error is less than $\epsilon$.
\end{proof}
\begin{remark}
    The error between the neural operator and the kernel operator is less than a given constant $\epsilon$. For the partial derivative of the kernels, we also have the continuous operator $\mathcal{M}: \mathbb{R}^+ \rightarrow C^1(\mathcal{T}) \times C^0(\mathcal{T}) \times C^0(\mathcal{T}) \times C^1(\mathcal{T}) \times C^1(\mathcal{T})$ 
    \begin{align}
        \mathcal{M}(\lambda_2)(x,\xi):=(K(x,\xi), \kappa_1(x,\xi),\kappa_2(x,\xi),\kappa_3(x,x),\kappa_4(x)) 
    \end{align}
    where $\kappa_1(x,\xi) =\lambda_2 K^w_x(x,\xi) - \lambda_1 K^w_\xi(x,\xi)- c(x)K^v(x,\xi)$, $\kappa_2(x,\xi)=\lambda_2 K^v_x(x,\xi)  + \lambda_2 K^v_\xi(x,\xi)$,$\kappa_3(x,x) = K^w(x,x) +\frac{c(x)}{\lambda_1+\lambda_2}$, and $\kappa_4(x) = K^v(x,0)+K^w(x,0)$. And there exists a neural operator $\hat{\mathcal{M}}$ such that for all $(x,\xi)\in \mathcal{T}$
    \begin{align}
        \sup_{\lambda_2 \in \mathcal{U}} \norm{\mathcal{M}(\lambda_2)(x,\xi) - \hat{\mathcal{M}}(\lambda_2)(x,\xi)} < \epsilon.
    \end{align}
    So there exists the neural operator $\mathcal{K}(\lambda_2)(x,\xi)$, such that 
    \begin{align}
        \sup_{\lambda_2\in\mathcal{U}} \norm{\mathcal{K}(\lambda_2)(x,\xi)  - \mathcal{\Hat{K}}(\lambda_2)(x,\xi)} + \norm{\partial_x(\mathcal{K}(\lambda_2)(x,\xi)  - \mathcal{\Hat{K}}(\lambda_2)(x,\xi))} + \norm{\partial_\xi(\mathcal{K}(\lambda_2)(x,\xi)  - \mathcal{\Hat{K}}(\lambda_2)(x,\xi))} < \epsilon
    \end{align}
\end{remark}
We then provide the stability analysis of the ARZ traffic system with the NO-approximated kernels. We first start with the approximated kernels and put them into the ARZ system to get the NO-approximated target system. For a given value of $\lambda_2$, defining the output of the neural operator $\mathcal{\Hat{K}}(\lambda_2)(x,\xi)$:
\begin{align}
    \Hat{K}^w &= \mathcal{\Hat{K}}^w(\lambda_2)(x,\xi),\\
    \Hat{K}^v &= \mathcal{\Hat{K}}^v(\lambda_2)(x,\xi).
\end{align}
For the NO-approximated kernels $\Hat{K}^w, \Hat{K}^v$, we have the following result for the NO-approximated system.
{
\begin{remark}
    The maximum approximation error $\epsilon$ gives the error bound of the neural operator approximation and provides the sufficient stability condition to prove the closed-loop system with NO-approximated kernels. It is dependent on the network size and the neural layers. In other words, the proposed method can be generalized given any $\epsilon$-accuracy by increasing the network size. The convergence speed is also related to the approximation error $\epsilon$. %A smaller approximation error leads to a faster convergence speed.    But the network size can not be determined through the universal approximation theorem and there is no guarantee that any finite size is enough to approximate the kernels within the error $\epsilon$.
\end{remark}
}

\begin{theorem}\label{es-NO-k}
   The system \eqref{origin1}-\eqref{origin2} with boundary conditions \eqref{bc_q}-\eqref{bc_v} is locally exponential stable under the control law \eqref{control_mu2k} with initial conditions $\Bar{\rho}(x,0)$, $\Bar{v}(x,0)$, satisfying
    \begin{align}
        \norm{(\Bar{\rho}(x,t),\Bar{v}(x,t))}_{L_2}^2 \leq c_1 \mathrm{e}^{-\eta t}\norm{(\Bar{\rho}(x,0),\Bar{v}(x,0))}^2_{L^2}, 
    \end{align}
    where $c_1 = \frac{m_1n_2k_1}{m_2n_1k_2}$, $m_1>0, m_2>0,n_1>0,n_2>0,k_1>0,k_2>0$, $\eta = \nu - \frac{2a\epsilon(2\lambda_1 + (1+L)\lambda_2)}{m_1\lambda_2}(1+\frac{1}{k_1}) - \frac{2a\epsilon \lambda_2}{m_1\lambda_2}$, $a>0$. The kernels are approximated by the neural operator \eqref{neuraloperator} with accuracy $\epsilon$. The traffic system can eventually achieve to its equilibrium.
\end{theorem}
\begin{proof}
    First, we define the error for the NO-approximated kernels and backstepping kernels:
$\Tilde{K}^w(x,\xi) = K^w(x,\xi) - \Hat{K}^w(x,\xi)$, $\Tilde{K}^v(x,\xi) = K^v(x,\xi) - \Hat{K}^v(x,\xi)$. We start from the boundary control model~\eqref{bs-q}-\eqref{bs-bc_l}, and the baskstepping transformation then is turned into:
\begin{align}
    \Hat{\alpha}(x,t) &= \Tilde{w}(x,t),\label{back-NO1}\\
    \Hat{\beta}(x,t) &= \Tilde{v}(x,t) - \int_0^x \Hat{K}^w(x,\xi)\Tilde{w}(\xi,t)d\xi -  \int_0^x \Hat{K}^v(x,\xi)\Tilde{v}(\xi,t)d\xi, \label{back-NO2}
\end{align}
the corresponding backstepping control law is
\begin{align}\label{control_mu2k}
    {U}(t) =  - \kappa \Tilde{w}(L,t) + \int_0^L \Hat{K}^w(L,\xi)\Tilde{w}(\xi,t) d\xi + \int_0^L \Hat{K}^v(L,\xi)\Tilde{v}(\xi,t) d\xi.
\end{align}
Thus we get the target system with the NO-approximated kernels as
\begin{align}
    \partial_t \Hat{\alpha}(x,t) + \lambda_1 \partial_x\Hat{\alpha}(x,t) &= 0, \label{tar-NO1}\\
    \partial_t \Hat{\beta}(x,t)  - \lambda_2 \partial_x \Hat{\beta}(x,t) &=\lambda_2(\Tilde{K}^w(x,0) + \Tilde{K}^v(x,0))\Tilde{v}(0,t) + (\lambda_1 + \lambda_2) \Tilde{K}^w(x,x)\Tilde{w}(x,t)\nonumber\\
    & + \int_0^x (\lambda_2 \Tilde{K}^w_x(x,\xi) + \lambda_1 \Tilde{K}^w_{\xi}(x,\xi))\Tilde{w}(\xi,t) d\xi \nonumber\\
    &+\int_0^x (\lambda_2 \Tilde{K}^v_x(x,\xi) + \lambda_2  \Tilde{K}^v_{\xi}(x,\xi))\Tilde{v}(\xi,t) d\xi, \\
    \Hat{\alpha}(0,t) &= -r\Hat{\beta}(0,t),\\
    \Hat{\beta}(L,t) &= 0\label{tar-NO2}.
\end{align}
For the target system \eqref{tar-NO1}-\eqref{tar-NO2} with the NO-approximated kernels, we define the Lyapunov candidate as
\begin{align}\label{Lyap}
    V_k(t) = \int_0^L \frac{\mathrm{e}^{-\frac{\nu}{\lambda_1}x}}{\lambda_1}\Hat{\alpha}^2(x,t) + a\frac{\mathrm{e}^{-\frac{\nu}{\lambda_2}x}}{\lambda_2} \Hat{\beta}^2(x,t) dx,
\end{align}
where the coefficients $\nu$ and $a$ are constants and $\nu>0$, $a>0$. The states of the NO-approximated backstepping target system $(\Hat{\alpha},\Hat{\beta})$ and the original states $(\Tilde{w},\Tilde{v})$ have equivalent $L^2$ norms 
\begin{align}
    k_1 \norm{(\Tilde{w}(x,t),\Tilde{v}(x,t))}_{L_2}^2 \leq \norm{(\Hat{\alpha}(x,t),\Hat{\beta}(x,t))}_{L_2}^2 \leq k_2 \norm{(\Tilde{w}(x,t),\Tilde{v}(x,t))}_{L_2}^2.
\end{align}
In the mean time, the Lyapuov functional $V_k(t)$ is also equivalent to the $L^2$ norm of the target system, so there exist two constants $m_1>0$ and $m_2>0$, 
\begin{align}
    m_1 \norm{(\Hat{\alpha}(x,t),\Hat{\beta}(x,t))}_{L_2}^2 \leq V_k(t) \leq m_2 \norm{(\Hat{\alpha}(x,t),\Hat{\beta}(x,t))}_{L_2}^2.
\end{align}
% \begin{align}
%     m_1||(\Tilde{w},\Tilde{v})||^2_{L^2} \leq V_k(t) \leq m_2||(\Tilde{w},\Tilde{v})||^2_{L^2}.
% \end{align}
Taking time derivative along the trajectories of the system, and then we plug the system dynamics. Integrating by parts, thus we get the following result of the Lyapunov candidate. The details of the proof are presented in the Appendix.~\ref{append}.

\begin{align}\label{lyapunov_bound}
    \Dot{V}_k(t) \leq - \eta V_k(t) + \left( r^2-a+2aL\epsilon \right)\Hat{\beta}^2(0,t) - \mathrm{e}^{-\frac{\nu}{\lambda_1}L}\Hat{\alpha}^2(L,t),
\end{align}
where $\eta = \nu - \frac{2a\epsilon(2\lambda_1 + (1+L)\lambda_2)}{m_1\lambda_2}(1+\frac{1}{k_1}) - \frac{2a\epsilon \lambda_2}{m_1\lambda_2}$. The coefficients $\nu, \epsilon, a$ are chosen such that
\begin{align}
    \eta >0, r^2-a+2aL\epsilon < 0.
\end{align}
So we get the following result:
\begin{align}
    \Dot{V}_k(t) &\leq -\eta V_k(t) \rightarrow V_k(t) \leq V(0)\mathrm{e}^{-\eta t}.
\end{align}
Using the equivalent norm of the Lyapunpov functional, we have:
\begin{align}
    ||(\Tilde{w}(x,t),\Tilde{v}(x,t))||^2_{L^2} \leq \mathrm{e}^{-\eta t}\frac{m_1k_1}{m_2k_2} ||(\Tilde{w}(x,0),\Tilde{v}(x,0))||^2_{L^2}.
\end{align}
Thus, the exponential stability of the NO-approximated PDE system \eqref{tar-NO1}-\eqref{tar-NO2} is proved. The state $\Tilde{w}(x,t)$ is obtained from \eqref{riemann}, so we have the following equivalent $L_2$ norm:
\begin{align}\label{equivalent}
    n_1\norm{(\Bar{\rho}(x,t),\Bar{v}(x,t))}_{L_2}^2 \leq \norm{(\Tilde{w}(x,t),\Tilde{v}(x,t))}^2_{L^2} \leq n_2\norm{(\Bar{\rho}(x,t),\Bar{v}(x,t))}_{L_2}^2,
\end{align}
where $n_1>0$ and $n_2 > 0$. Therefore, for the original $(\Bar{\rho}(x,t),\Bar{v}(x,t))$ system, we get
\begin{align}
    \norm{(\Bar{\rho}(x,t),\Bar{v}(x,t))}_{L_2}^2 \leq c_1 \mathrm{e}^{-\eta t}  ||(\Bar{\rho}(x,0),\Bar{v}(x,0))||^2_{L^2}.
\end{align}
where $c_1 = \frac{m_1n_2k_1}{m_2n_1k_2}$. Thus, we have proved that the original system \eqref{origin1}-\eqref{origin2} with boundary conditions \eqref{bc_q}-\eqref{bc_v} is locally exponential stable under the NO-approximated kernels and the system can eventually achieve to its equilibrium.
This completes the proof of Theorem \ref{es-NO-k}.
\end{proof}
This neural operator method is similar to the backstepping method. Also, then we take the coordinate transformation of  the original unstable PDE system \eqref{origin1}-\eqref{origin2} to get the boundary control model. To determine the key parameter in the boundary control model, we select character speed $\lambda_2$ as the input to the neural operator. With the trained neural operator, we get the backstepping kernels directly without solving kernel equations. Applying the NO-based kernels to the control law \eqref{control_mu2k}, and then adding the control law to the boundary control model. Thus we get the NO-based target PDE system \eqref{tar-NO1}-\eqref{tar-NO2}. We have proved that the NO-based target system is exponentially stable in the spatial-temporal domain through the Lyapunov analysis. For the DeepONet to learn an operator $\Hat{\mathcal{K}}(\lambda_2)(x,\xi)$, the inputs are taken as the characteristic speed $\lambda_2$. Then the characteristic speed $\lambda_2$ goes into the branch net and the triangular domain coordinate goes into the trunk net to train the model. The output of the brunch net and trunk net then are made dot product to get the final learned operator for the mapping $\lambda_2 \rightarrow K^w(x,\xi), K^v(x,\xi)$. The detailed diagram of the neural operator can be found in Fig.~\ref{bsdia}. The computation time and complexity are highly reduced by using the learned operator to get the backstepping kernels. There is no need to solve the kernel equations online anymore with trained operator, which can be more efficient for real-time implementation and more cheaper for the traffic administration.

In the previous section, the backstepping kernels are approximated by neural operators $\mathcal{\hat{K}}^w(\lambda_2)(x,\xi)$ and $\mathcal{\hat{K}}^v(\lambda_2)(x,\xi)$. Using the NO-approximated backstepping kernels, we have proved that the PDE system is exponentially stable. However, the control law \eqref{control_bs} still requires integration of the kernels along the road, resulting real-time implementation on the freeway difficult. Therefore, we extend the neural operator to directly approximate the mapping from the characteristic speed $\lambda_2$ to the control law~\eqref{control_bs}. The stability we achieve in the control law mapping is practical. Recalling the control law \eqref{control_bs}, we define the operator mapping $\mathcal{H}\left(\lambda_2\right): \mathbb{R}^+ \rightarrow \mathbb{R}$ that maps $\lambda_2$ to $U(t)$.  The expression of backstepping control law \eqref{control_bs} shows that there is no explicit form for the mapping from $\lambda_2$ to $U(t)$. The relation between $\lambda_2$ and $U(t)$ is characterized by the kernel equations \eqref{ker1}-\eqref{ker4}. The control law mapping is
\begin{align}\label{controlmu2c}
    U(t) = \mathcal{H}\left(\lambda_2\right)(L, t),
\end{align}
and the NO-approximated mapping for $\mathcal{H}\left(\lambda_2\right): \mathbb{R}^+ \rightarrow \mathbb{R} $ is defined as $\mathcal{\Hat{H}}(\lambda_2): \mathbb{R}^+ \rightarrow \mathbb{R}$. The traffic system is practical stable under the NO-approximated control law $\mathcal{\Hat{H}}(\lambda_2)$. The detailed proof of stability results is shown in Appendix~\ref{practical}.

\subsection{Physics-inform neural operator for kernel approximation}
In the previous section, we demonstrated that the NO-approximated control gain kernels ensure exponential stability of the closed system, and that the NO-approximated control law can achieve practical exponential stability. In this section, we extend the neural operator to a physics-informed neural operator (PINO) by incorporating physics constraints into the loss function. This extension aims to reduce the reliance on training data typical of the pure NO method. 
%As we know, PINN fails in handing different initial conditions and boundary conditions of PDEs, which leads to the need for retraining. The neural operator is not restricted by the different initial conditions and boundary conditions, but it requires enough training data to train the model. Obtaining input-output pairs can be difficult in some situations, which may be prohibitively expensive. Compared with PINN, it can not only solve one single PDE system  but also be trained without input-output pairs. 

Considering a general case for a linear or nonlinear differential operator $\mathcal{Q}: \mathcal{A} \times \mathcal{W} \rightarrow \mathcal{F}$, where $(\mathcal{A},\mathcal{W}, \mathcal{F})$ are function spaces. The differential operator takes the form:
\begin{align}
    \mathcal{Q}(\mathbf{u},w) &= 0, \quad \text{in} \quad D \subset \mathbb{R}^d \\
    w & = g, \quad \text{in} \quad \partial D
\end{align}
where $\mathbf{u} \in \mathcal{A}$ denotes the input functions (ie. The characteristic speed in this study) and $w \in \mathcal{W}$ denotes the solutions of the PDE system (i.e., the backstepping kernels). $g$ represents boundary conditions and initial conditions of the PDEs. Additionally, we have the operator mapping $\mathcal{G}(\mathbf{u})(y)$ following Theorem \ref{Deeptheo} using the formulation of $\mathcal{Q}(\mathbf{u},w)$. 
\begin{align}
    \mathcal{G}(\mathbf{u})(y) = w(\mathbf{u}).
\end{align}
Using the previous settings again, the neural operator can be obtained by \eqref{neuralop}. Then the loss function can be defined as
\begin{align}
    \mathcal{L}(\theta)=\mathcal{L}_{operator}(\theta)+\mathcal{L}_{physics}(\theta),
\end{align}
where $\mathcal{L}_{operator}(\theta)$ is the loss for the operator ,called data loss,  and $\mathcal{L}_{physics}(\theta)$ denotes the physical loss following the definition of loss function in PINN and PINN-type methods. The data loss is given as:
\begin{align}
     \mathcal{L}_{operator}(\theta) = \norm{\mathcal{G}_{\mathbb{N}}\left(\mathbf{u}_m\right)(y) - \mathcal{G}\left(\mathbf{u}\right)(y)}_{L_2}^2 = \int_D \left| \mathcal{G}_{\mathbb{N}}\left(\mathbf{u}_m\right)(y) - \mathcal{G}\left(\mathbf{u}\right)(y) \right|^2 dy.
\end{align}
The physical loss is defined as:
\begin{align}
    \mathcal{L}_{physics}(\theta) = \norm{\mathcal{Q}(\mathbf{u}, w)}_{L_2}^2 + \norm{w - g}_{L_2}^2 = \int_D \left| \mathcal{Q}(\mathbf{u}, w) \right|^2 dx + \int_D \left| w - g \right|^2 dx.
\end{align}
More specifically, the physical loss consists of equation loss and boundary loss in this study, $\mathcal{L}_{physics}(\theta) = \mathcal{L}_{equation}(\theta) +  \mathcal{L}_{boundary}(\theta)$. The equations loss is defined as
\begin{align}
    \mathcal{L}_{equation}(\theta) = \norm{\lambda_{2} {K}^w_{x}(x,\xi) - \lambda_{1} {K}^w_\xi(x,\xi) - c(x){K}^v(x,\xi)}_{L_2}^2+\norm{\lambda_2 {K}^v_x(x,\xi) - \lambda_2 {K}^v_{\xi}(x,\xi)}_{L_2}^2. 
    \label{phy1}
\end{align}
The boundary loss is defined as
\begin{align}
    \mathcal{L}_{boundary}(\theta) &= \norm{K^w(x,x) + \frac{c(x)}{\lambda_1 + \lambda_2}}_{L_2}^2 + \norm{K^v(x,0) + K^w(x,0)}_{L_2}^2.
    \label{phy2}
\end{align}
The diagram of PINO is shown in Fig.~\ref{consturc-deeppi}.
{
    The proposed PINO is an extension of the NO methods. The primary distinction between PINO and NO lies in the design of the loss function within the trained model. PINO seeks to integrate PDE model constraints into the training process by penalizing the loss function. Consequently, the training of PINO may be facilitated by prior knowledge of the PDE system embedded in the model.
    The weights assigned to the operator loss and the physics loss are hyperparameters that can be defined by the user or tuned, and they are crucial in enhancing the trainability of PINO. Depending on the choice of these weights, the trained PINO model will rely more heavily on either the operator data or the physical kernel equations. In this paper, we consider equal weights for the operator loss and physics loss. Further discussion on the selection and tuning of weights for the physics-informed neural network and PINO can be found in~\citep{wang2021understanding,wang2022and,karniadakis2021physics}.
}

\begin{figure}[!tbp]
    \centering
    \includegraphics[width = 0.8\textwidth]{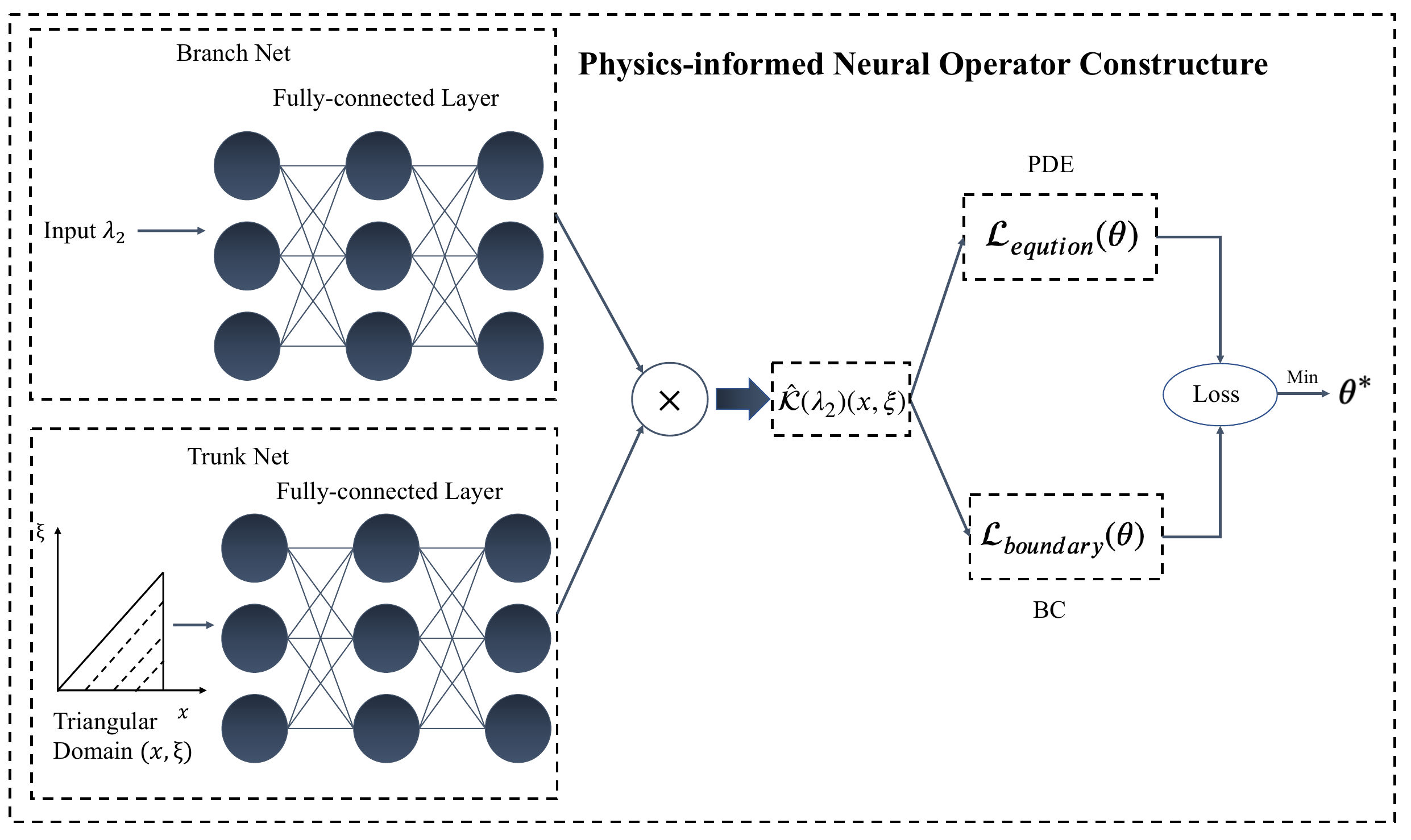}
    \caption{The diagram of physics-informed neural operator structure for backstepping kernels}
    \label{consturc-deeppi}
\end{figure}

\section{Experiments}
In this section, we present and analyze the performance of the proposed neural operator controllers for the ARZ traffic PDE system, and also provide comparisons with model-based controllers: (i) backstepping controller (ii) PI controller.

{
\subsection{Simulation and training setups}
\textbf{Simulation setup}. To train the neural operator, we use numerical simulations to generate training data. We run the simulation on a road of length $L=500 \text{m}$ for a duration of $T=300\text{s}$. 
The free-flow speed is $v_f = 144 \text{km/h}$, the maximum density is $\rho_m = 160 \text{veh/km}$, the equilibrium density is selected as $\rho^\star = 120 \text{veh/km}$, $v^\star = 36 \text{km/h}$ the reaction time for the drives adapting to speed is $\tau = 60 \text{s}$, and $\gamma = 1$. For the initial conditions, we use the sinusoidal inputs to mimic the stop-and-go traffic waves.
\begin{align}
    \rho(x,0) &= \rho^\star + 0.1 \sin \left( \frac{3\pi x}{L}\right)\rho^\star,\\
    v(x,0) &= v^\star - 0.1\sin \left( \frac{3\pi x}{L}\right)v^\star.
\end{align}
Using the above parameters of the traffic system, we first run the simulation using numerical solvers to get the evolution of traffic states in the simulation period.
}

{
\textbf{Training setup}. To obtain sufficient training data, we use $1000$ different values for $\rho^\star \in [90 \text{veh/km},130\text{veh/km}]$ to obtain varying values for $\lambda_2 \in [18\text{km/h},90\text{km/h}]$ and different kernels $K^w(x,\xi), K^v(x,\xi)$. We use 1000 samples of $\lambda_2$ and its corresponding kernels $K^w(x,\xi), K^v(x,\xi)$ by numerically solving the kernel equations. The samples of $\lambda_2$ are randomly sampled from the interval $[18\text{km/h},90\text{km/h}]$. The data is divided into a training set and a test set in a 9:1 ratio, with 900 samples for training and 100 samples for testing.  We use the Adam optimizer with a learning rate of 0.0001, decaying every 200 epochs over a total of 1000 epochs. The batch size is set to 20. The training process takes approximately 10 minutes on a single Nvidia 4090Ti GPU.
}

{
\textbf{Model structure}. 
The DeepONet architecture is employed as the foundational structure of the neural operator. Both branch and trunk networks are designed as fully connected networks. The input to the branch network is selected as the traffic congestion wave speed $\lambda_2$, while the input to the trunk network is chosen as the triangular domain. The output of the trained neural operator is then generated by taking the dot product of the outputs from the branch and trunk networks.
}

\subsection{ PI controller}
To assess the performance of the NO-based controllers, we also include a comparison with the PI controller. Previous research has demonstrated that the PI boundary controller can stabilize the traffic system~\citep{zhang_pi_2019}. The PI controller is installed at the outlet of the road section, resulting in the boundary condition of the traffic speed being $\Tilde{v}(L,t) = U_{PI}(t)$. The control law is given by:
\begin{align}
    U_{PI}(t) = v^\star + k_p^v (v(0,t) - v^\star) + k_i^v \int_0^t (v(0,t) - v^\star) ds,
\end{align}
where $k_p^v$ and $k_i^v$ are tuning gains.
{
\subsection{PINN-based controller}
In addition to the PI controller, we further compared NO-based methods with PINN-approximated kernels. The input of PINN is the grid size of the triangular domain $x,\xi$ and the output is the backstepping kernel $K^2(x,\xi), K^v(x,\xi)$. The structure of PINN is the same as~\citep{raissi2019physics, karniadakis2021physics}. We also have the general form of the kernel equations.
\begin{align}
    \mathcal{N}[u] &= 0, \quad x, \xi \in \mathcal{T}, \\
    u &= g_{\text{PINN}}, \quad u \in \partial \mathcal{T},
\end{align}
where $u(x,\xi)$ denotes the latent solution of PDEs. $\mathcal{N}$ is a nonlinear operator which describes the differentiation of PDEs. Then we define $f(x,\xi)$ to denote the residual of PDEs.
\begin{align}
    f(x,\xi) := \mathcal{N}[u].
\end{align}
The loss function is the same as PINO in Sec 3.3, consisting of data loss and physics loss.
\begin{align}
    \mathcal{L}_{\text{PINN}} = \mathcal{L}_{\text{data}}(\theta) + \mathcal{L}_{\text{physics}}(\theta),
\end{align}
where the $\mathcal{L}_{\text{physics}}(\theta)$ is the same as in Eq.~\eqref{phy1}-\eqref{phy2}, and the data loss denotes the error between the model output and the ground truth data
\begin{align}
    \mathcal{L}_{\text{data}}(\theta) = \int_\mathcal{T} \norm{u(x,\xi) - u_{\text{real}}(x,\xi)}.
\end{align}
Besides, the trained PINN model only learns one set of specific backstepping kernels. It can not output right kernels for the different $\lambda_2$.
}

\subsection{Simulation results of closed-loop system}
Since the NO-approximated kernels and control law are based on the backstepping method, we select the standard backstepping controller as the baseline for the simulation. The open-loop results for the traffic system are depicted in Fig.~\ref{open-loop}. Traffic density and speed oscillations persist throughout the simulation period, leading to the occurrence of stop-and-go waves. The density and speed of the closed-loop results using the backstepping method are illustrated in Fig.~\ref{density_all}(a) and Fig.~\ref{velocity_all}(a). It can be observed that the traffic density and speed all converge to the equilibrium points at the finite time $130\text{s}$. The closed-loop results of NO-based controller PI controller, and PINN-based controller are illustrated in Fig.~\ref{density_all} and~\ref{velocity_all}. It is revealed that all control methods effectively stabilize traffic oscillations. Traffic density and speed converge to their equilibrium point, $\rho^\star = 120 \text{veh/km}$ and $v^\star = 36 \text{km/h}$, respectively, despite the sinusoidal initial conditions that initially induce instability throughout the road section.
The closed-loop results for the PI controller are depicted in Fig.~\ref{density_all}(b) and Fig.~\ref{velocity_all} (b). The closed-loop results for the PINN-based controller are shown in Fig.~\ref{density_all}(c) and Fig.~\ref{velocity_all}(c). The closed-loop results for NO-approximated kernels are shown in Fig.~\ref{density_all}(d),~\ref{velocity_all} (d). However, traffic waves are still observable at $150\text{s}$ due to the small approximation error of the neural operator. 

Regarding the results of the NO-approximated control law, the same parameter settings as the previous section were utilized to generate training data, consisting of 900 instances. The results of the approximated control law mapping are illustrated in Fig.~\ref{density_all}(e) and Fig.~\ref{velocity_all}(e). It is observed that the NO-approximated control law can practically stabilize the system, as the system does not uniformly converge to the equilibrium point. Small oscillations in both density and speed persist throughout the entire simulation period, preventing uniform convergence to the equilibrium points. This is reasonable because the neural operator only learns the mapping from $\lambda_2$ to $U(t)$. As we know, the control law in \eqref{control_bs} consists of two parts: the backstepping kernels and the system states at the current time step. The system is practically stable under the condition of the NO-approximated control law.

Subsequently, we present the results for the density and speed errors between other methods and the backstepping method, as depicted in Fig.~\ref{err_density_all} and Fig.~\ref{err_velocity_all}. The results exclude the PI controller, as it represents a different type of controller compared to the NO-based methods, which all belong to the same backstepping category.
The error between the closed-loop result of the backstepping controller and the NO-approximated kernels is illustrated in Fig.~\ref{err_velocity_all}(a). 
It is evident that there are some errors at the initial stage of the NO-approximated kernels. The maximum error of the density is approximately $1.5 \text{veh/km}$ at the location of $80 \text{m}$ after $50 \text{s}$. The maximum speed error is $0.48 \text{km/h}$. The density and speed error reduce to zero after about $150\text{s}$. However, the density and speed errors of the NO-approximated control law persist throughout the entire time period. The density oscillation is smaller than $0.6 \text{veh/km}$ on average and the oscillation of speed is smaller than $0.4 \text{km/h}$. 
\begin{table}[htbp]
    \centering
    \caption{The closed-loop density and speed errors under different controllers }
    \begin{tabular}{c c c c c }
    \hline
       \multirow{3}{*}{\textbf{Method}}  & \multicolumn{2}{c}{$\rho(x,t)(\text{veh/km})$} & \multicolumn{2}{c}{$v(x,t)(\text{km/h})$}\\
       \cline{2-5}
         & \textbf{\makecell{Max absolute error} } & \textbf{\makecell{Mean absolute error}} & \textbf{\makecell{Max absolute error}} & \textbf{\makecell{Mean absolute error}}\\
    \hline
    % PINN-kernels & $4.3556$ & $0.3063$ & $1.3554$ & $0.1421$ \\
    % NO-kernels & $1.5025$ & $0.1159$ & $0.4835$ & ${0.0489}$ \\
    % PINO-kernels & $2.6778$ & $0.1325$ & $0.8185$ & $0.0576$ \\
    % NO-control law & ${3.8684}$ & $0.5840$ & ${1.2438}$ & $0.3423$\\
    PINN-kernels & $3.63\%$ & $0.26\%$ & $3.76\%$ & $0.39\%$ \\
    NO-kernels & $1.25\%$ & $0.09\%$ & $1.34\%$ & $0.14\%$ \\
    PINO-kernels & $2.23\%$ & $0.11\%$ & $2.27\%$ & $0.16\%$ \\
    NO-control law & $3.22\%$ & $0.49\%$ & $3.45\%$ & $0.95\%$\\
    \hline
    \end{tabular}
    \label{error_tab_rho_v_all}
\end{table}

For the training procedure of PINO, we utilize half the samples of the training dataset as before to train the model to see whether PINO can still stabilize the traffic system. Using the trained PINO model, we do the prediction for the backstepping kernels. The simulation settings are the same as before, and the results of PINO-approximated kernels are depicted in Fig.~\ref{density_all}(d),~\ref{velocity_all}(d). Fig.~\ref{density_all}(d) shows the density of the PINO-based result while Fig.~\ref{velocity_all} (d) shows the speed result. It is observed that the  maximum density and speed error of PINO-approximated kernels are $2.7 \text{veh/km}$, $0.8\text{km/h}$ from Fig.~\ref{err_density_all} (c),~\ref{err_velocity_all}(c), respectively. The density and speed errors under different NO-based schemes and the PINN method are presented in Tab.~\ref{error_tab_rho_v_all}. NO-based methods all outperform the PINN-based method.

\begin{figure}[!htbp]
\centering
\subfigure[Density]{\includegraphics[width=0.45\textwidth]{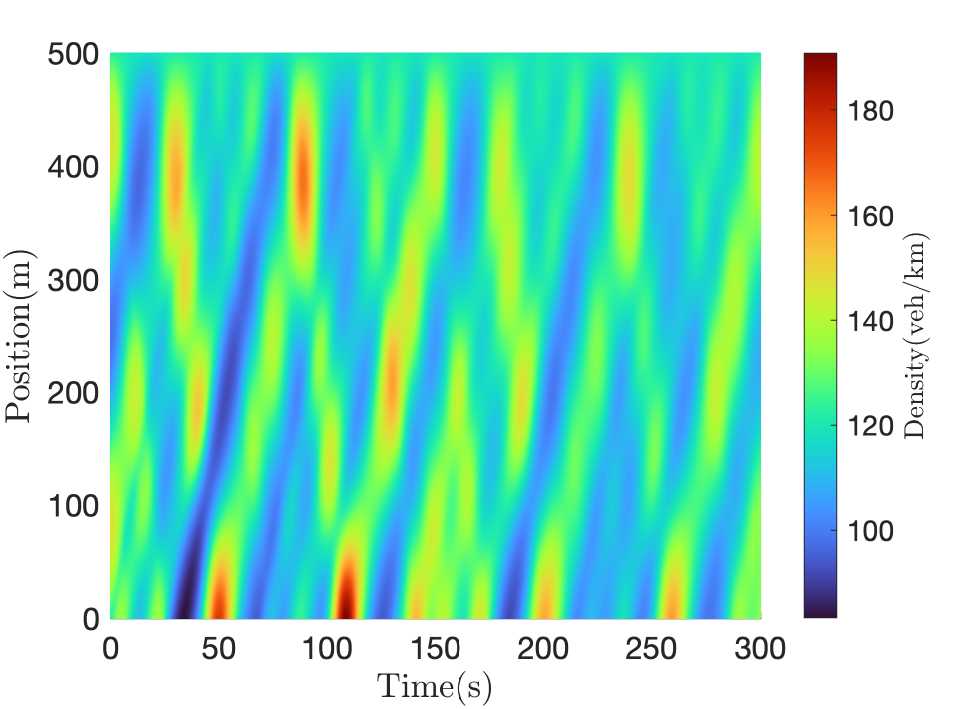}}
\subfigure[Speed]{\includegraphics[width=0.45\textwidth]{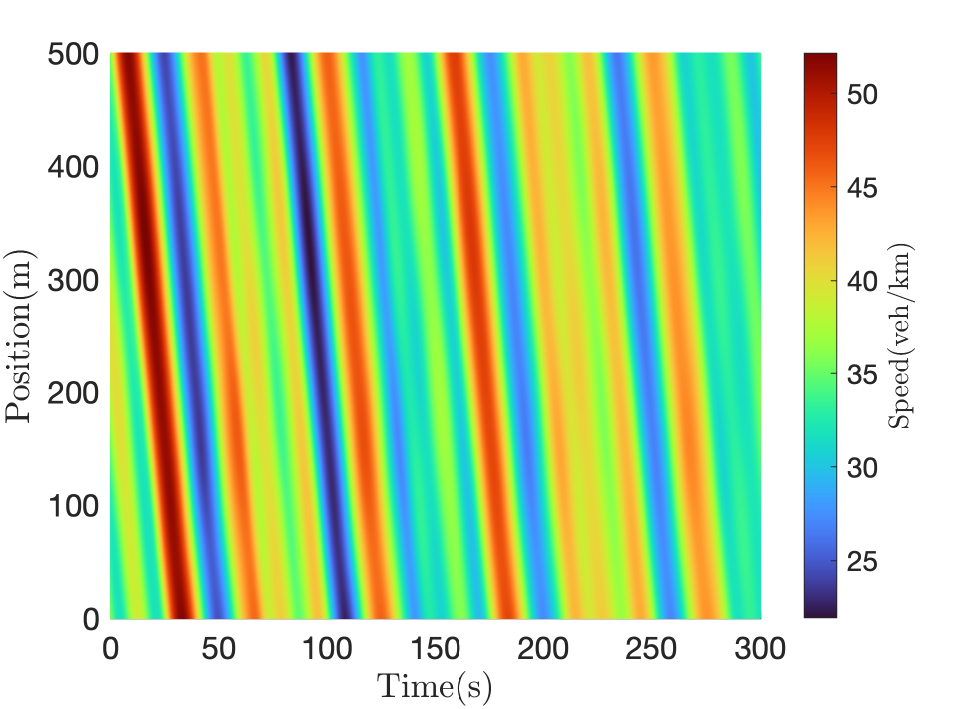}}
\caption{Traffic density and speed evolution of the open-loop system}
\label{open-loop}
\end{figure}

\begin{figure}[!htbp]
    \centering
     \subfigure[Backstpping]{\includegraphics[width=0.32\textwidth]{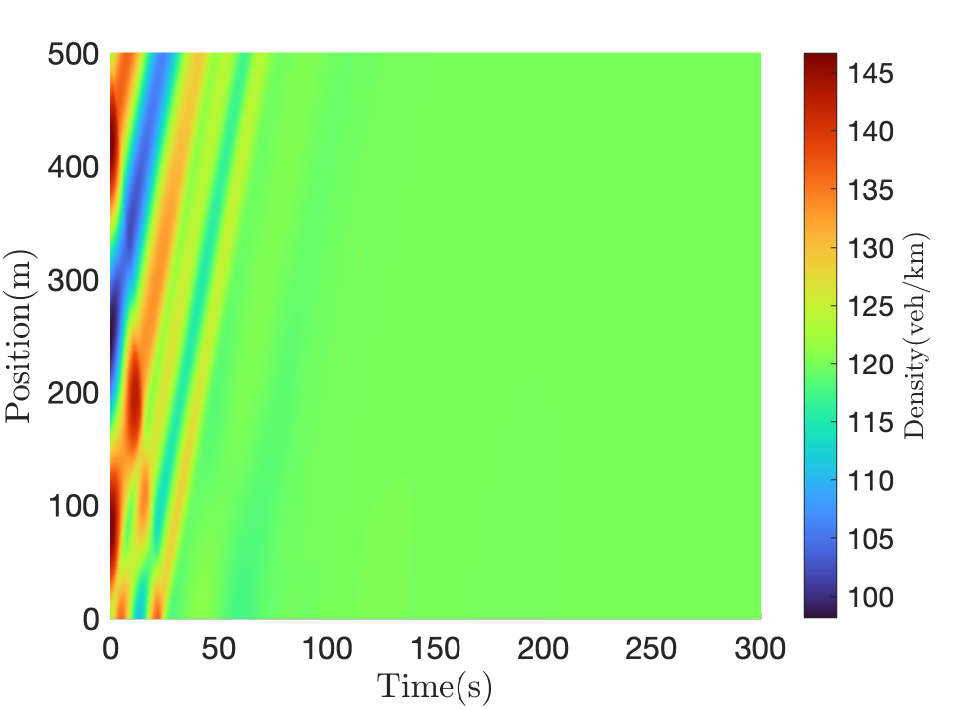}}
    \subfigure[PI controller]{\includegraphics[width=0.32\textwidth]{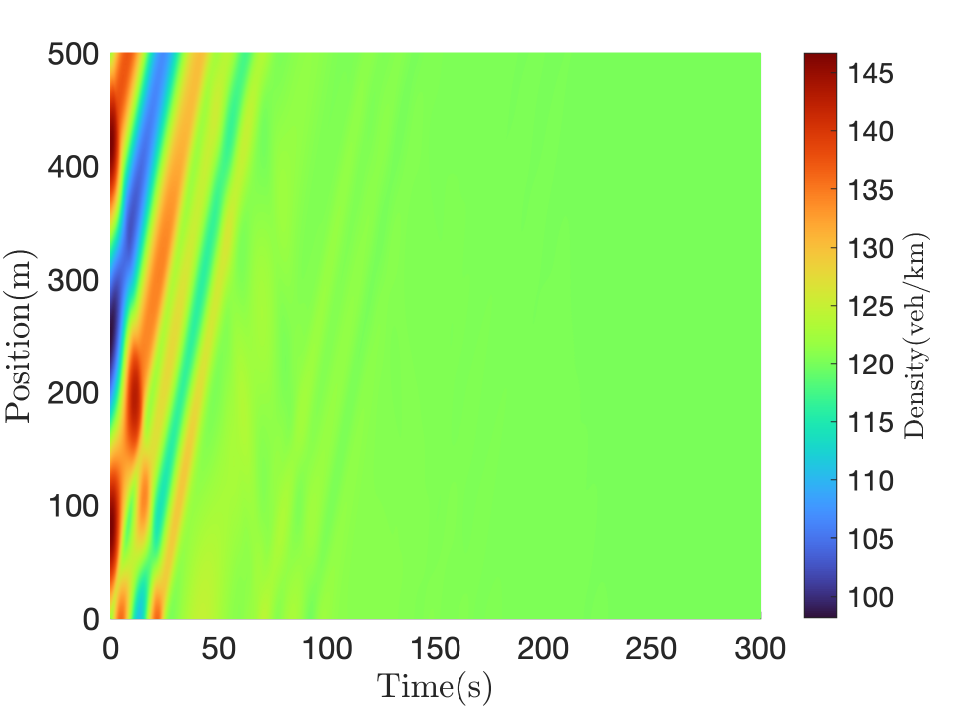}}
    \subfigure[PINN]{\includegraphics[width=0.32\textwidth]{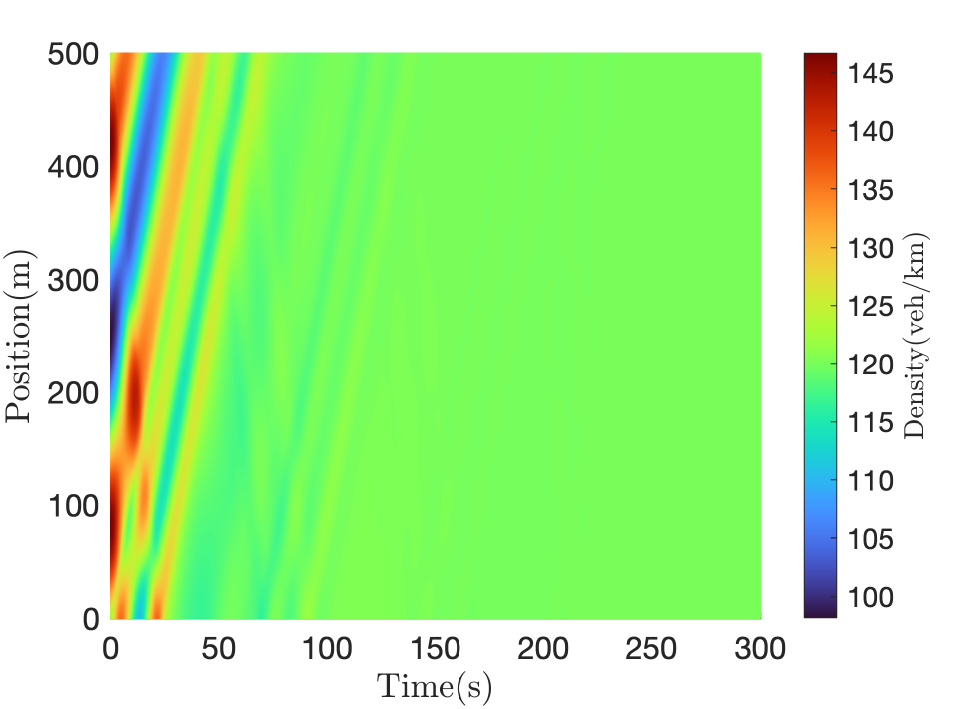}}\\
    \subfigure[NO-kernel]{\includegraphics[width=0.32\textwidth]{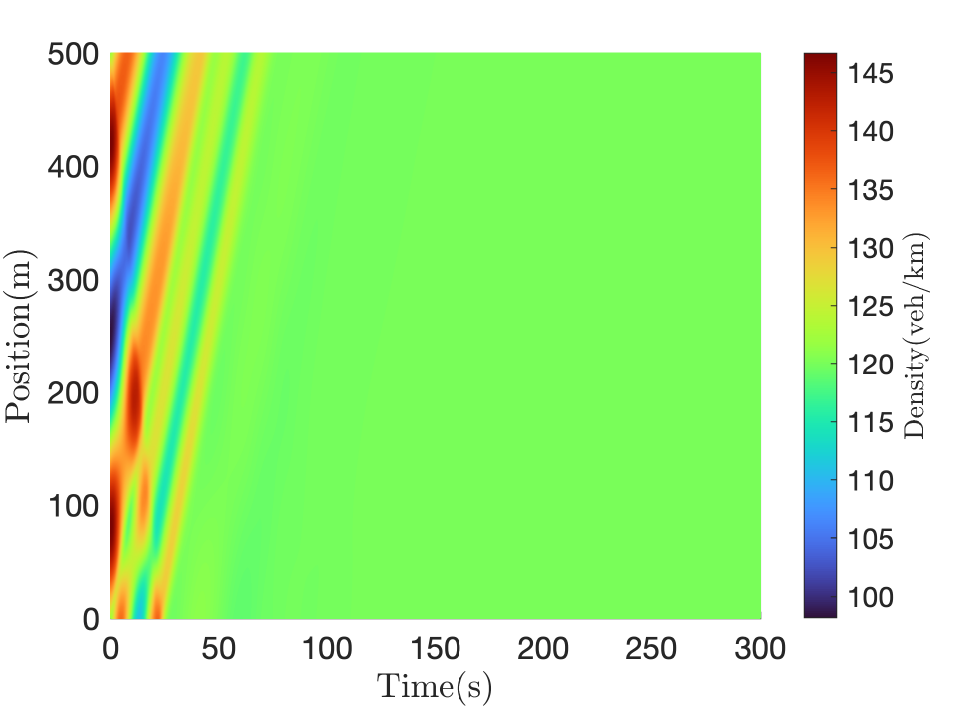}}
    \subfigure[NO-control]{\includegraphics[width=0.32\textwidth]{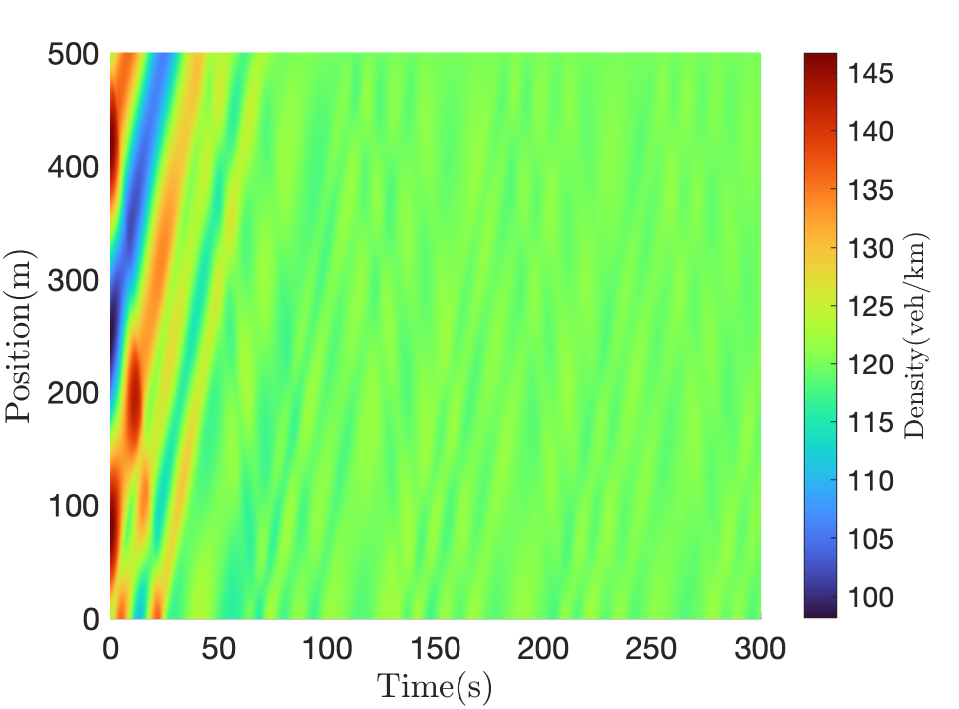}}
    \subfigure[PINO-kernel]{\includegraphics[width=0.32\textwidth]{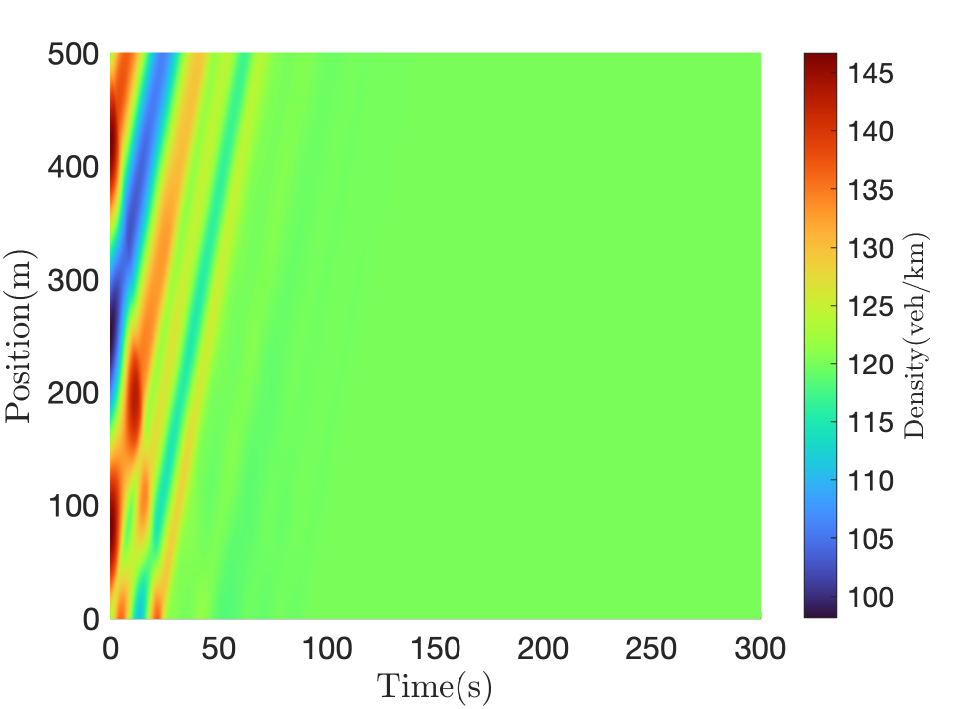}}\\
    \caption{Closed-loop traffic density evolution with different control designs }
    \label{density_all}
\end{figure}
\begin{figure}[!htbp]
    \centering
    \subfigure[Backstpping]{\includegraphics[width=0.32\textwidth]{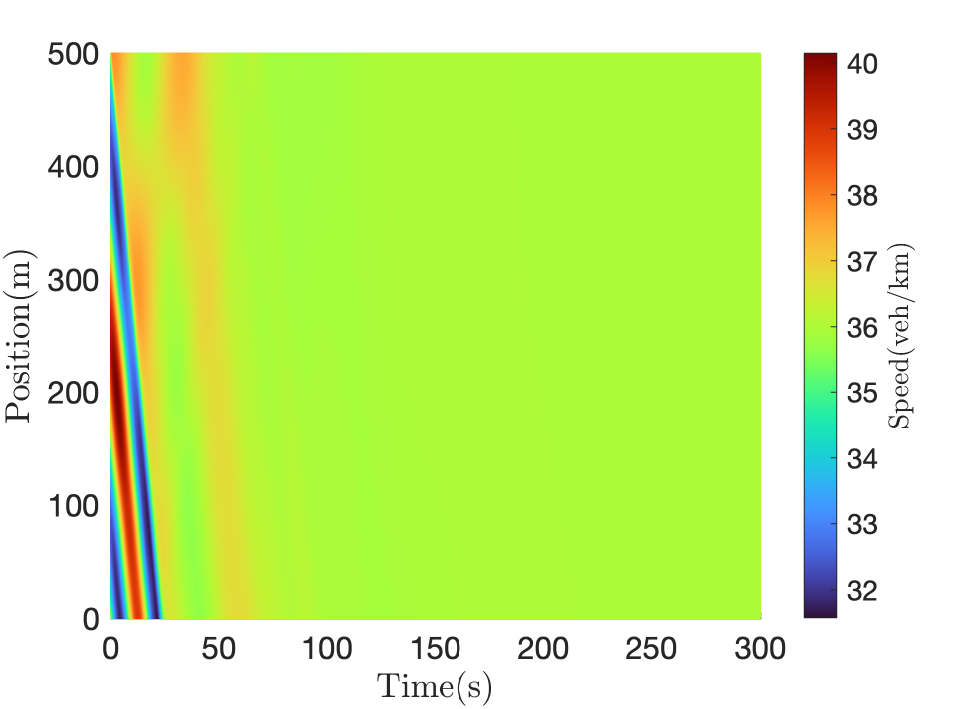}}
    \subfigure[PI controller]{\includegraphics[width=0.32\textwidth]{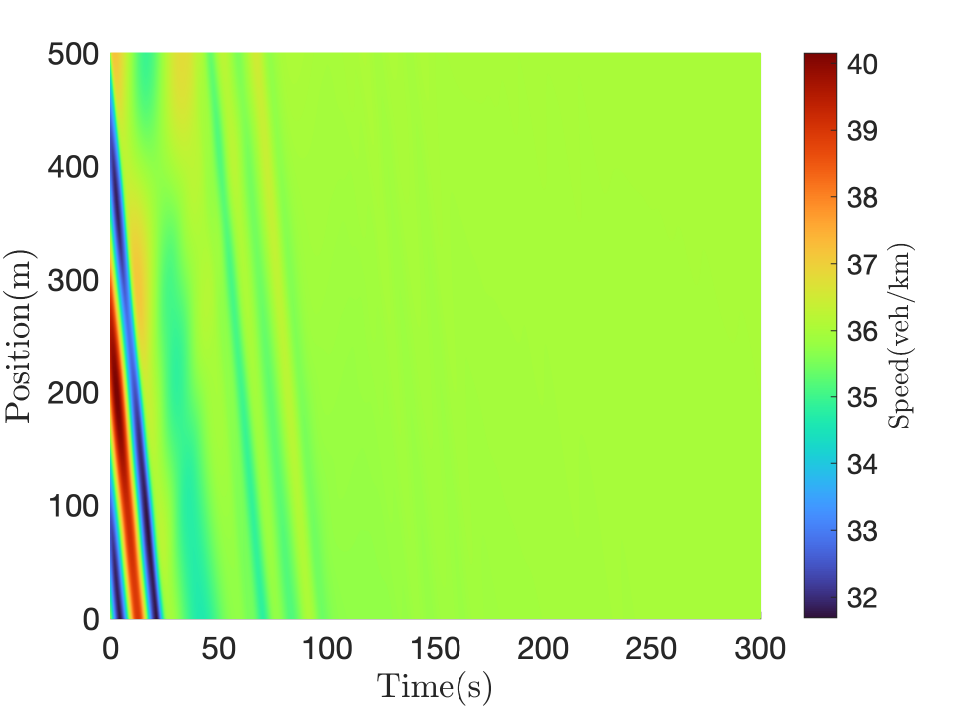}}
    \subfigure[PINN]{\includegraphics[width=0.32\textwidth]{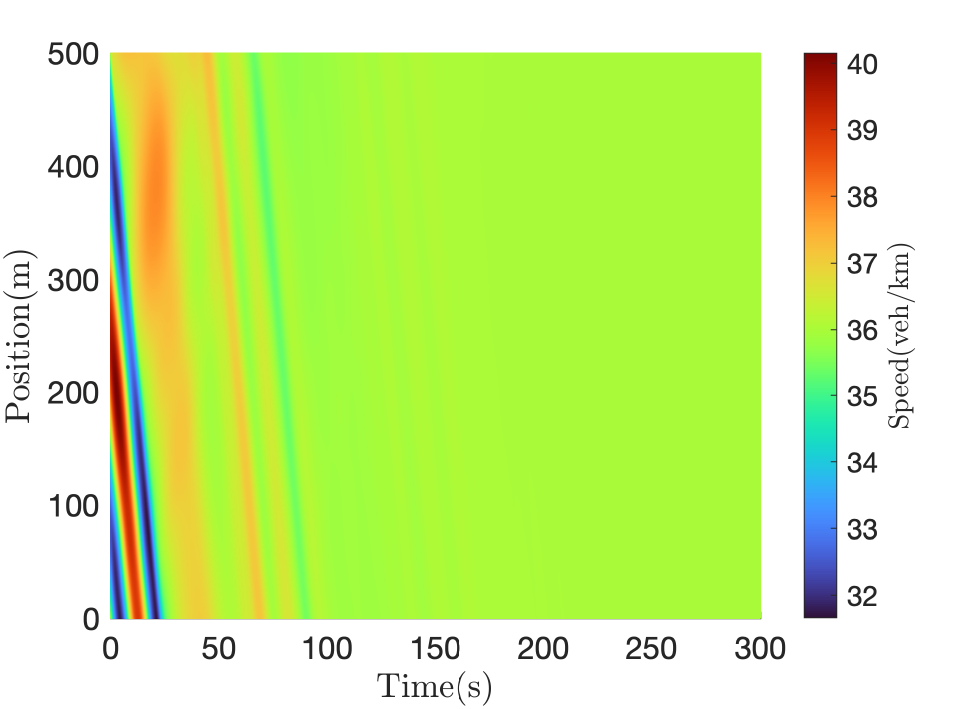}}\\
    \subfigure[NO-kernel]{\includegraphics[width=0.32\textwidth]{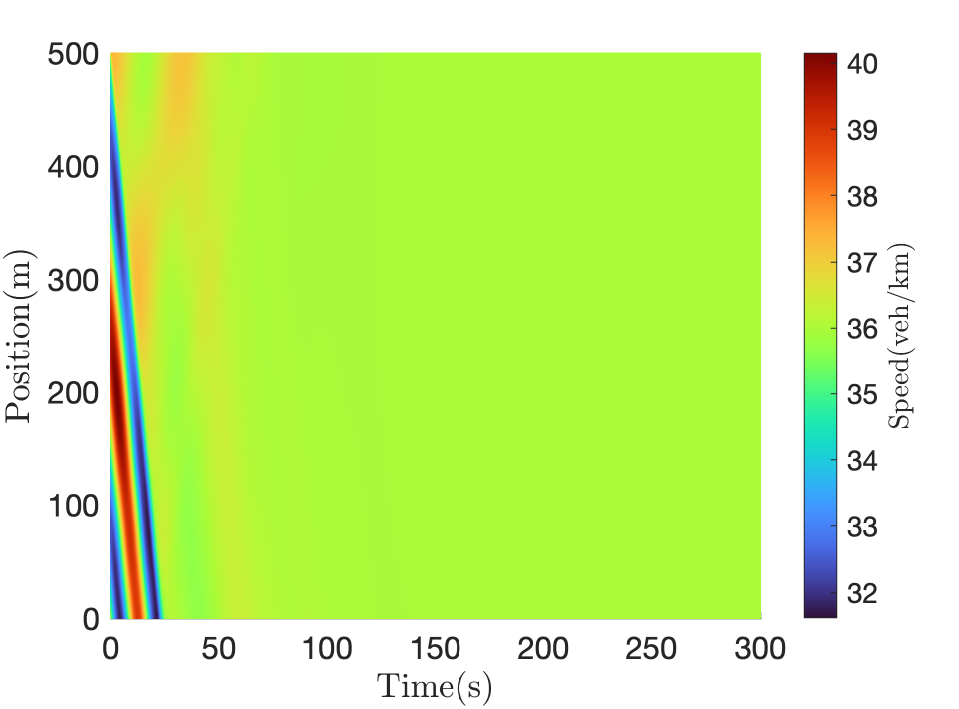}}
    \subfigure[NO-control]{\includegraphics[width=0.32\textwidth]{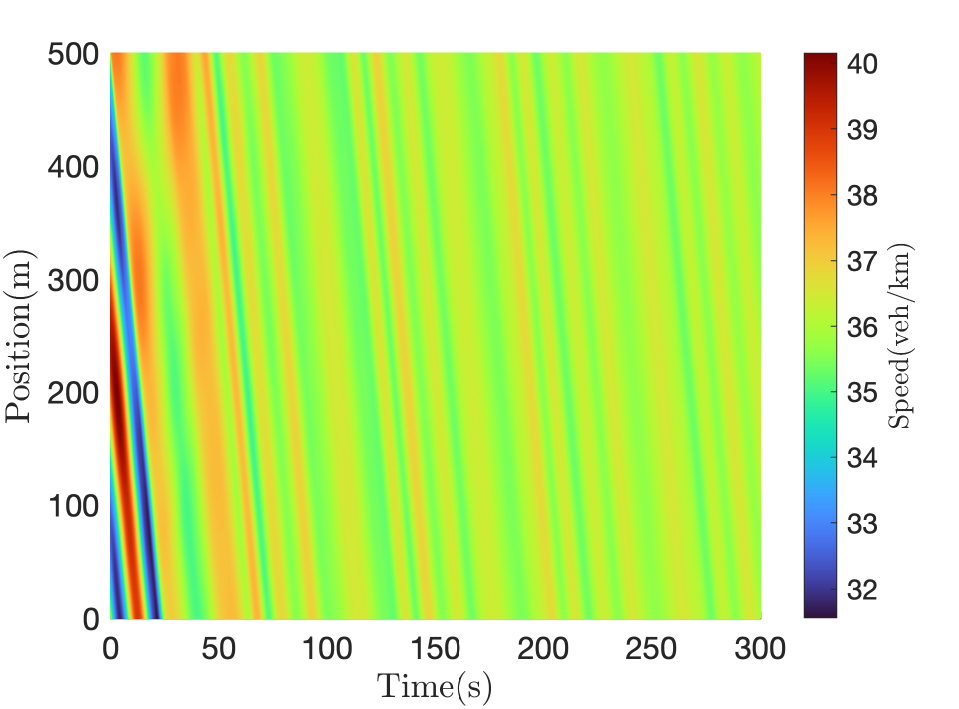}}
    \subfigure[PINO-kernel]{\includegraphics[width=0.32\textwidth]{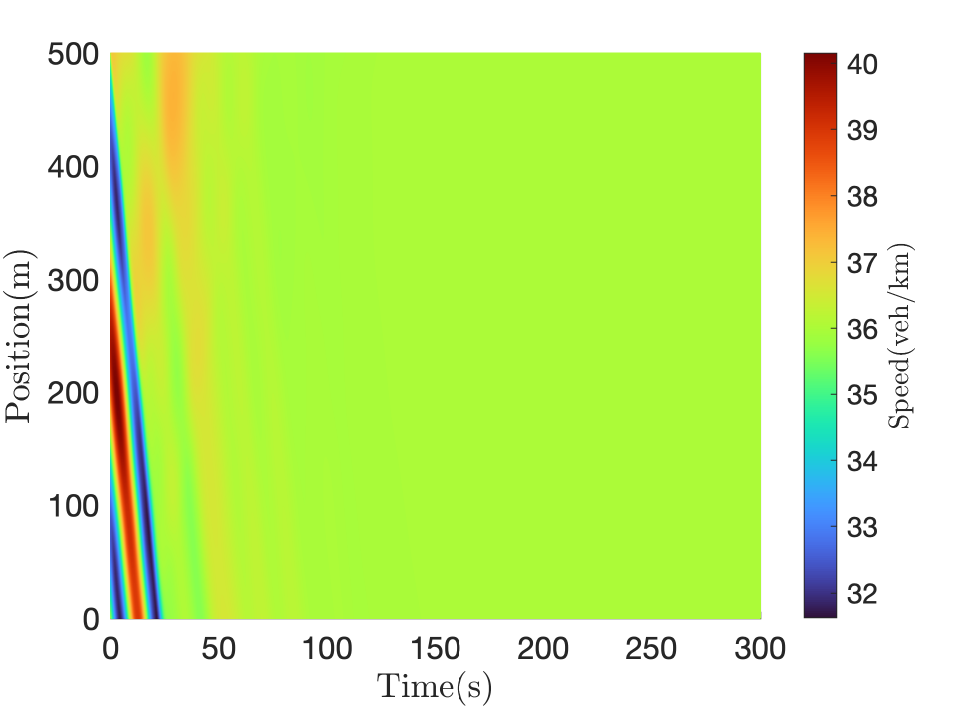}}
    \caption{Closed-loop traffic speed evolution with different control designs}
    \label{velocity_all}
\end{figure}
\begin{figure}[!htbp]
    \centering
    \subfigure[NO-kernel]{\includegraphics[width=0.32\textwidth]{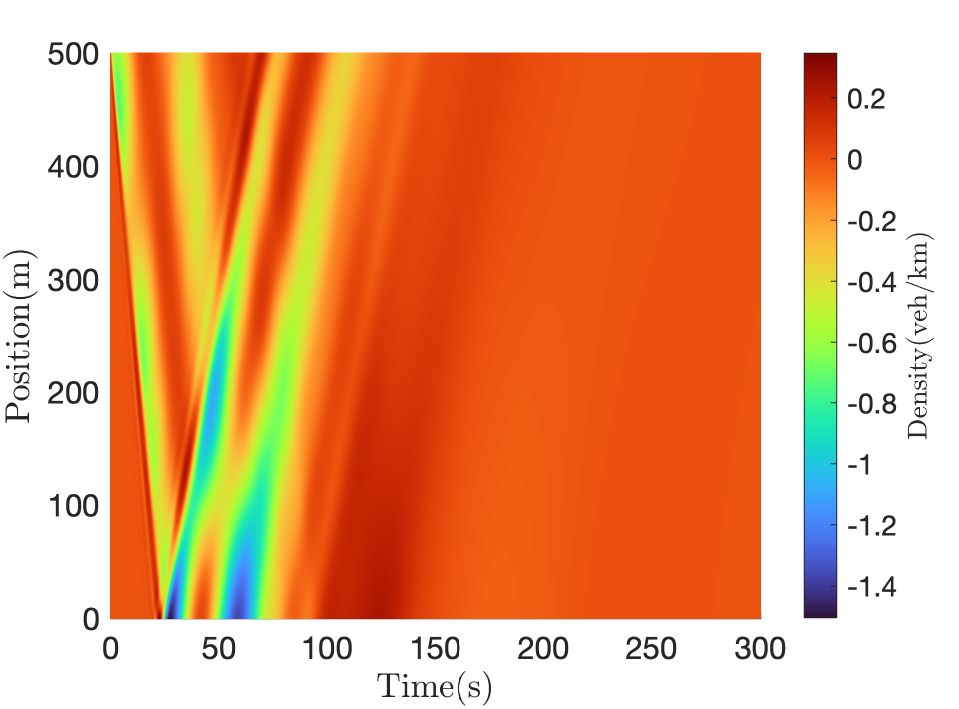}}
    \subfigure[NO-control]{\includegraphics[width=0.32\textwidth]{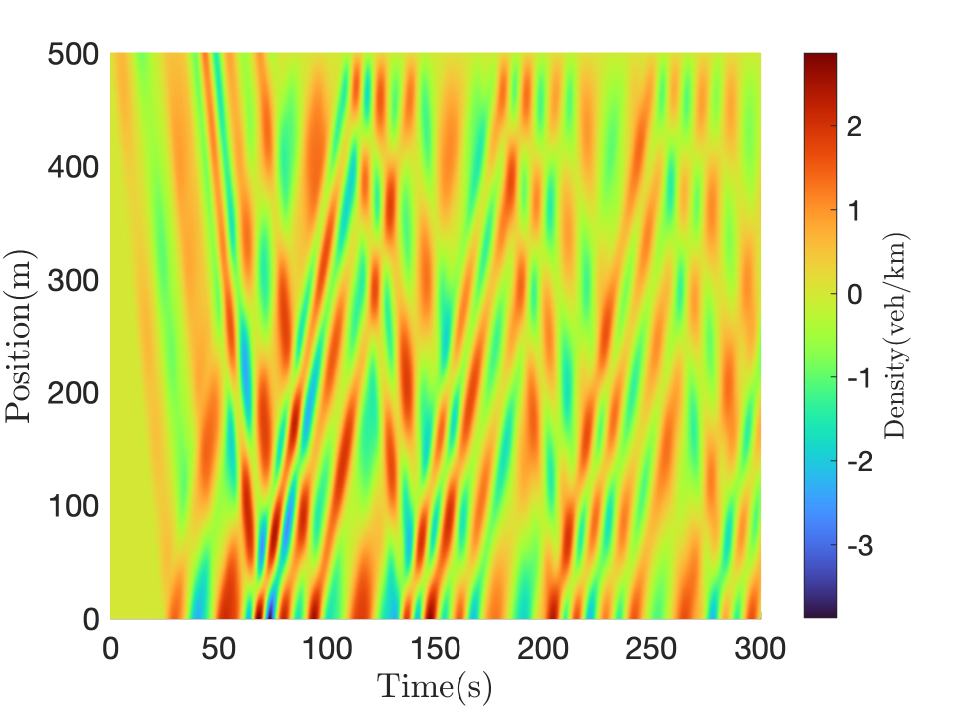}}
    \subfigure[PINO-kernel]{\includegraphics[width=0.32\textwidth]{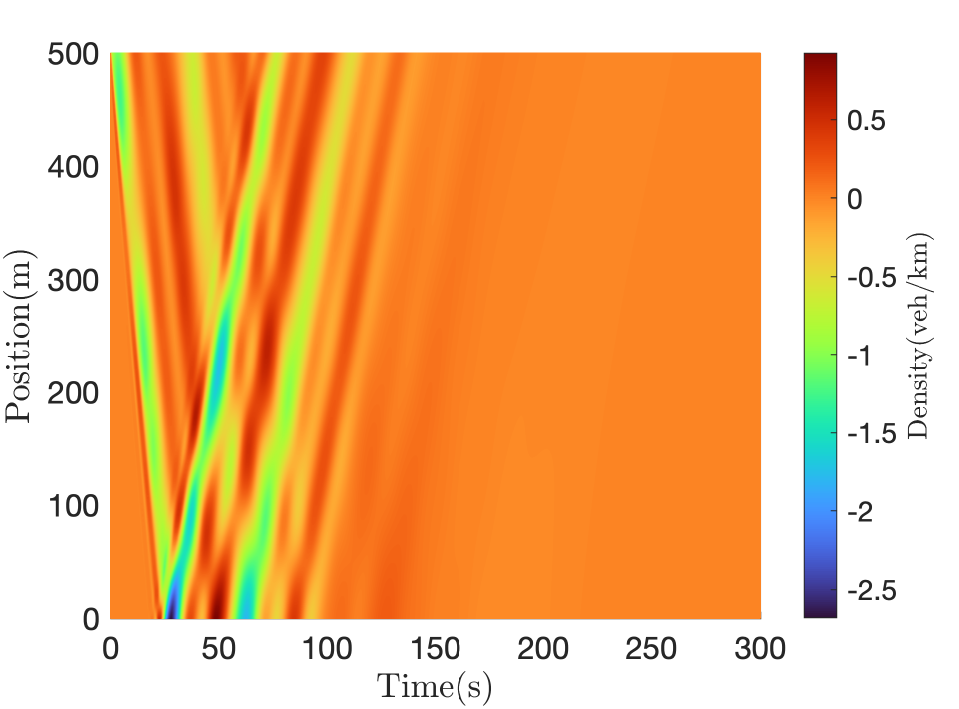}}
    \caption{Traffic density error under different schemes}
    \label{err_density_all}
\end{figure}
\begin{figure}[!htbp]
    \centering
    \subfigure[NO-kernel]{\includegraphics[width=0.32\textwidth]{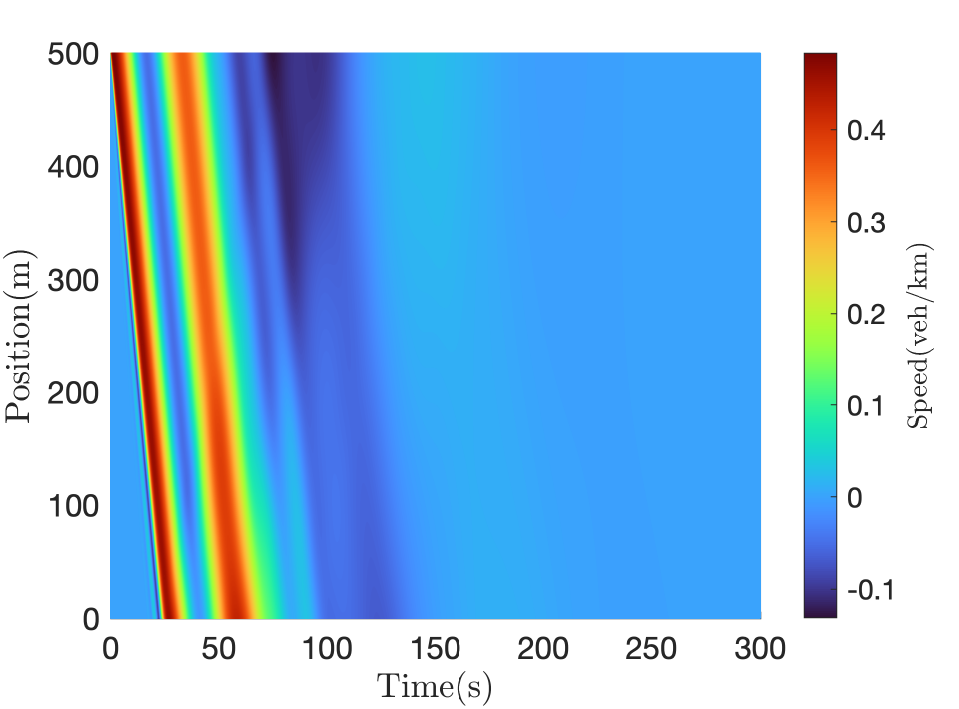}}
    \subfigure[NO-control]{\includegraphics[width=0.32\textwidth]{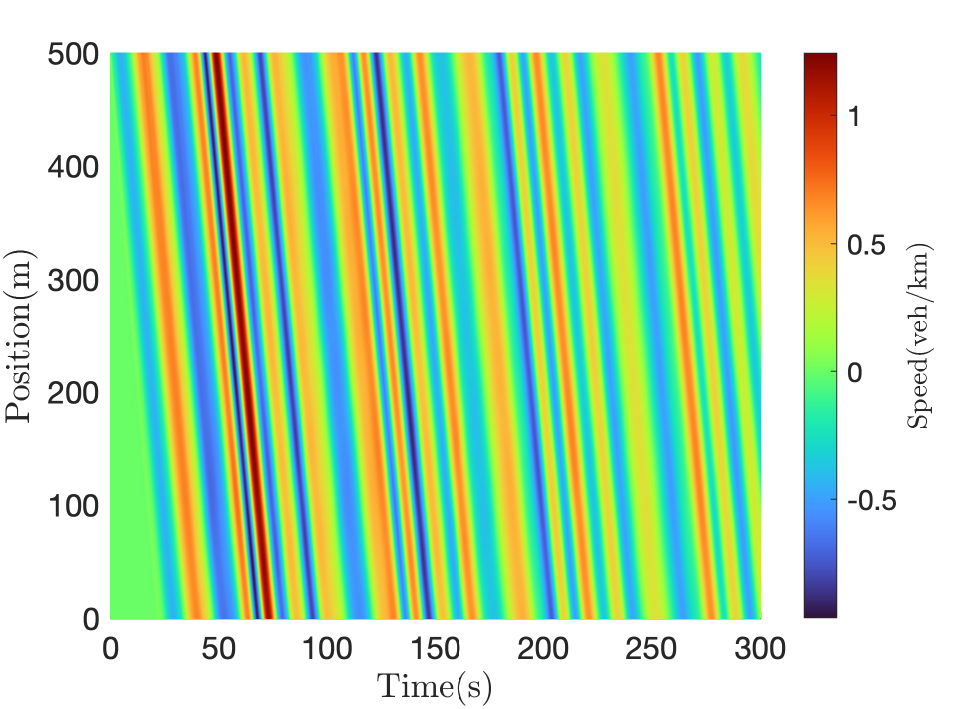}}
    \subfigure[PINO-kernel]{\includegraphics[width=0.32\textwidth]{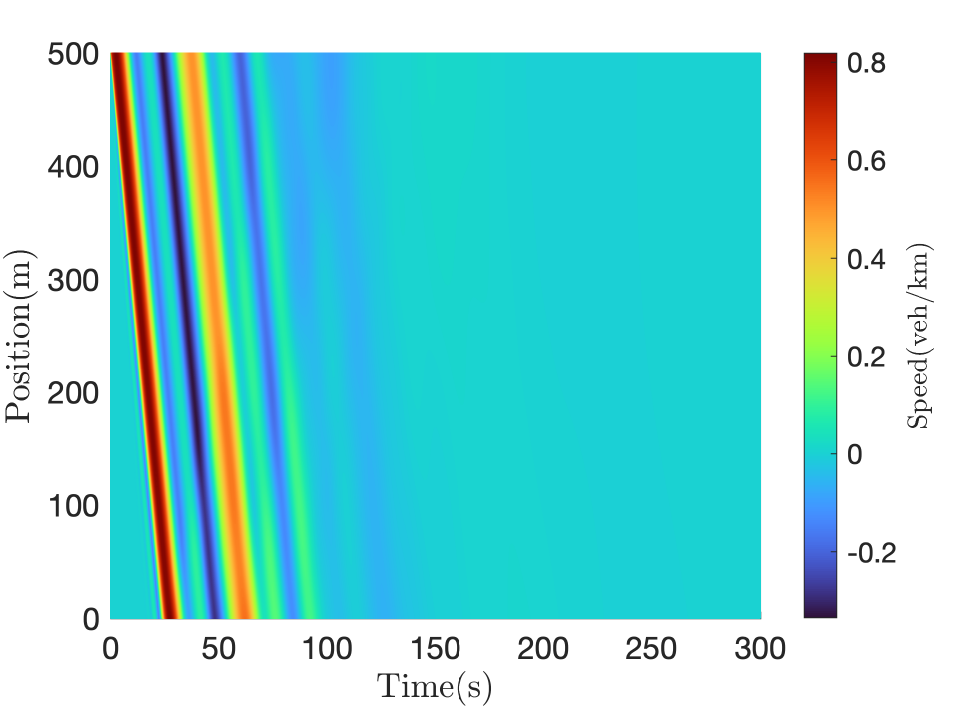}}
    \caption{Traffic velocity error under different schemes}
    \label{err_velocity_all}
\end{figure}

\subsection{Simulation results of backstepping kernels}
Previous section gives the density and speed of the closed-loop system. The NO-based methods learn the backstepping kernels, except for NO-approximated control law. It directly learns the boundary control law for the traffic system. In this section, the results for the NO-approximated kernels are provided. The backstepping kernels, NO-approximated kernels and PINO-approximated kernels are shown in Fig.~\ref{kernel_all}. 
To evaluate the performance of the three NO methods, we use the backstepping kernels as the baseline for comparison. 
The error between the NO methods and the backstepping method is shown in Fig.~\ref{err_kernel_all}. The first row of Fig.~\ref{kernel_all} are the results of kernel $\hat{K}^w(x,\xi)$, while the second row shows the results of kernel $\hat{K}^v(x,\xi)$. It is evident that the NO-approximated kernels provided the best approximation to the backstepping kernels. 
The PINO-approximated kernels exhibit oscillations in the triangular domain, resulting in larger state errors compared to the NO-approximated errors. Therefore, the model does not accurately capture the properties of the kernels. The errors between the NO-based methods and backstepping kernels are depicted in Fig.~\ref{err_kernel_all}. In the first column, the errors of $K^w(x,\xi)$ using different schemes are presented, while the second column shows the errors of $K^v(x,\xi)$. The maximum error of $K^w(x,\xi)$ and $K^v(x,\xi)$ using NO-approximated kernels is $1.452\times 10^{-4}$ and the maximum error of is $1.451\times10^{-4}$. Compared with NO-approximated kernels, the errors of $K^w(x,\xi)$ and $K^v(x,\xi)$ are higher using PINO-approximated methods. The maximum and mean errors under different schemes are shown in Tab.~\ref{error_tab_kernel_all}. The NO-based methods still achieve good performance compared with the PINN method.  
\begin{table}[htbp]
    \centering
    \caption{The errors of kernels under different schemes}
    \begin{tabular}{c c c c c }
    \hline
       \multirow{2}{*}{\textbf{Method}}  & \multicolumn{2}{c}{$\hat{K}^w(x,\xi)$} & \multicolumn{2}{c}{$\hat{K}^v(x,\xi)$}\\
       \cline{2-5}
         & \textbf{\makecell{Max absolute error}} & \textbf{\makecell{Mean absolute error}} & \textbf{\makecell{Max absolute error}} & \textbf{\makecell{Mean absolute error}}\\
    \hline
    PINN-kernels & $7.833\times 10^{-4}$ & $1.551\times 10^{-4}$ & $1.109\times 10^{-3}$ & $1.525\times 10^{-4}$\\
    NO-kernels & $1.452\times 10^{-4}$ & $1.450\times 10^{-4}$ & $1.451\times 10^{-4}$ & $1.071\times 10^{-4}$ \\
    % \hline
    PINO-kernels & $6.684\times 10^{-4}$ & $1.232 \times 10^{-4}$ & $2.650\times10^{-3}$ & $1.090\times 10^{-4}$ \\
    \hline
    % \makecell{PINO-kernels\\(only physics loss)} & $3.371\times 10^{-3}$ & $1.261\times 10^{-3}$ & $6.804\times 10^{-2}$ & $1.465\times10^{-2}$\\
    % \hline
    \end{tabular}
    \label{error_tab_kernel_all}
\end{table}

{
In addition to the approximation error of backstepping kernels and traffic states, we also added three traffic performance indices to test the performance of the different methods, including fuel consumption, total travel time (TTT) and comfort value to compare the control performance introduced in~\citep{treiber_traffic_2013}. The definition of the performance indices are:
\begin{align}
    J_{\text{fuel}} &= \int_0^{T}\int_0^L \max\{0,b_0+ b_1 v(x,t) +b_2v(x,t)a(x,t) + b_3a^2(x,t)\} \rho(x,t) dx dt\\
    J_{\text{comfort}} &= \int_0^T \int_0^L (a^2(x,t) + a_t^2(x,t))\rho(x,t) dx dt\\
    J_{\text{TTT}} &= \int_0^T \int_0^L\rho(x,t) dx dt
\end{align}
where the coefficient of fuel consumption model is selected as $b_0=2.5\times10^{-3} \text{1/s}$, $b_1 = 2.45\times 10^{-7}$1/m, $b_2 = 1.25\times 10^{-8} s^2/m^2$, $b_3 = 9.5\times 10^{-5} s^3/m^2$~\citep{ahn1998microscopic}. $a(x,t)$ denotes the local acceleration $a(x,t) = v_t(x,t) + v(x,t) v_x(x,t)$. The results of different NO-based methods are shown in Tab.~\ref{quanindices}.
\begin{table}[htbp]
    \centering
    \caption{The results of different methods with traffic performance indices}
    \begin{tabular}{c c c c}
    \hline
      \textbf{Method}   &  \textbf{Fuel consumption} ($\downarrow$) & \textbf{Driving discomfort} ($\downarrow$) & \textbf{Total travel time} ($\downarrow$) \\
      \hline
        PINN-kernels & +1.04\% & -1.05\% & +1.04\%\\
        NO-kernels & +1.07\% & -23.21\% & +1.07\%\\
        PINO-kernels & +1.06\% & -29.59\% & +1.06\%\\
        NO-control law & +1.08\% & +63.21\% & +1.08\%\\
    \hline
    \end{tabular}
    \label{quanindices}
\end{table}
The backstepping method is chosen as the baseline for evaluating the indices. The performance of the backstepping method is compared against three other approaches. The analysis reveals that while the NO-approximated and PINO-approximated kernels result in higher fuel consumption and increased total travel time, they significantly enhance driving comfort. Specifically, an improvement of nearly 30\% in driving comfort is achieved at the cost of a 1\% increase in fuel consumption and total travel time, which is considered acceptable. However, the NO-approximated control law leads to a substantial 60\% decrease in driving comfort. This decrease in comfort, characterized by frequent small oscillations across the spatial-temporal domain, results in more abrupt changes in local acceleration, thereby diminishing the driving experience and potentially increasing the risk of traffic accidents.
}

\begin{figure}[htbp]
    \centering
    \subfigure[Backstepping]{\includegraphics[width=0.32\textwidth]{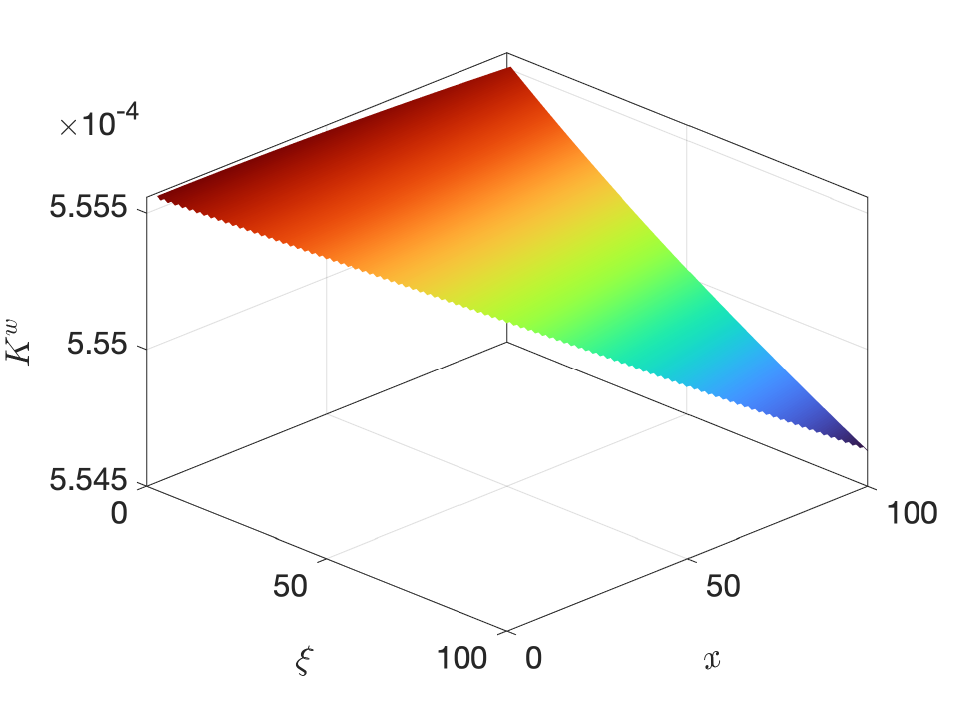}}
    \subfigure[NO-kernel]{\includegraphics[width=0.32\textwidth]{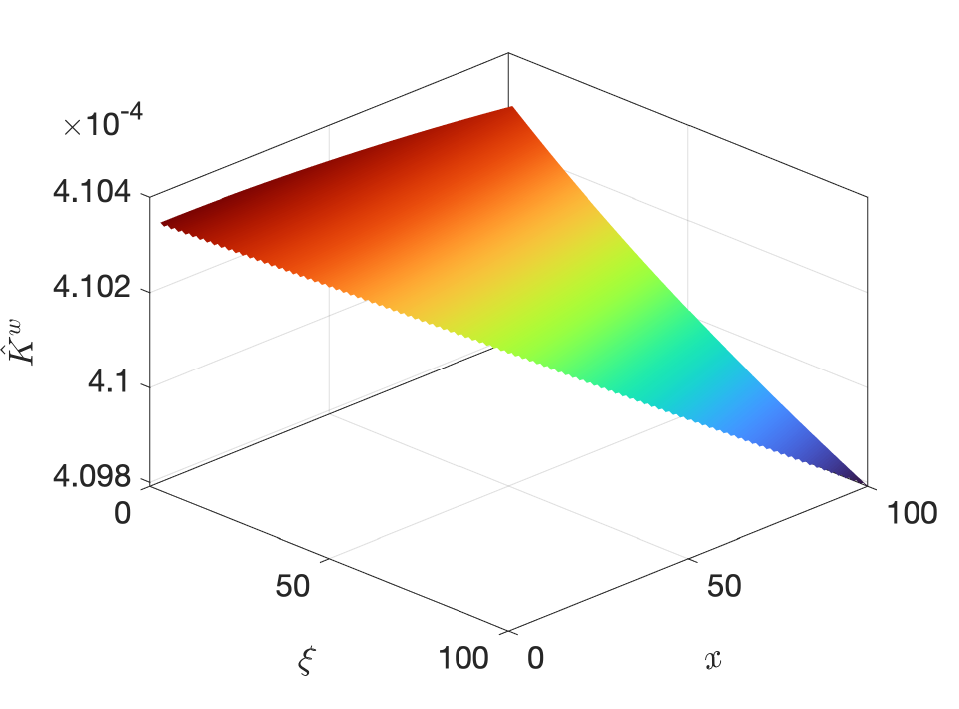}}
    \subfigure[PINO-kernel]{\includegraphics[width=0.32\textwidth]{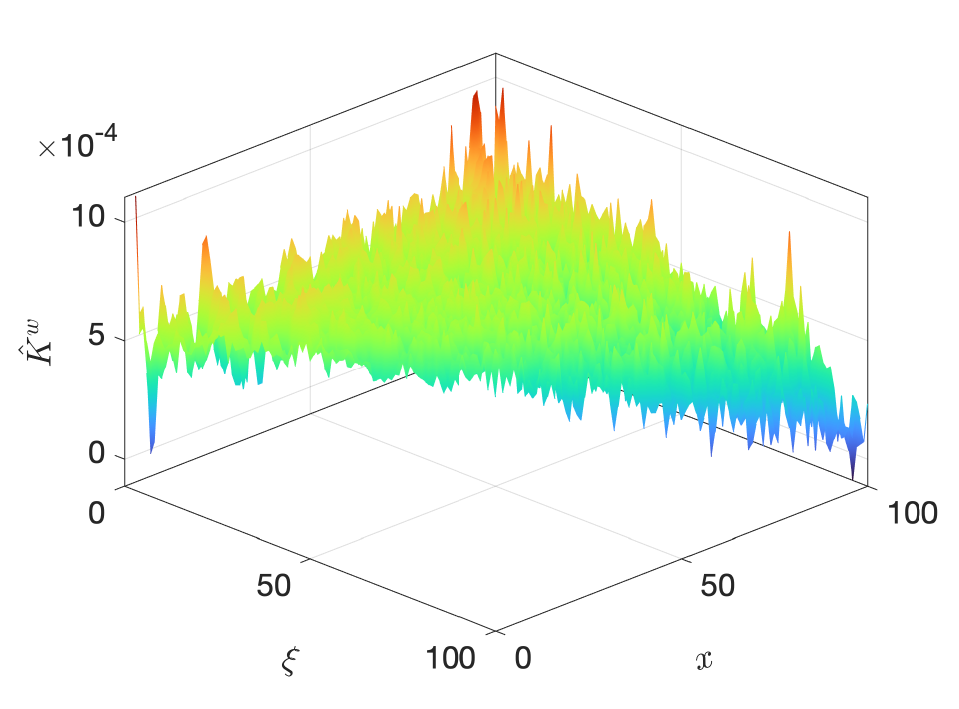}}\\
    \subfigure[Backstepping]{\includegraphics[width=0.32\textwidth]{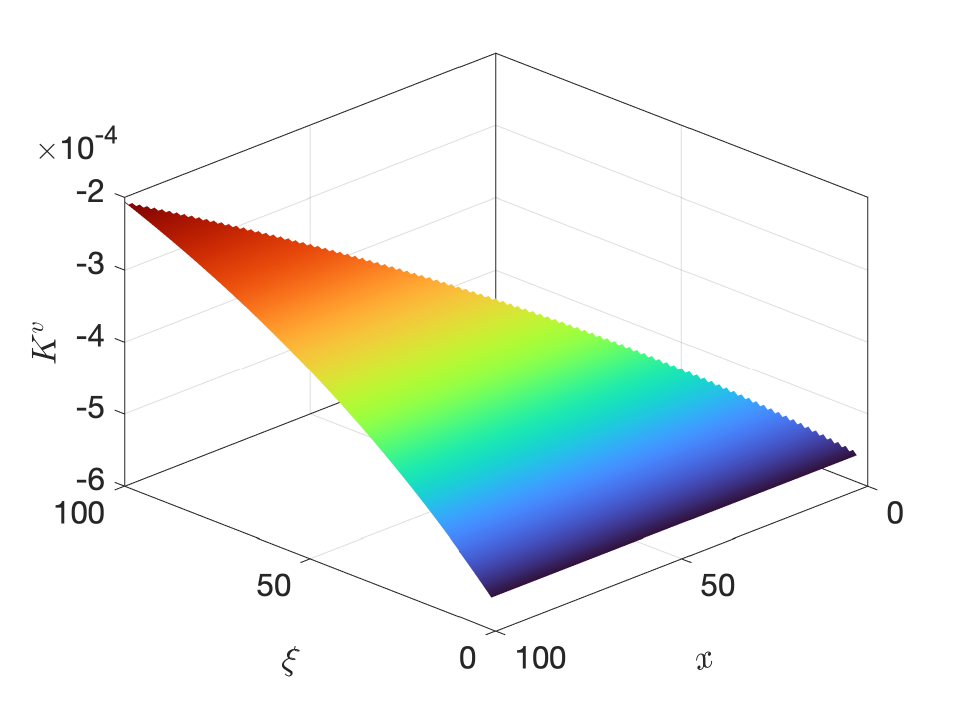}}
    \subfigure[NO-kernel]{\includegraphics[width=0.32\textwidth]{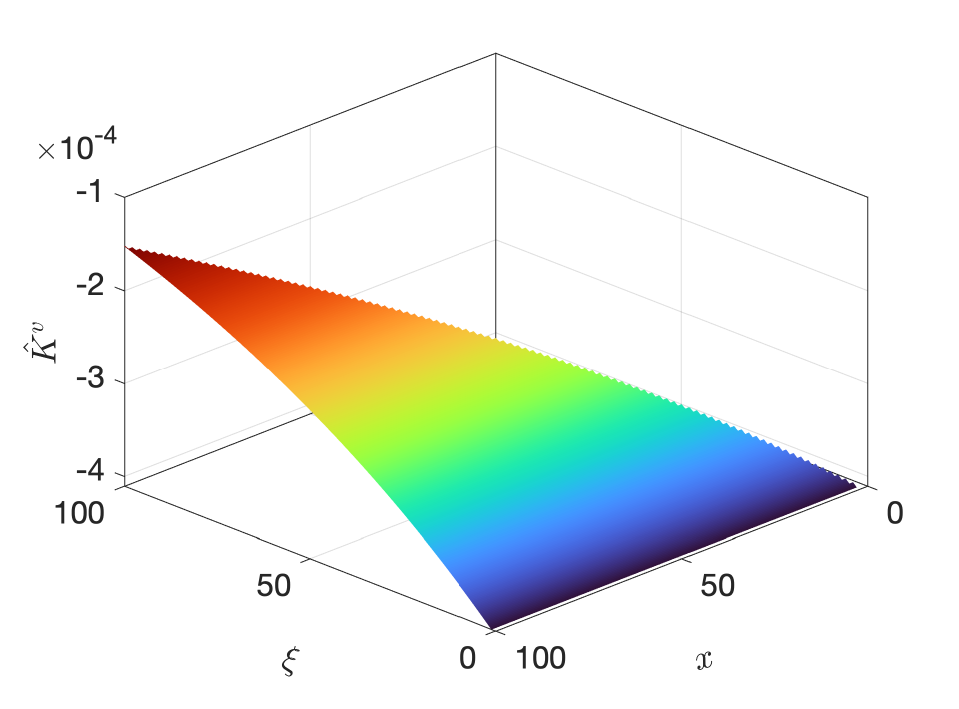}}
    \subfigure[PINO-kernel]{\includegraphics[width=0.32\textwidth]{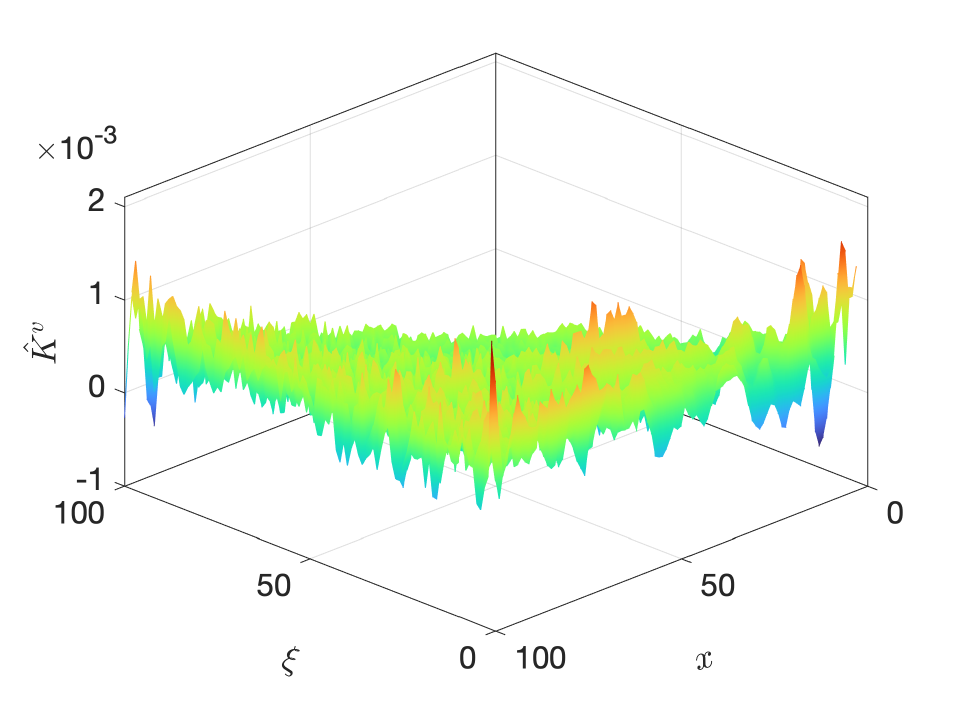}}\\
    \caption{Backstepping kernels under different schemes}
    \label{kernel_all}
\end{figure}

\begin{figure}[htbp]
    \centering
    \subfigure[NO-kernel]{\includegraphics[width=0.45\textwidth]{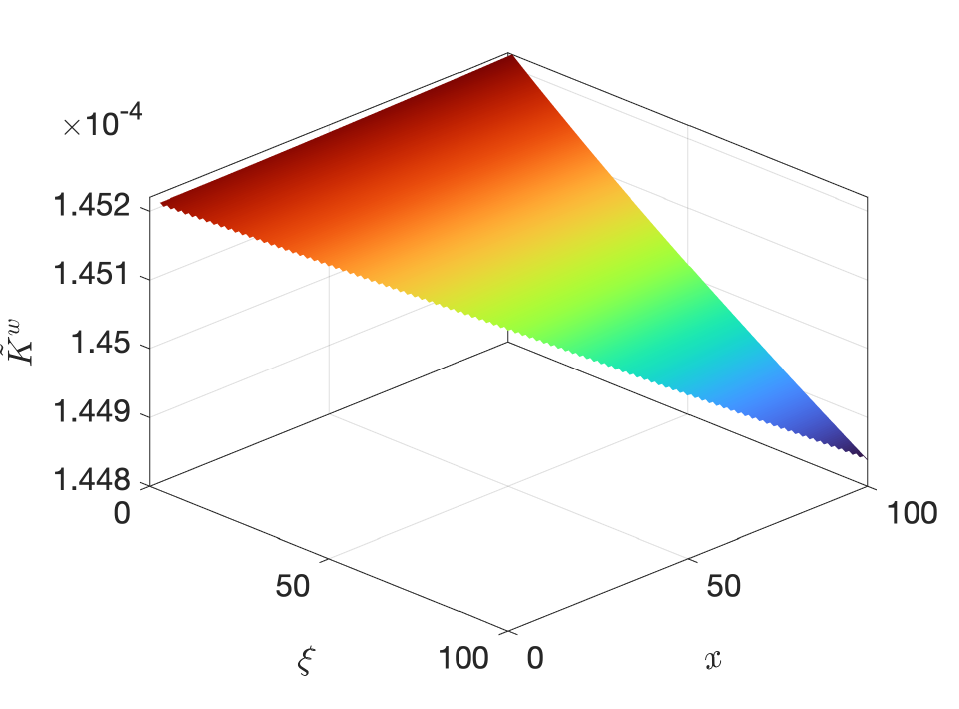}}
    \subfigure[PINO-kernel]{\includegraphics[width=0.45\textwidth]{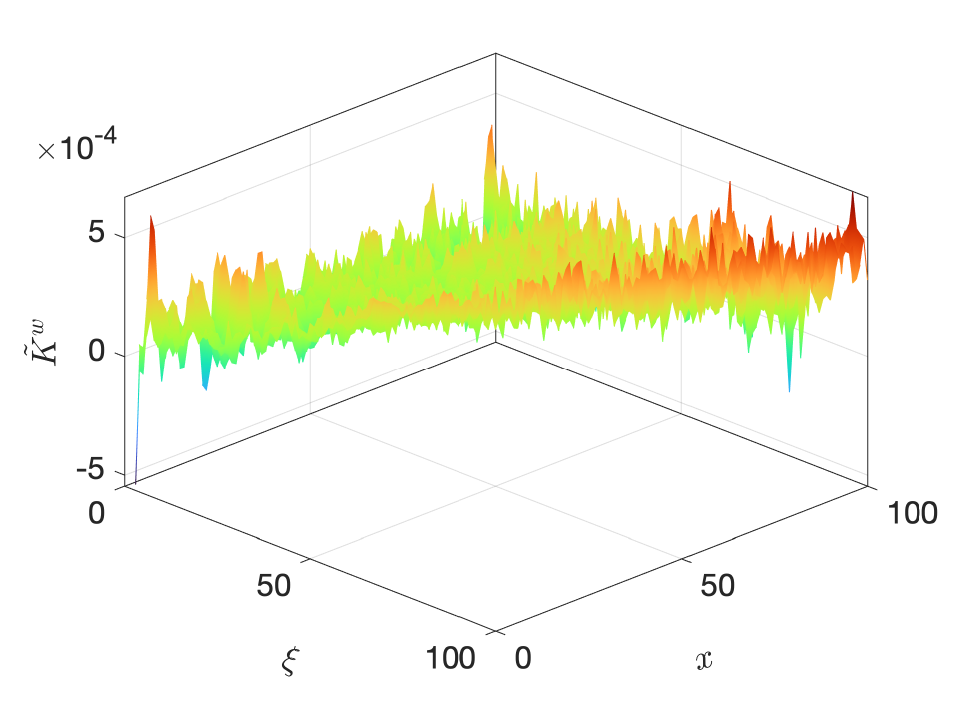}}\\
    \subfigure[NO-kernel]{\includegraphics[width=0.45\textwidth]{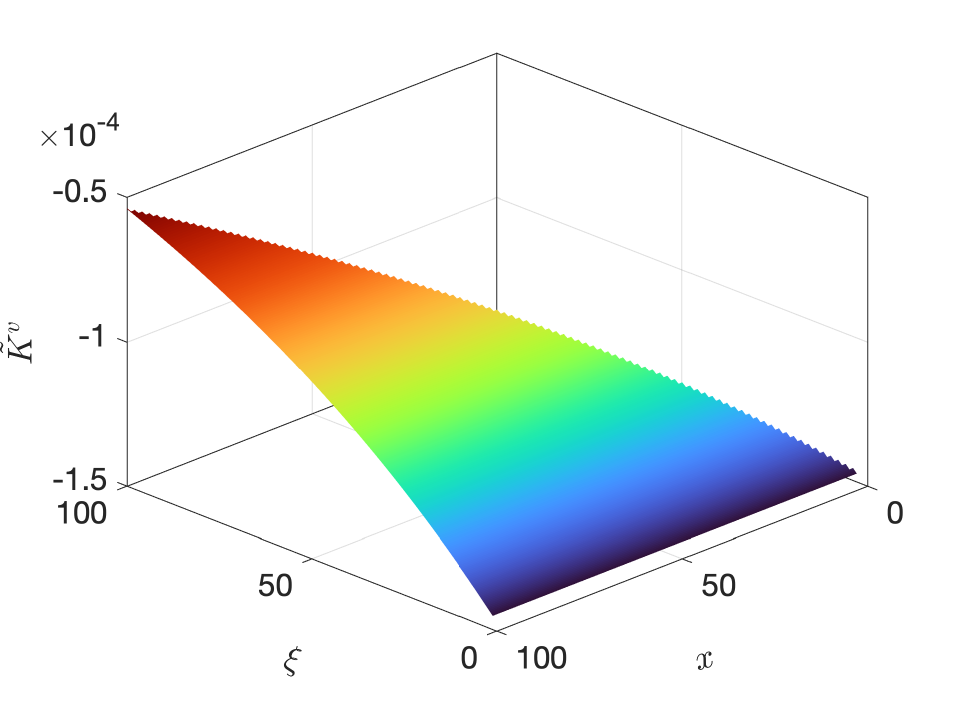}}
    \subfigure[PINO-kernel]{\includegraphics[width=0.45\textwidth]{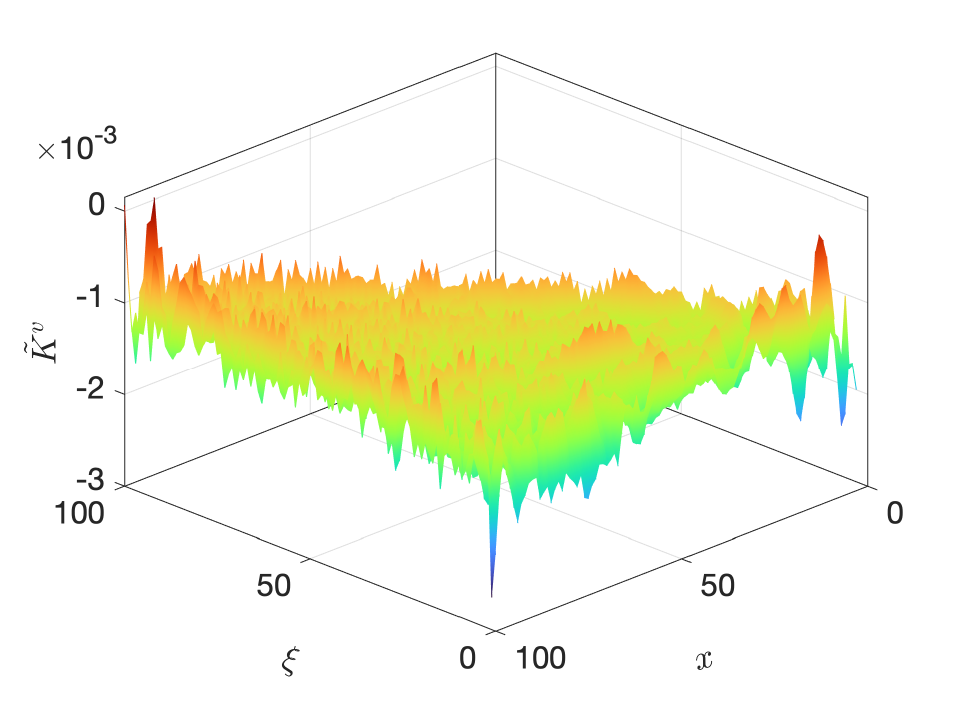}}\\
    \caption{Errors of backstepping kernels under different schemes}
    \label{err_kernel_all}
\end{figure}

\begin{figure}[htbp]
\centering
\subfigure[Control input $U(t)$]{\includegraphics[width=0.45\textwidth]{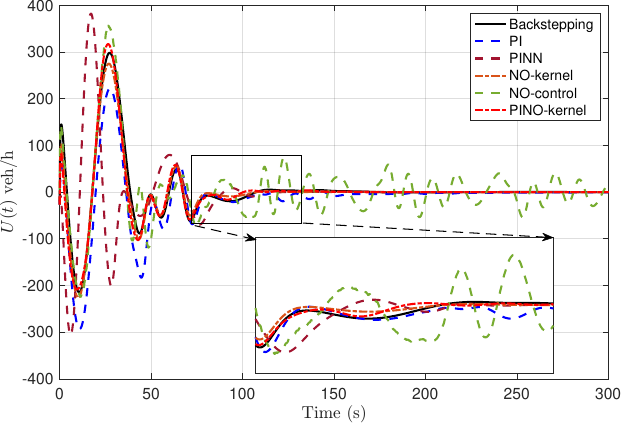}}
\subfigure[Norm of states]{\includegraphics[width=0.43\textwidth]{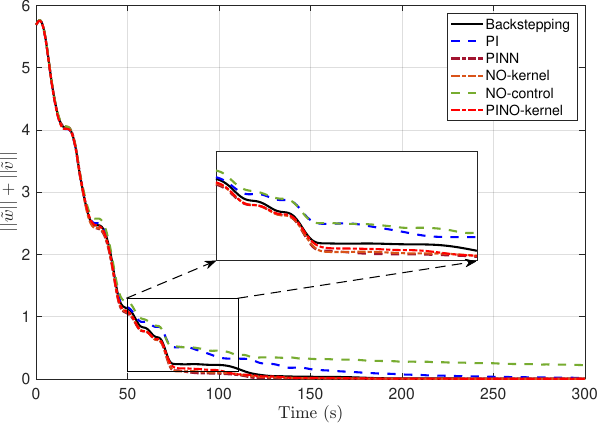}}
\caption{Comparison of $U(t)$ and the norm of states }
\label{compare-U}
\end{figure}

The comparisons of control law and norm of states are shown in Fig.~\ref{compare-U}. Six methods mentioned in this paper are considered. From the results of $U(t)$, the NO-based methods can approximate the backstepping control law well. All the controllers are eventually stabilizing the system. However, the norm of the states of the NO-approximated control law converges to zero slower than other controllers because we only get the practical stability results for the traffic system. Overall, it can be found that the NO-based and PINO-based methods achieve satisfactory closed-loop results.

The computation time of the neural operator, the backstepping controller, the PI controller and PINN-based controller are shown in Tab.~\ref{Tab-res-1}. We set the backstepping control method as the baseline of the system. The Mean Squared Error(MSE) of the traffic system is calculated by
\begin{align}
    \text{MSE}_\rho &= \frac{1}{N} \sum_{(x,t)\in \mathcal{D}} \left(\frac{\rho(x,t) -\hat{\rho}(x,t)}{\rho^\star}\right)^2\\
    \text{MSE}_v &= \frac{1}{N} \sum_{(x,t)\in \mathcal{D}}  \left( \frac{v(x,t) - \hat{v}(x,t)}{v^\star} \right)^2
\end{align}
where $\hat{\rho}(x,t)$ and $\hat{v}(x,t)$ are the density and speed generated by the NO-based methods. $\rho(x,t)$ and $v(x,t)$ are the density and speed of backstepping method. $N$ is the sampled points by numerical methods such as the Godunov scheme~\citep{godunov1959finite} in the spatial-temporal domain $\mathcal{D} = [0,L]\times [0,T]$. It indicates the average approximation error during the simulation period. It can be observed that the average computation times of the NO-approximated methods are $298$ times faster than the backstepping controller and $10 $ times faster than the PI controller with a loss of accuracy less than $1\%$, giving the possibility of accelerating the online application in a real traffic system. 

{
The computation time of backstepping controller is fast in our simulations. However, this test is limited to a single road segment of 500 meters for 3 min. In real world traffic control problem, such as the cascaded freeway segments~\citep{yu_simultaneous_2022} described by two sets of ARZ models, multi-class traffic~\citep{burkhardt_stop-and-go_2021} or mixed-autonomy traffic~\citep{zhang2023mean}, the computation time will increase. Extending the ARZ model to encompass entire traffic networks would further amplify the computational burden and time requirements. Therefore, NO-based methods offer a significant advantage in large-scale traffic scenarios, as their computation time is two orders of magnitude shorter than that of the backstepping method, thereby reducing the overall computational burden.
}

{
\textbf{Summary of simulation results.}  
The simulation results highlight the performance of the backstepping method, PI controller, PINN-based controller, and NO-based methods. All approaches successfully stabilize sinusoidal traffic waves within a finite time, except for the NO-approximated controller, which only achieves practical exponential stability. Among these approaches, the NO-approximated kernels exhibits the smallest density and speed errors during the simulation period and demonstrate the fastest computation speed, averaging just $1.997 \times 10^{-4}$ seconds per time step. Both the NO-approximated and PINO-approximated kernel methods significantly enhance driving comfort with nearly equivalent cost of fuel consumption and total travel time compared to the PINN. From the perspective of computational efficiency, all NO-based methods outperform the other approaches, with the potential to be up to 298 times faster than the backstepping controller, thus facilitating the acceleration of online control applications in real traffic scenarios.
}

\begin{table}[htbp]
    \centering
    \caption{The average computation time and MSE of different schemes}
    \begin{tabular}{c c c c }
    \hline
       \multirow{2}{*}{\textbf{Method}}  &  \multirow{2}{*}{\textbf{Average computation time(s)}} & \multicolumn{2}{c}{\textbf{MSE \%}} \\
       \cline{3-4}
        & &density $\rho$ & velocity $v$\\
       \hline
       Backstepping controller  & $5.960\times 10^{-2}$ & / & /\\
       % \hline
       PI controller &  $1.901\times 10^{-3}$ (31x) & $6.973\times 10^{-3}$ & $1.046\times 10^{-2}$\\
       % \hline
       PINN & $2.466\times 10^{-3}$ (24x) & $2.466\times 10^{-3}$ & $7.523 \times 10^{-3}$\\
       NO-approximated kernels &  $1.997\times 10^{-4}$ (298x) & $3.783 \times 10^{-4}$ & $8.578\times 10^{-4}$\\
       % \hline
       PINO-approximated kernels & $7.854\times10^{-4}$ (75x)  & $5.861\times 10^{-4}$ & $1.593\times 10^{-3}$\\
        NO-approximated control law & $1.701\times 10^{-3}$ (75x) & $3.821\times 10^{-3}$ & $1.274\times 10^{-3}$\\
       % \hline
       \hline
       % PINO-approximated kernels (only physics loss) &  $8.257\times 10^{-4}$ &  $6.9911$ \\
       % \hline
    \end{tabular}
    \label{Tab-res-1}
\end{table}

{
\subsection{Experiments with different demands and conditions}
It is revealed that the developed NO-based method performs well in mitigating the stop-and-go traffic in the previous section. It is also needed to mention that whether the proposed NO-based methods could still stabilize traffic for different scenarios. Next, we will test the trained NO in different traffic conditions. 
}

{
In the data collection and training phases, sinusoidal initial conditions are utilized to simulate stop-and-go traffic waves. However, real-world traffic conditions encompass a variety of scenarios, such as varying traffic demands and non-recurrent traffic conditions. Therefore, it is essential to evaluate the performance of NO-based methods under various initial traffic conditions.
To address this, we first assess the performance of the NO method under three distinct traffic demand levels: high, medium, and low. The inflow demand $q^\star$ and the corresponding equilibrium density and speed for each demand level are detailed in Table~\ref{differendemand}.
\begin{table}[htbp]
    \centering
    \caption{The different demand level with equilibrium density and speed}
    \begin{tabular}{c c c }
    \hline
       \textbf{In-flow demand} $q^\star$(veh/h)  &  \textbf{Equilibrium density} $\rho^\star$ (veh/km) &  \textbf{Equlibrium speed} $v^\star$ (km/h) \\
       \hline
     2025 (\text{High demand})& 100 & 20.25 \\
     1856 (\text{Medium demand})& 110 & 16.87\\
     1620 (\text{Low demand})& 120 & 13.5 \\
    \hline
    \end{tabular}
    \label{differendemand}
\end{table}
By applying different demand levels $q^\star$ to the inlet of the road segment, the results for different demand levels are illustrated in Fig.~\ref{densitydifferentdeman} and Fig.~\ref{velocitydifferentdeman}.
    \begin{figure}
        \centering
        \subfigure[$q^\star$=2025]{\includegraphics[width=0.32\linewidth]{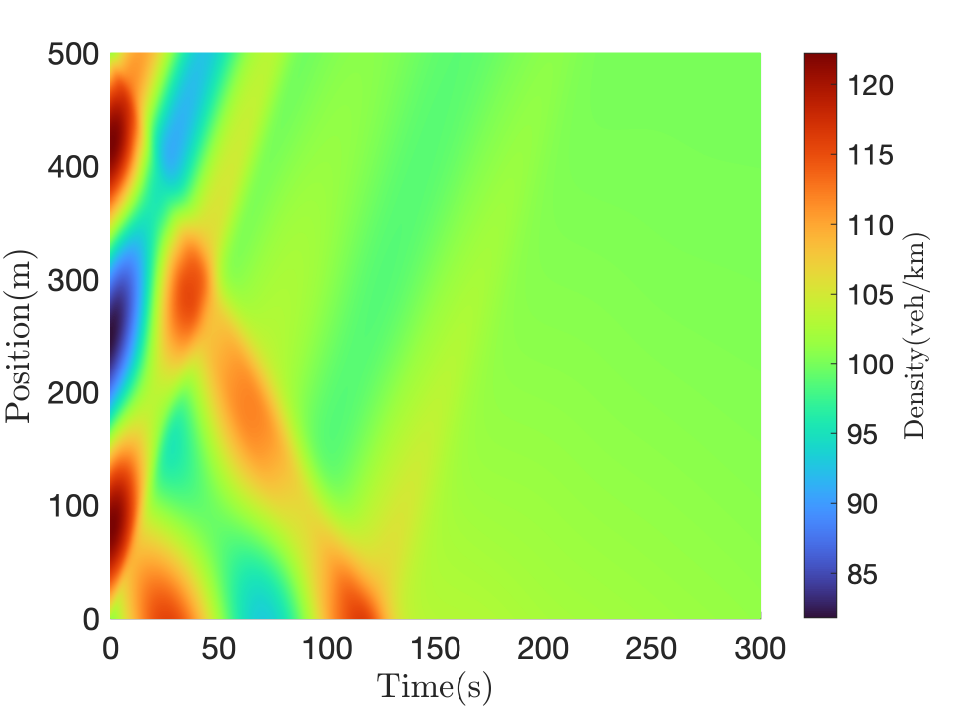}}
        \subfigure[$q^\star$=1856]{\includegraphics[width= 0.32\linewidth]{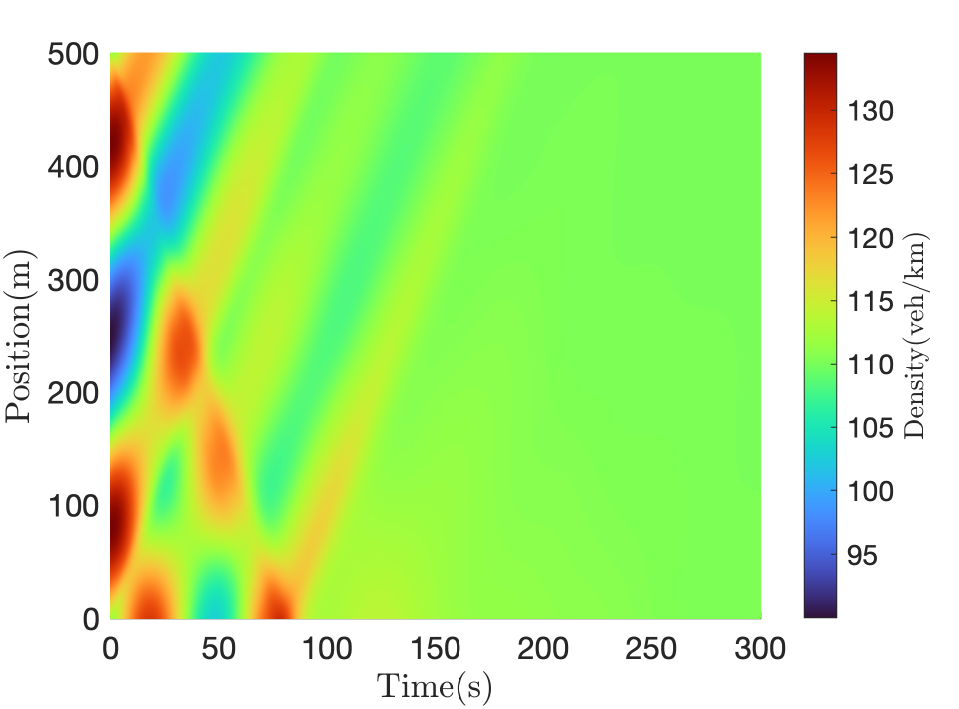}}
        \subfigure[$q^\star$=1620]{\includegraphics[width= 0.32\linewidth]{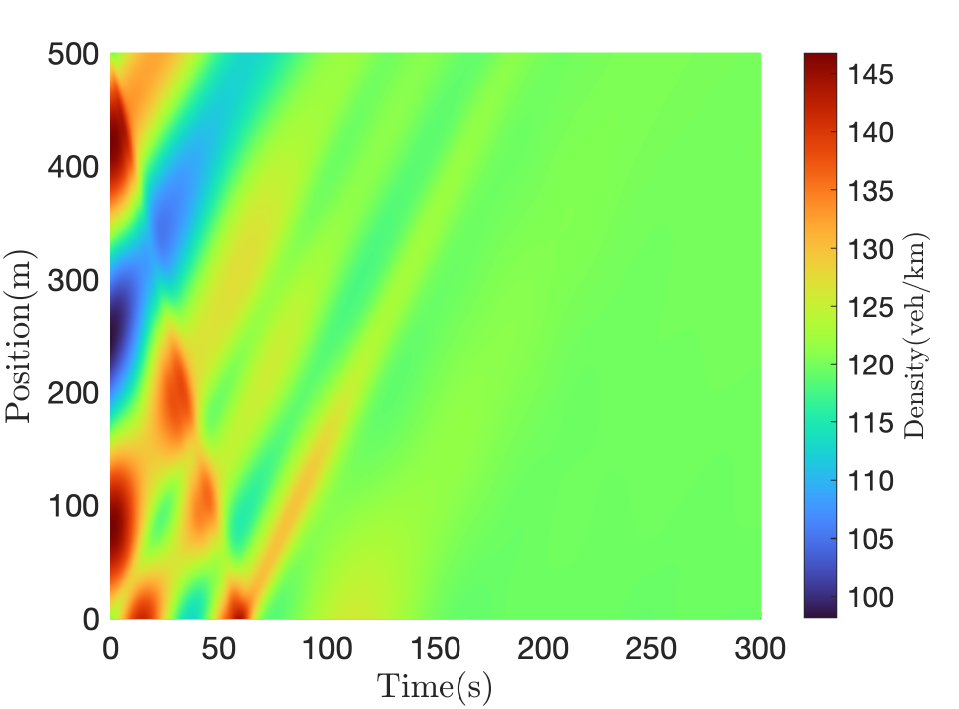}}
        \caption{Density evolution of different demand level}
        \label{densitydifferentdeman}
    \end{figure}
    \begin{figure}
        \centering
        \subfigure[$q^\star$=2025]{\includegraphics[width=0.32\linewidth]{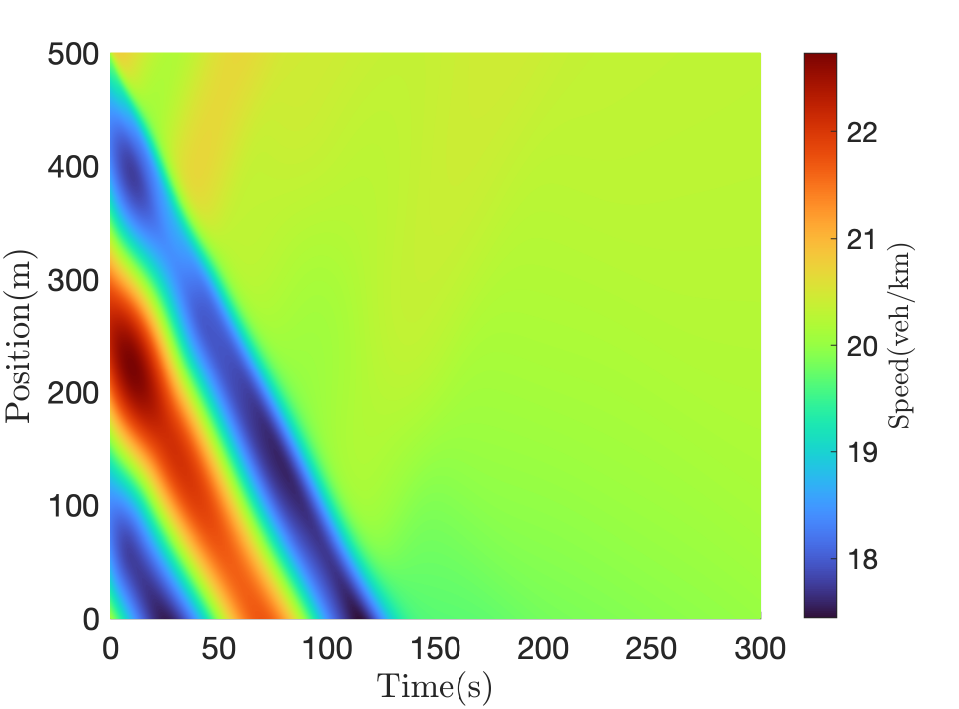}}
        \subfigure[$q^\star$=1856]{\includegraphics[width=0.32\linewidth]{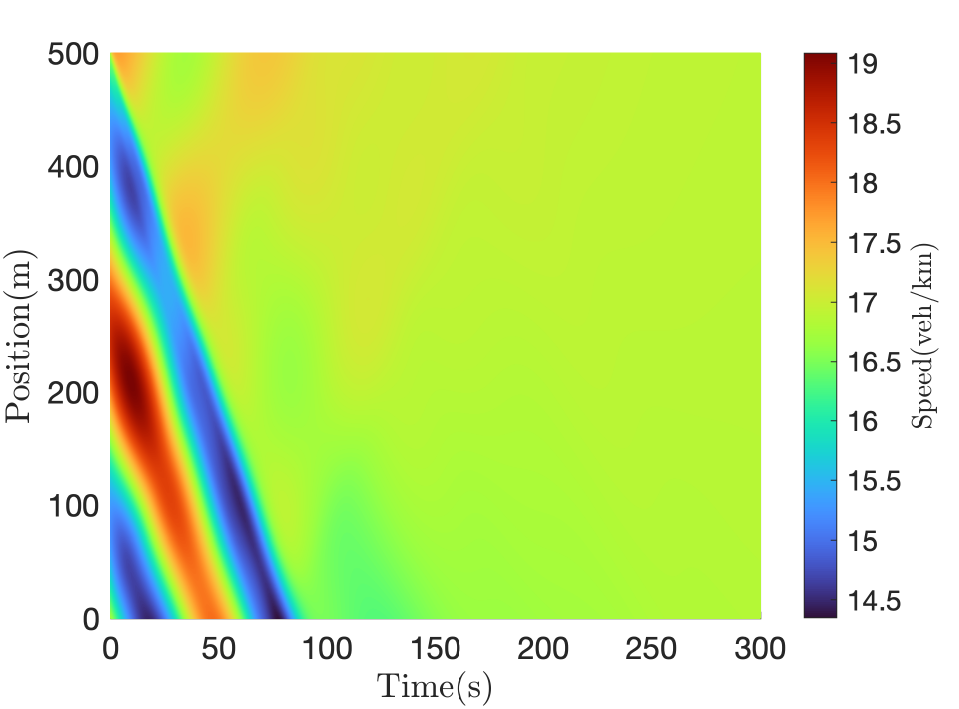}}
        \subfigure[$q^\star$=1620]{\includegraphics[width=0.32\linewidth]{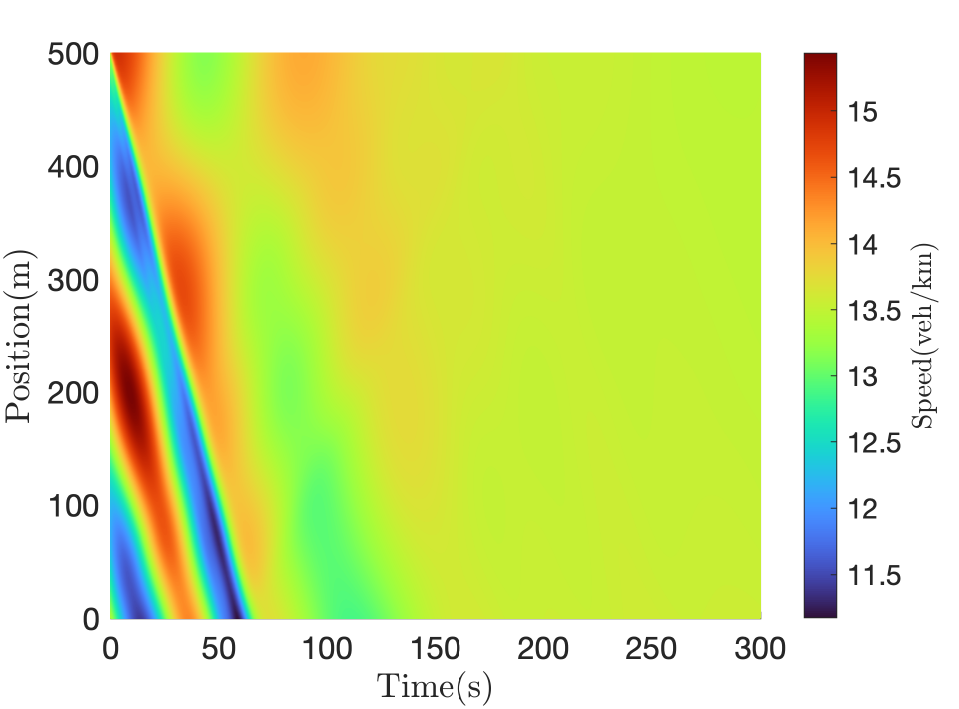}}
        \caption{Speed evolution of different demand level}
        \label{velocitydifferentdeman}
    \end{figure}
It is observed that the method remains effective in stabilizing the traffic system across different demand levels. Both traffic density and speed converge to their equilibrium points within a finite time.
}

{
Additionally, we evaluate the performance of the developed method under non-recurrent traffic conditions. Specifically, two types of non-recurrent traffic conditions are tested: one with sinusoidal initial conditions at a low frequency, which results in non-recurrent behavior, and another with linear initial conditions. For the first case, the initial conditions are defined as follows:
\begin{align}
    q(x,0) &= q^\star + 0.05 \sin(\frac{\pi x}{L}) q^\star,\\
    v(x,0) &= v^\star -0.05 \sin(\frac{\pi x}{L}) v^\star.
\end{align}
Using the relationship $q = \rho \times v$, we can determine the initial condition of traffic density. This initial condition is designed to replicate a sudden deceleration occurring in the middle of the road segment, resulting in a density wave propagating from upstream to downstream, while a speed wave moves from downstream to upstream. Consequently, the density will initially increase and then decrease, whereas the speed will decrease and subsequently increase. The initial conditions for both density and speed are depicted in Figure~\ref{IC_sin}.
\begin{figure}
    \centering
    \subfigure[Density]{\includegraphics[width= 0.45\linewidth]{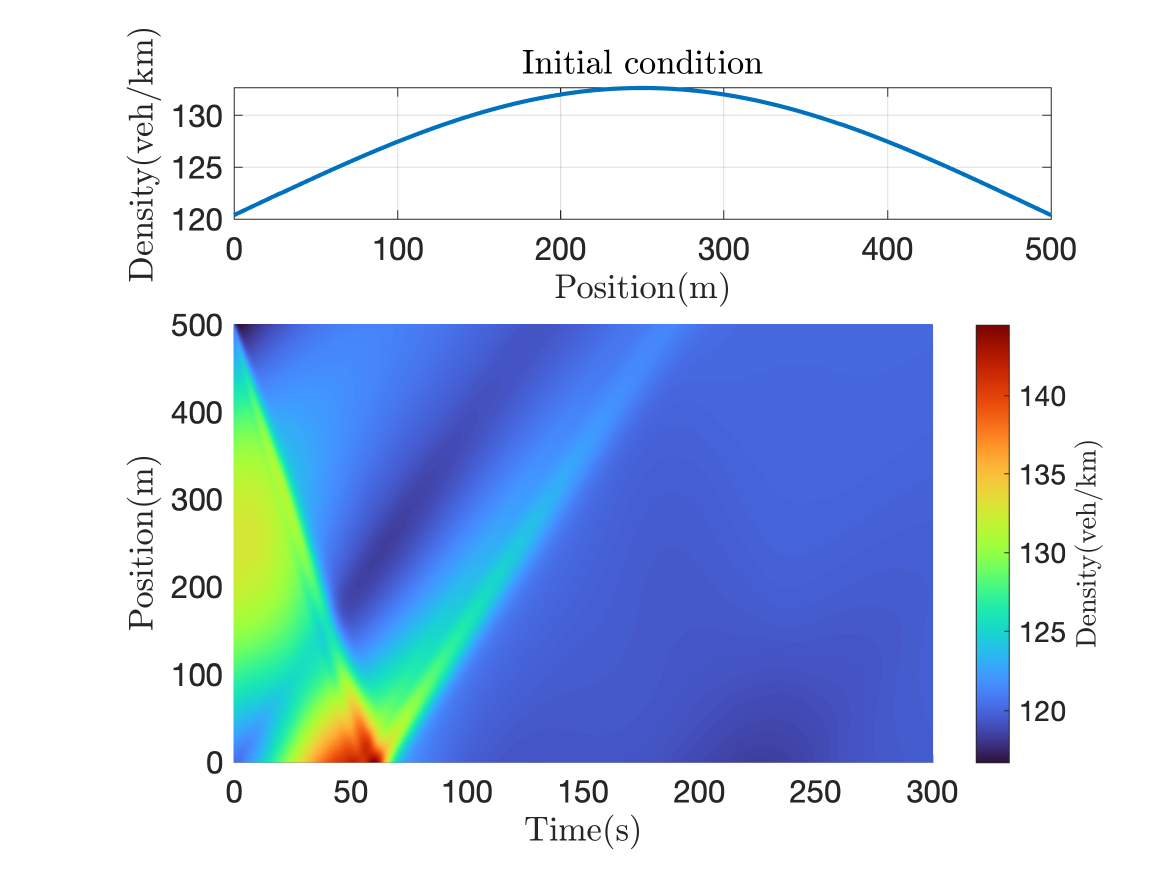}}
    \subfigure[Speed]{\includegraphics[width= 0.45\linewidth]{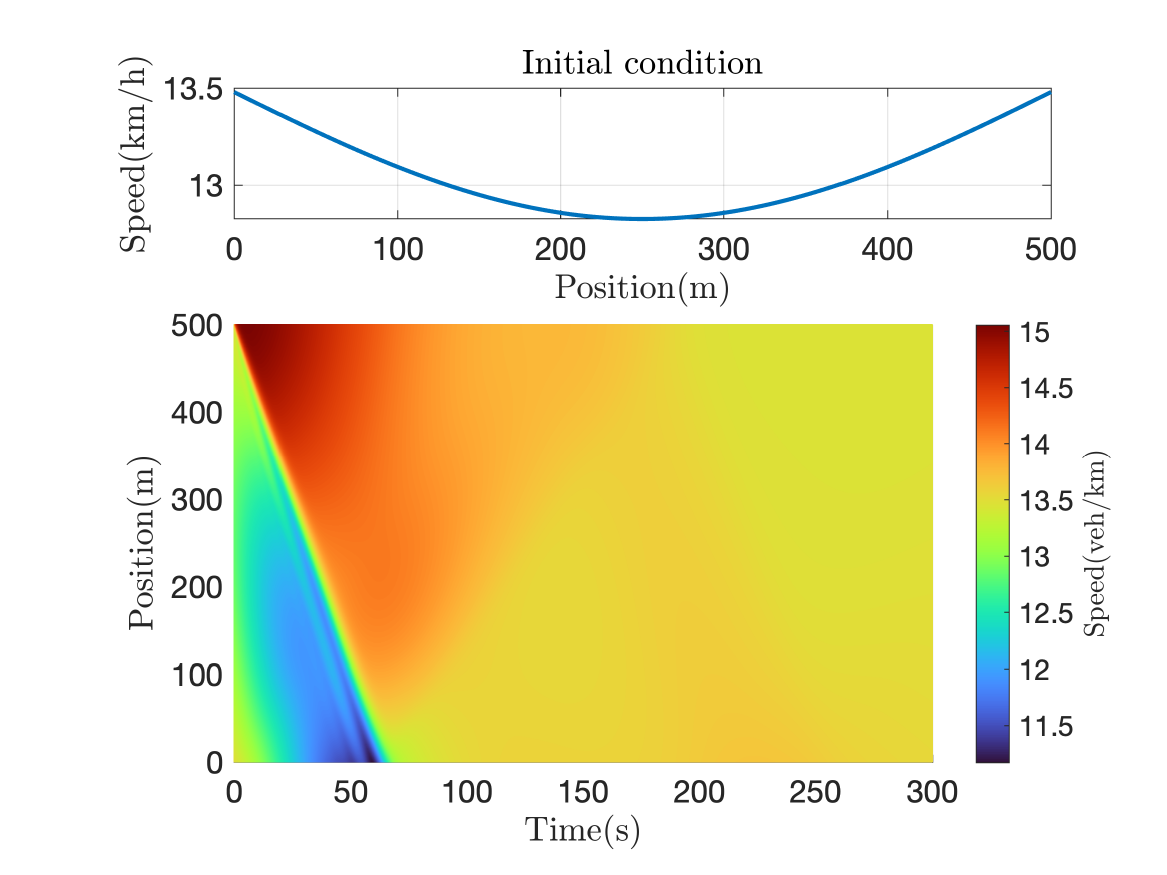}}
    \caption{Initial conditions and results of density and speed}
    \label{IC_sin}
\end{figure}
By applying the NO-approximated kernels to these initial conditions of density and speed, the evolution of traffic states is obtained and illustrated in the bottom of Figure~\ref{IC_sin}.
Then, linear initial conditions are applied to simulate a scenario in which vehicles decelerate downstream due to lane closure. This scenario leads to a decrease in traffic speed and an increase in traffic density due to vehicle accumulation. The linear initial conditions are defined as follows:
\begin{align}
    q(x,0) &= 9\times 10^{-5} x,\\
    v(x,0) &= -7\times 10^{-4} x + 0.375.
\end{align}
Using the flow-density relation, we obtain the initial condition of traffic density and speed, as shown in Fig.~\ref{IC_linear}.
\begin{figure}
\centering
\subfigure[Density]{\includegraphics[width= 0.45\linewidth]{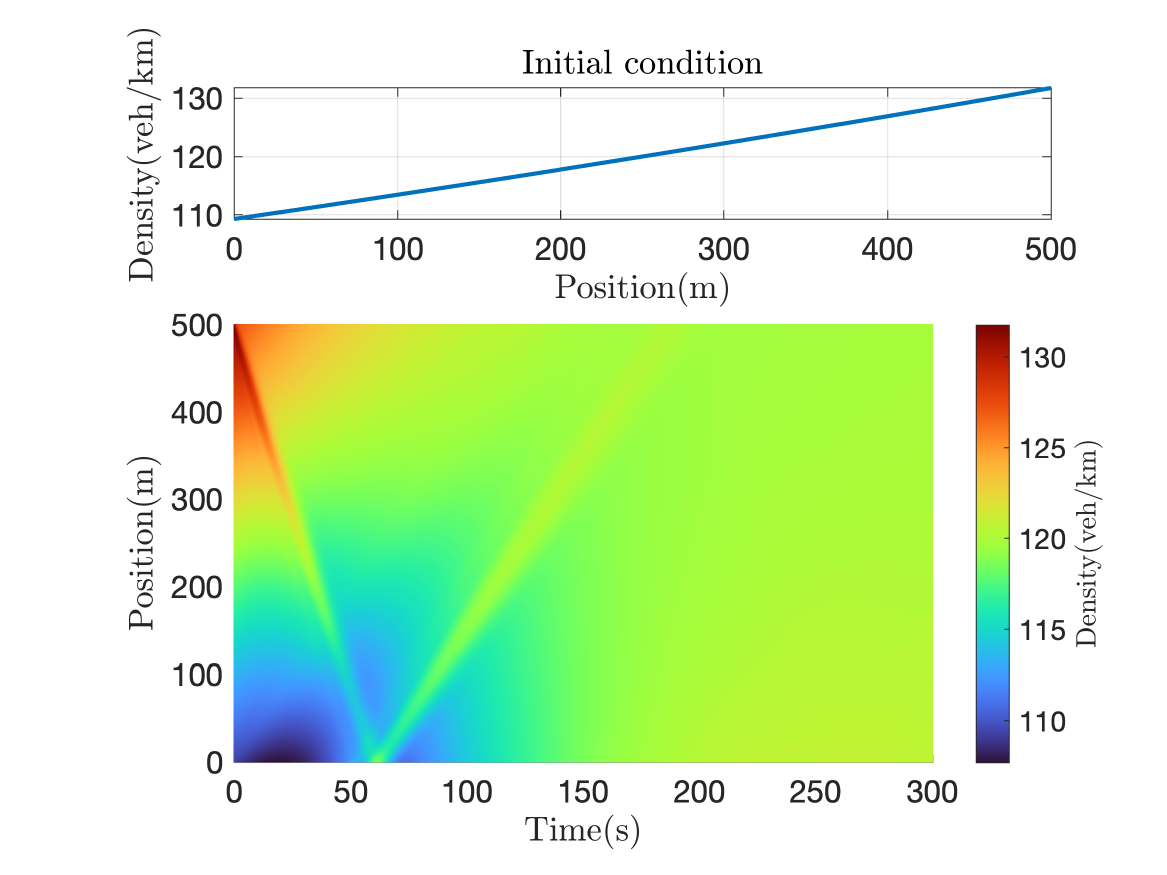}}
\subfigure[Speed]{\includegraphics[width= 0.45\linewidth]{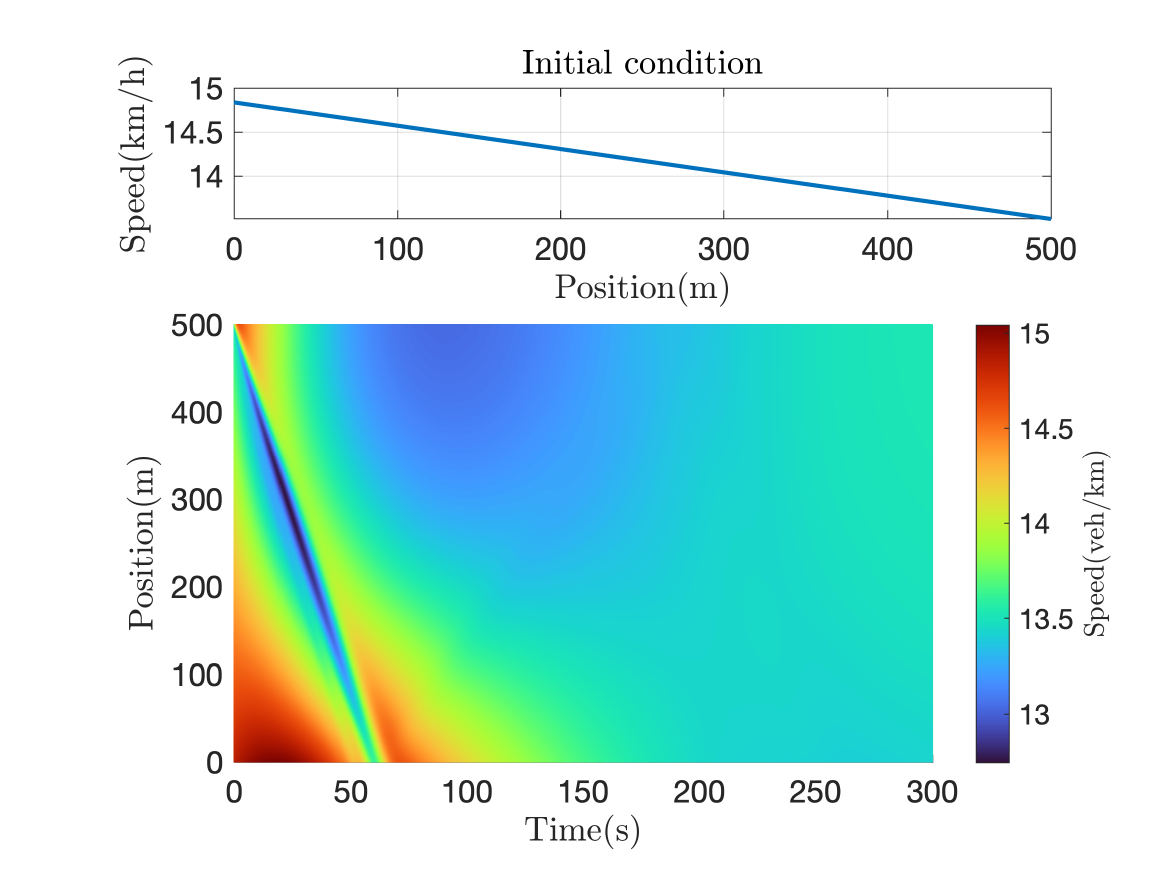}}
\caption{Initial conditions and results of density and speed}
\label{IC_linear}
\end{figure}
Applying the NO method to the linear initial conditions, the evolution of traffic states with NO-approximated kernels is depicted in the bottom of Figure~\ref{IC_linear}. The results demonstrate that, for both types of initial conditions, our method effectively stabilizes the traffic system. Traffic density and speed converge to their equilibrium points, indicating that the developed NO method is robust to variations in traffic conditions.
}

{
\subsection{Experiments with real traffic data}
In the previous section, we evaluated the NO-based methods within a simulation environment. We now proceed to apply the trained neural operator to the ARZ system with a calibrated fundamental diagram using real traffic data. Specifically, we utilize the NGSIM dataset, which includes vehicle trajectory data collected in Emeryville, California, USA, on April 13, 2005. The dataset is segmented into 15-minute intervals, and we select data from 4:00 pm to 4:15 pm for calibrating the fundamental diagram. The chosen spatial-temporal domain is $\text{L}=500$m and $\text{T}=700$s. To derive the density and flow from the trajectory data, we reconstruct the NGSIM data using Eide's formula~\citep{edie1963discussion} and the same method~\citep{zhao2023observer,fan_data-fitted_2013,yu_pde_2021}. The the spatial-temporal domain is partitioned into small grids with dimensions $\Delta x = 20 \text{m}$ and $\Delta t = 15 \text{s}$. 
}
{
We adopt the three-parameter $(\zeta,\kappa,p)$ fundamental diagram to calibrate the NGSIM dataset, the three-parameter fundamental diagram is denoted by the following form
\begin{align}
    Q(\rho) = \zeta \left( a +\frac{(b-a)\rho}{\rho_m}  - \sqrt{1 + \kappa^2\left( \frac{\rho}{\rho_m} - p \right)^2}\right),
\end{align}
where $a = \sqrt{1 + \kappa^2 p^2}$, $b = \sqrt{1 + \kappa^2(1-p)^2}$. The maximum density $\rho_m$ is defined by
\begin{align}
    \rho_m = \frac{\text{number of lanes}}{\text{vehicle length}\times \text{safety factor}},
\end{align}
where the vehicle length is $5$m and the safety factor is selected as $1.5$. There are 6 lanes in the road segment of the NGSIM dataset, therefore, the maximum density is $\rho_m = 800 \text{veh/km}$. The calibrated fundamental diagram is shown in Fig.~\ref{calfunall}(a), the calibrated three parameters are $\zeta = 1339.38$, $\kappa = 16.53$, $p = 0.28$. And the corresponding density-velocity relation is illustrated in Fig.~\ref{calfunall}(b).
\begin{figure}
    \centering
    \subfigure[Flow-density relation \label{calfun}]{\includegraphics[width=0.45\linewidth]{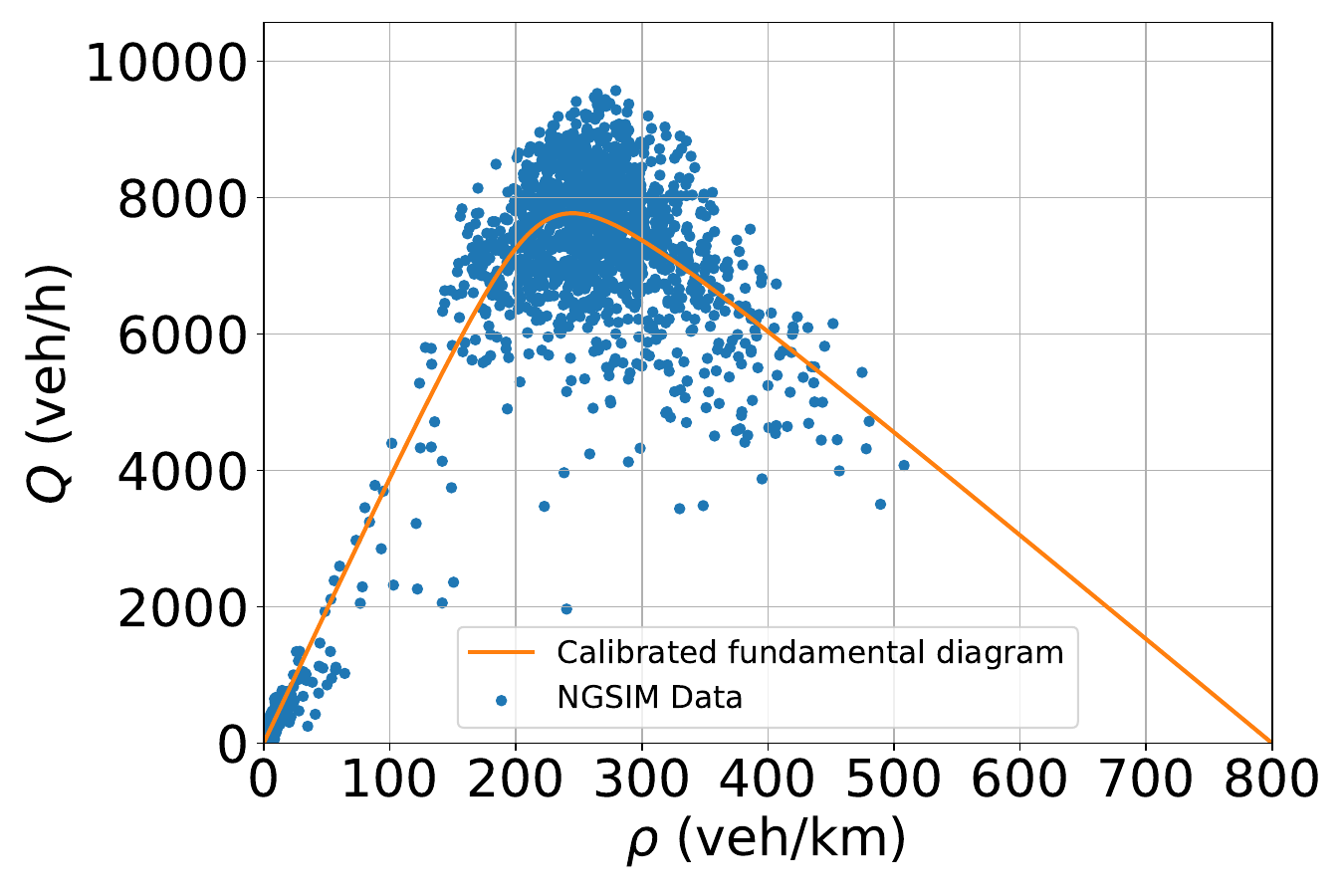}}
    \subfigure[Density-velocity relation \label{calfunV}]{\includegraphics[width=0.45\linewidth]{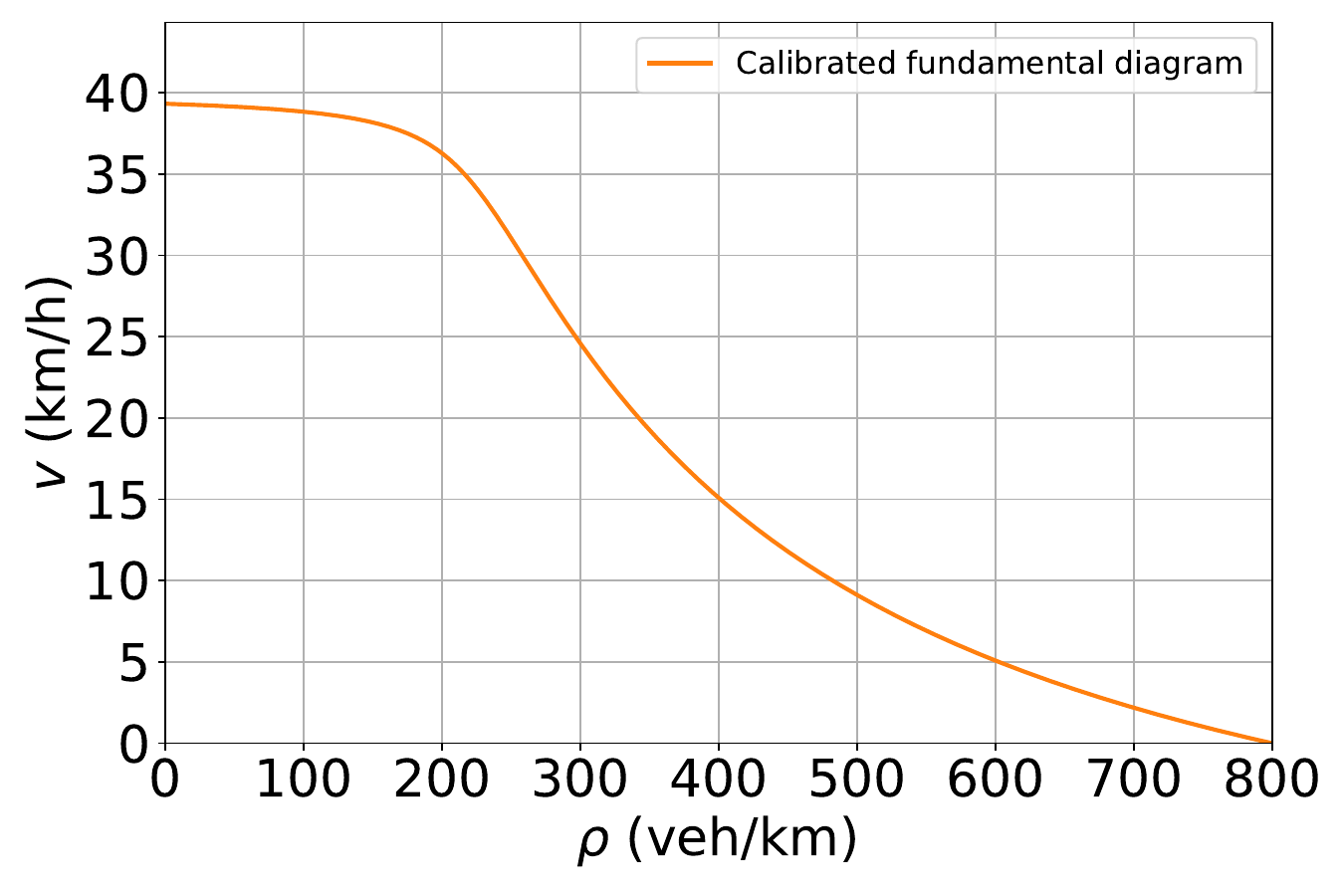}}
    \caption{The calibrated fundamental diagram from NGSIM dataset}
    \label{calfunall}
\end{figure}
}

{
The fundamental diagram and relaxation parameter are calibrated for the ARZ model. The trained neural operator for backstepping kernels is then employed to derive the control law required to stabilize the calibrated ARZ system. For the calibrated ARZ system, the equilibrium density is set as $320 \text{veh/km}$ and the corresponding equilibrium speed is $22.3 \text{km/h}$, as calculated by the three-parameters fundamental diagram. Initial conditions are configured with sinusoidal inputs to simulate stop-and-go traffic, while all other settings remain consistent with previous experiments. The results for the evolution of traffic states are presented in Fig~\ref{NOstatesthree}.
\begin{figure}[!tbp]
    \centering
    \subfigure[Density \label{NOdensitythreedensity}]{\includegraphics[width=0.45\linewidth]{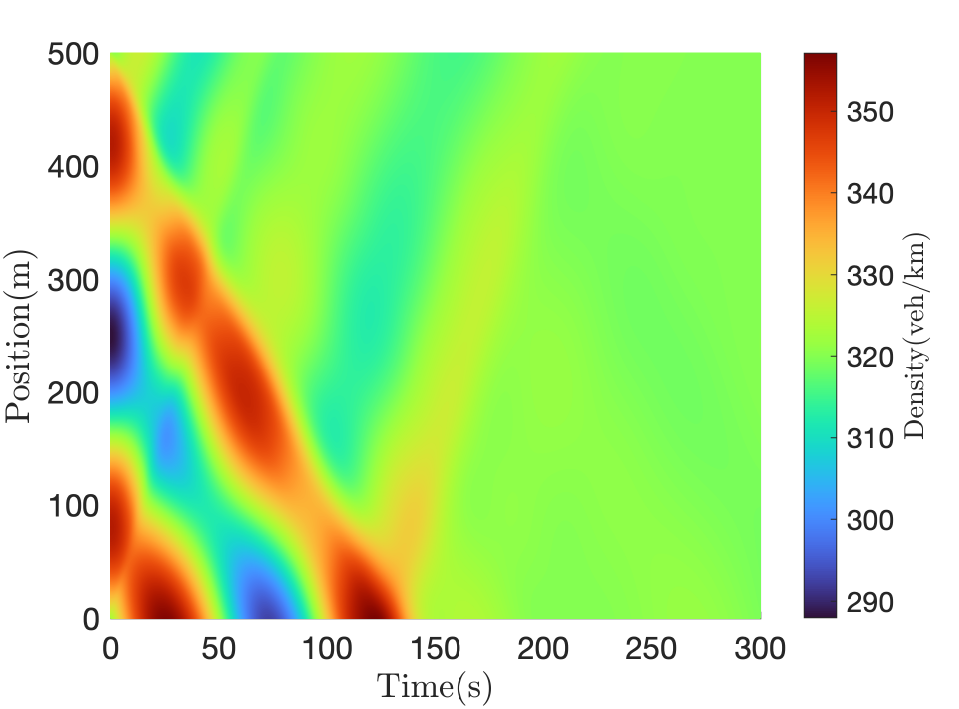}}
    \subfigure[Speed \label{NOdensitythreevelocity}]{\includegraphics[width=0.45\linewidth]{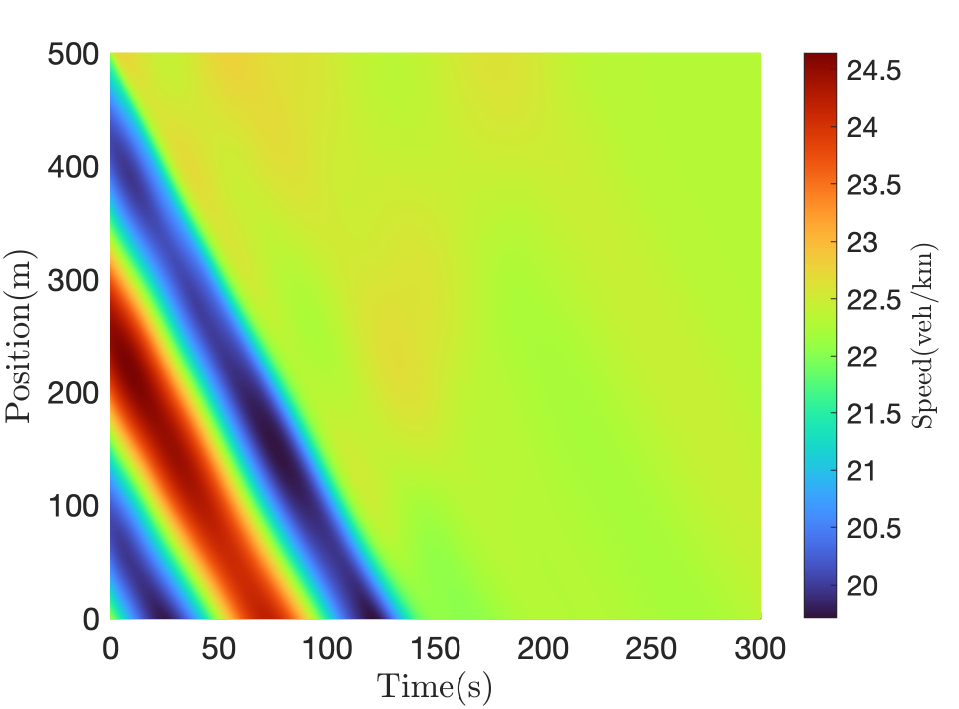}}
    \caption{NO method for calibrated ARZ model}
    \label{NOstatesthree}
\end{figure}
The experimental results demonstrate that the trained neural operator effectively stabilizes the traffic system with the calibrated fundamental diagram using real traffic data. Traffic states converge to their equilibrium points within a finite time.
}

\section{Conclusion}
In this paper, we propose an operator learning framework for the boundary control of traffic systems. First, the ARZ PDE model is adopted to describe the spatial-temporal evolution of traffic density and speed. We first define the operator mapping from the model parameter, i.e., characteristic speed $\lambda_2$ to the backstepping control kernels $K^w$,$K^v$ and then the directly map to a boundary control law $U(t)$. The neural operators using DeepONet are trained to approximate the two operator mappings. Subsequently, the Lyapunov analysis is conducted to derive the theoretical stability for the NO-approximated closed-loop system. To further extend our operator learning framework, we incorporate physical constraints into the neural operator, specifically the equations describing the mapping from the characteristic speed to backstepping kernels, and then the PINO is established.

{
The performance of the neural operator is assessed using both simulated and real traffic data. The NO-approximated kernels, NO-approximated control law, and PINO-approximated kernels are evaluated under consistent settings. The backstepping method is used as the baseline, with comparisons made to the PI controller and PINN-based controller. The results show that both the NO-approximated and PINO-approximated mappings achieve satisfactory accuracy and provide a significant computational speedup of 298 times compared to the backstepping controller, with only a 1\% loss in accuracy.
}
{
To evaluate the robustness of the proposed NO-based methods, various traffic conditions are tested. The NO-based methods demonstrate strong performance in stabilizing different traffic scenarios. Additionally, using the NGSIM data to calibrate the fundamental diagram and applying it to real traffic systems shows that the NO method has significant potential for practical applications in freeway traffic control.
}

{
Future work should include extending the operator learning framework to explore the performance of NO-based methods in network traffic. Additionally, applying these methods to other traffic-related problems, such as traffic assignment and vehicle routing, would be a valuable area of investigation.
}
% Code should be open to public

\appendix
\section*{Appendix}
\counterwithin*{equation}{section}
\renewcommand\theequation{\thesection \arabic{equation}}
\section{Proof of Theorem \ref{es-NO-k}}\label{append}

Taking time derivative along the trajectories of the system, plugging the system dynamics, we have
\begin{align}
     \Dot{V}_k(t) &= -\nu V_k(t) + (r^2-a) \Hat{\beta}^2(0,t) -\mathrm{e}^{-\frac{\nu}{\lambda_1}L}\Hat{\alpha}^2(L,t) \nonumber\\
     &+ \int_0^L 2a\frac{\mathrm{e}^{-\frac{\nu}{\lambda_2}x}}{\lambda_2}\Hat{\beta}(x,t) \left( \lambda_2 (\Tilde{K}^w(x,0) + \Tilde{K}^v(x,0))\Tilde{v}(0,t)+(\lambda_1+\lambda_2) \Tilde{K}^w(x,x)\Tilde{w}(x,t)\right.\nonumber\\
    & + \int_0^x (\lambda_2 \Tilde{K}^w_x(x,\xi) + \lambda_1 \Tilde{K}^w_{\xi}(x,\xi))\Tilde{w}(\xi,t) d\xi \nonumber\\
    &+\int_0^x (\lambda_2 \Tilde{K}^v_x(x,\xi) + \lambda_2  \Tilde{K}^v_{\xi}(x,\xi))\Tilde{v}(\xi,t) d\xi \left.\right)dx,
\end{align}
For the integral term, we take the norm and use the Young inequality and Cauchy-Schwartz inequality. Then combining the equivalent norm of the Lyapunov candidate, we get:
\begin{align}
    \int_0^L \norm{2a\frac{\mathrm{e}^{-\frac{\nu}{\lambda_2}x}}{\lambda_2}\Hat{\beta}(x,t) \lambda_2 (\Tilde{K}^w(x,0) + \Tilde{K}^v(x,0))\Tilde{v}(0,t)}dx &\leq \int_0^L \norm{2a\epsilon{\mathrm{e}^{-\frac{\nu}{\lambda_2}x}}(\Hat{\beta}^2(x,t) + \Tilde{v}^2(0,t))}dx\nonumber\\
    &\leq \frac{2a\epsilon}{m_1}V_k(t) + 2aL\epsilon \Hat{\beta}^2(0,t).
\end{align}
Using the same method, we can easily get the results for the other terms of the Lyapunov candidate. For the second term, we have:
\begin{align}
    \int_0^L \norm{2a\frac{\mathrm{e}^{-\frac{\nu}{\lambda_2}x}}{\lambda_2}\Hat{\beta}(x,t)(\lambda_1+\lambda_2) \Tilde{K}^w(x,x)\Tilde{w}(x,t)}dx &\leq \int_0^L\norm{2a\epsilon\frac{\lambda_1 + \lambda_2}{\lambda_2}{\mathrm{e}^{-\frac{\nu}{\lambda_2}x}}(\Hat{\beta}^2(x,t) + \Tilde{w}^2(x,t))}dx\nonumber\\
    &\leq \frac{2a\epsilon(\lambda_1+\lambda_2)}{m_1\lambda_2}(1+\frac{1}{k_1})V_k(t).
\end{align}
For the third term, we can get the following results:
\begin{align}
    &\int_0^L \norm{2a\frac{\mathrm{e}^{-\frac{\nu}{\lambda_2}x}}{\lambda_2}\Hat{\beta}(x,t)+ \int_0^x (\lambda_2 \Tilde{K}^w_x(x,\xi) + \lambda_1 \Tilde{K}^w_{\xi}(x,\xi))\Tilde{w}(\xi,t) d\xi}dx \nonumber\\
    &\leq \int_0^L \norm{2a\epsilon\frac{\lambda_1 + \lambda_2}{\lambda_2}{\mathrm{e}^{-\frac{\nu}{\lambda_2}x}}L(\Hat{\beta}^2(x,t) + \Tilde{w}^2(x,t))}dx \leq \frac{2a\epsilon(\lambda_1+\lambda_2)}{m_1\lambda_2}(1+\frac{1}{k_1})V_k(t).
\end{align}
For the last term, we have:
\begin{align}
    \int_0^L \norm{2a\frac{\mathrm{e}^{-\frac{\nu}{\lambda_2}x}}{\lambda_2}\Hat{\beta}(x,t)+ \int_0^x (\lambda_2 \Tilde{K}^v_x(x,\xi) + \lambda_2  \Tilde{K}^v_{\xi}(x,\xi))\Tilde{v}(\xi,t) d\xi}dx &\leq  \int_0^L \norm{2a\epsilon{\mathrm{e}^{-\frac{\nu}{\lambda_2}x}}L(\Hat{\beta}^2(x,t) + \Tilde{v}^2(x,t))}dx \nonumber\\
    &\leq \frac{2a\epsilon L}{m_1}(1+\frac{1}{k_1})V_k(t).
\end{align}
Using the bound of the four terms, we can derive the estimation of the Lyapunov candidate in \eqref{lyapunov_bound}.

\section{Proof the practical stability of the traffic system}\label{practical}
Based on the properties of the neural operators in Lemma~\ref{Deeptheo}, it can be found that the NO-approximated mapping $\mathcal{\Hat{H}}(\lambda_2)$ satisfies the following lemma:
\begin{lemma}
    For any $\epsilon > 0$,  there exists a neural operator $\mathcal{\Hat{H}}(\lambda_2)$ that can approximate the control law mapping in the spatial-temporal domain $(x,t) \in [0,L]\times \mathbb{R}^+$:
    \begin{align}\label{no-control}
        \max_{\lambda_2\in \mathcal{U}}\left| \mathcal{H}\left(\lambda_2\right)(L) - \mathcal{\Hat{H}}(\lambda_2)(L) \right| < \epsilon.
    \end{align}
\end{lemma}
\begin{proof}
    We know from Theorem \ref{Deeptheo} that the neural operator using DeepONet can approximate operator mapping within accuracy $\epsilon$. We can apply Theorem \ref{Deeptheo} again to approximate the mapping from the characteristic speed $\lambda_2$ to $U(t)$. Thus, the lemma is proven.
\end{proof}
Applying the NO-approximated control law to the system \eqref{bs-q}-\eqref{bs-bc_l}, we get the target system as:
\begin{align}
    \partial_t \Check{\alpha}(x,t) + \lambda_1 \partial_x\Check{\alpha}(x,t) &= 0, \label{nocontrol-q}\\
    \partial_t \Check{\beta}(x,t) - \lambda_2 \partial_x \Check{\beta}(x,t) &= 0,\\
    \Check{\alpha}(0,t) &= -r\Check{\beta}(0,t),\\
    \Check{\beta}(L,t) &= \mathcal{H}\left(\lambda_2\right)(L, t)  - \mathcal{\Hat{H}}(\lambda_2)(L, t). \label{nocontrol-bc_l}
\end{align}
Compared with the target system in \citep{yu_traffic_2019}, the system \eqref{nocontrol-q}-\eqref{nocontrol-bc_l} is not strictly exponentially stable due to the approximation error of NO-approximated control law. 

{
To prove the stability of the traffic system under the NO-approximated control law, we first give the definition of the local practical exponential stability of the traffic PDE system,
\begin{definition}[Local practical exponential stability~\citep{bhan2023neural,teel1999semi}]
    For the ARZ traffic system whose control law is approximated by NO, the traffic system is said to be locally practically exponentially stable if the state satisfy
    \begin{align}
        \norm{(\Bar{\rho}(x,t),\Bar{v}(x,t))}^2_{L^2} \leq  \mathrm{e}^{-\mu t} \norm{(\Bar{\rho}(x,0),\Bar{v}(x,0))}^2_{L^2} + \kappa(\epsilon),
    \end{align}
    where $\mu > 0$ and $\kappa : \mathbb{R}^+ \rightarrow \mathbb{R}^+$ is of class $\mathcal{K}$ function with strictly increasing  and $\kappa(0)=0$ properties.
\end{definition}
}
Thus, we have the following theorem for the NO-approximated control law:
\begin{theorem}\label{thm4}
    The system \eqref{origin1}-\eqref{origin2} with boundary conditions \eqref{bc_q}-\eqref{bc_v} is locally practically exponentially stable under the NO-approximated control law $\hat{\mathcal{H}}(\lambda_2)(L,t)$ with initial conditions $\Bar{\rho}(x,0),\Bar{v}(x,0)$, such that 
    \begin{align}
        ||(\Bar{\rho}(x,t),\Bar{v}(x,t))||^2_{L^2} \leq c_2 \mathrm{e}^{-\nu t} ||(\Bar{\rho}(x,0),\Bar{v}(x,0))||^2_{L^2} + \frac{a}{m_4k_4} \mathrm{e}^{-\frac{\nu}{\lambda_2}L}\epsilon^2.
    \end{align}
    where $c_2 = \frac{m_3n_4k_3}{m_4 n_3k_4}$, $m_3>0$, $m_4>0$, $n_3>0$, $n_4>0$, $k_3>0$, $k_4>0$. The control law in \eqref{controlmu2c} is approximated by the neural operator $\hat{\mathcal{H}}(\lambda_2)(L,t)$ with accuracy $\epsilon$ in \eqref{no-control}.
\end{theorem}
\begin{proof}
For the NO-approximated control law, using the Lyapunov candidate again to analyze the stability of the target system \eqref{nocontrol-q}-\eqref{nocontrol-bc_l}.
\begin{align}
    V_U(t) = \int_0^L \frac{\mathrm{e}^{-\frac{\nu}{\lambda_1}x}}{\lambda_1}\Check{\alpha}^2(x,t) + a\frac{\mathrm{e}^{-\frac{\nu}{\lambda_2}x}}{\lambda_2} \Check{\beta}^2(x,t) dx,
\end{align}
where the coefficients $\nu$ and $a$ are the same as before. The states of the target system with NO-approximated control law $(\Check{\alpha},\Check{\beta})$ still have equivalent $L_2$ norm with the system $(\Tilde{w},\Tilde{v})$. 
\begin{align}
    k_3 \norm{(\Tilde{w}(x,t),\Tilde{v}(x,t))}_{L_2}^2 \leq \norm{(\Check{\alpha}(x,t),\Check{\beta}(x,t))}_{L_2}^2 \leq k_4 \norm{(\Tilde{w}(x,t),\Tilde{v}(x,t))}_{L_2}^2,
\end{align}
where $k_3>0$ and $k_4>0$. The same is true for the Lyapunov functional, it has the following equivalent norm with the NO-approximated target system. There exist $m_3>0,m_4>0$, 
\begin{align}\label{eqnormVu}
    m_3\norm{(\Check{\alpha}(x,t),\Check{\beta}(x,t))}^2_{L^2} \leq V_U(t) \leq m_4\norm{(\Check{\alpha}(x,t),\Check{\beta}(x,t))}^2_{L^2}.
\end{align}
Taking time derivative along the trajectories, putting it into the system dynamics, and integrating by parts, we have:
\begin{align}
    \Dot{V}_U(t) = -\nu {V}_U(t) + (r^2 - a)\Check{\beta}^2(0,t) + a\mathrm{e}^{-\frac{\nu}{\lambda_2}L}\Check{\beta}^2(L,t) -\mathrm{e}^{-\frac{\nu}{\lambda_1}L}\Check{\alpha}^2(L,t).
\end{align}
The term $a\mathrm{e}^{-\frac{\nu}{\lambda_2}L}\Check{\beta}^2(L,t)$ is equal to 0 in the ideal situation when the NO-approximated mapping achieves $100\%$ accuracy of approximation which means that $\mathcal{H}\left(\lambda_2\right)(L) - \mathcal{\Hat{H}}(\lambda_2)(L) = 0$ and we can easily get the exponential stability for the system. Here the mapping has the error $\epsilon$. We take the $r^2 -a \leq 0$, and we get
\begin{align}
   V_U(t) \leq V_U(0) \mathrm{e}^{-\nu t} + a \mathrm{e}^{-\frac{\nu}{\lambda_2}L}\sup_{0 \leq \varsigma \leq t} (\mathcal{H}\left(\lambda_2\right)(L) - \mathcal{\Hat{H}}(\lambda_2)(L))^2(L,\varsigma).
\end{align}
Using \eqref{eqnormVu}, we have
\begin{align}
    ||(\Tilde{w}(x,t),\Tilde{v}(x,t))||^2_{L^2} \leq \frac{m_3k_3}{m_4 k_4}\mathrm{e}^{-\nu t} ||(\Tilde{w}(x,0),\Tilde{v}(x,0))||^2_{L^2} + \frac{a}{m_4k_4} \mathrm{e}^{-\frac{\nu}{\lambda_2}L}\epsilon^2.
\end{align}
Thus, we have proved that the system \eqref{bs-q}-\eqref{bs-bc_l} is locally practically exponentially stable. Using the equivalent norm 
\begin{align}
    n_3\norm{(\Bar{\rho}(x,t),\Bar{v}(x,t))}_{L_2}^2 \leq \norm{(\Tilde{w}(x,t),\Tilde{v}(x,t))}^2_{L^2} \leq n_4\norm{(\Bar{\rho}(x,t),\Bar{v}(x,t))}_{L_2}^2,
\end{align}
thus, we get
\begin{align}
    ||(\Bar{\rho}(x,t),\Bar{v}(x,t))||^2_{L^2} \leq c_2 \mathrm{e}^{-\nu t} ||(\Bar{\rho}(x,0),\Bar{v}(x,0))||^2_{L^2} + \frac{a}{m_4k_4} \mathrm{e}^{-\frac{\nu}{\lambda_2}L}\epsilon^2,
\end{align}
where $c_2=\frac{m_3n_4k_3}{m_4 n_3k_4}$. Therefore, the original system \eqref{origin1}-\eqref{origin2} with boundary conditions \eqref{bc_q} and \eqref{bc_v} is locally practically exponentially stable. This finish the proof of Theorem \ref{thm4}.  
\end{proof}

% \section{Acknowledgement}
% This research is partially supported by 
%% Loading bibliography style file
% \bibliographystyle{model1-num-names}
\bibliographystyle{cas-model2-names}

% Loading bibliography database
\bibliography{ref}

\end{document}